\documentclass[journal]{IEEEtran}

\usepackage{latexsym,amssymb,exscale,graphicx,amsfonts,amsmath}
\usepackage{multirow,multicol}
\usepackage{epstopdf}
\usepackage{marvosym}
\usepackage{diagbox}
\usepackage{mathrsfs}
\usepackage{makecell}
\usepackage[normalem]{ulem}
\usepackage{epsfig} 
\usepackage{setspace}
\usepackage[table]{xcolor}
\usepackage{xcolor,colortbl}
\usepackage{lipsum}
\usepackage{microtype}
\usepackage{stfloats}
\usepackage[noadjust]{cite}
\usepackage{lettrine}
\usepackage{mathtools, cuted}
\usepackage{amsmath}
\usepackage{blkarray}
\usepackage[lined,ruled,linesnumbered,boxed,noend]{algorithm2e}
\usepackage{mathtools}
\usepackage[english]{babel}
\usepackage[utf8]{inputenc}
\usepackage{bbm}
\usepackage{bm}
\usepackage{amsthm, amssymb}
\usepackage{optidef}
\usepackage{cases}
\usepackage{url}
 \usepackage[hidelinks]{hyperref}
\newtheorem{theorem}{Theorem}[section] 

\newtheorem{lemma}{Lemma}[section]

\usepackage{comment}

\usepackage{amsbsy}

\def\bfB{\boldsymbol{B}}

\def\ddefloop#1{\ifx\ddefloop#1\else\ddef{#1}\expandafter\ddefloop\fi}
\def\ddef#1{\expandafter\def\csname bb#1\endcsname{\ensuremath{\mathbb{#1}}}}
\ddefloop ABCDEFGHIJKLMNOPQRSTUVWXYZ\ddefloop
\def\ddef#1{\expandafter\def\csname c#1\endcsname{\ensuremath{\mathcal{#1}}}}
\ddefloop ABCDEFGHIJKLMNOPQRSTUVWXYZ\ddefloop
\def\ddef#1{\expandafter\def\csname v#1\endcsname{\ensuremath{\boldsymbol{#1}}}}
\ddefloop ABCDEFGHIJKLMNOPQRSTUVWXYZabcdefghijklmnopqrstuvwxyz\ddefloop
\def\ddef#1{\expandafter\def\csname v#1\endcsname{\ensuremath{\boldsymbol{\csname #1\endcsname}}}}
\ddefloop {alpha}{beta}{gamma}{delta}{epsilon}{varepsilon}{zeta}{eta}{theta}{vartheta}{iota}{kappa}{lambda}{mu}{nu}{xi}{pi}{varpi}{rho}{varrho}{sigma}{varsigma}{tau}{upsilon}{phi}{varphi}{chi}{psi}{omega}{Gamma}{Delta}{Theta}{Lambda}{Xi}{Pi}{Sigma}{varSigma}{Upsilon}{Phi}{Psi}{Omega}{ell}\ddefloop
\newcommand{\quotes}[1]{``#1''}
\newcommand{\lrbrackets}[1]{\left(#1\right)}
\newcommand{\ceil}[1]{\lceil#1\rceil}

\newcommand\lp[1]{\Arrowvert #1 \Arrowvert} 

\makeatletter
\newcommand{\nosemic}{\renewcommand{\@endalgocfline}{\relax}}
\newcommand{\dosemic}{\renewcommand{\@endalgocfline}{\algocf@endline}}
\let\oldnl\nl
\newcommand{\nonl}{\renewcommand{\nl}{\let\nl\oldnl}}
\makeatother

\newcommand{\ting}[1]{\textcolor{red}{Ting: #1}}
\newcommand{\yudi}[1]{\textcolor{cyan}{Yudi: #1}}
\newcommand{\rev}[1]{\textcolor{black}{#1}}

\newif\ifisappendix
\isappendixtrue

\IEEEoverridecommandlockouts

\begin{document}
	
\title{Preventing Outages under Coordinated Cyber-Physical Attack with Secured PMUs}
\author{Yudi Huang, \textit{Student Member, IEEE}, Ting He, \textit{Senior Member, IEEE}, Nilanjan Ray Chaudhuri, \\\textit{Senior Member, IEEE}, and Thomas La Porta \textit{Fellow, IEEE}
\thanks{The authors are with the School of Electrical Engineering and Computer Science, Pennsylvania State University, University Park, PA 16802, USA
(e-mail: \{yxh5389, tzh58, nuc88, tfl12\}@psu.edu). }
\thanks{A preliminary version of this work was presented at SmartGridComm'21~\cite{yudi21SmartGridComm}.}
\thanks{ This work was supported by the National Science Foundation under award ECCS-1836827. }
}

\maketitle

\begin{abstract}
Due to the potentially severe consequences of coordinated cyber-physical attacks (CCPA), the design of defenses has gained significant attention. A popular approach is to eliminate the existence of attacks by either securing existing sensors or deploying secured PMUs. In this work, we improve this approach by lowering the defense target from \emph{eliminating attacks} to \emph{preventing outages} and reducing the required number of PMUs. To this end, we formulate the problem of \emph{PMU Placement for Outage Prevention (PPOP)} {under DC power flow model} as a tri-level non-linear optimization problem and transform it into a bi-level mixed-integer linear programming (MILP) problem. Then, we propose an alternating optimization framework to solve PPOP by iteratively adding constraints, for which we develop two constraint generation algorithms. In addition, for large-scale grids, we propose a polynomial-time heuristic algorithm to obtain suboptimal solutions. {Next, we extend our solution to achieve the defense goal under AC power flow model.} Finally, we evaluate our algorithm on IEEE 30-bus, 57-bus, 118-bus, and 300-bus systems, which demonstrates the potential of the proposed approach in greatly reducing the required number of PMUs.
\end{abstract}
	

\section{Introduction}\label{sec: intro}

\emph{Coordinated cyber-physical attacks (CCPA)}~\cite{deng2017ccpa} have gained a great deal of attention due to the stealthiness of such attacks and the potential for severe damage on to the smart grid. The power of CCPA is that its physical component damages the grid while its cyber component masks such damage from the control center (CC) to prolong outages and potentially enable cascades. For instance, in the Ukrainian power grid attack \cite{Fairley16Spectrum}, attackers remotely switched off substations (damaging the physical system) while disrupting the control through telephonic floods and KillDisk server wiping (damaging the cyber system). \looseness=-1 



Defenses against CCPA can be broadly categorized into \emph{detection} and \emph{prevention}. Attack detection mechanisms aim at detecting attacks that are otherwise undetectable using traditional bad data detection (BDD) by exploiting knowledge unknown to the attacker \cite{chaojun2015detecting}. However, the knowledge gap between the attacker and the defender may disappear due to more advanced attacks, and relying on detection alone risks severe consequences in case of misses. Therefore, in this work, we focus on preventing attacks using secured sensors. 

We consider a powerful attacker with full knowledge of the pre-attack state of the grid and the locations of secured PMUs.  The attacker launches an optimized CCPA where the physical attack disconnects a limited number of lines and the cyber attack falsifies the breaker status and the measurements from unsecured sensors to mask the physical attack while misleading security constrained economic dispatch (SCED) at the CC. {Such attacks can result in severe cascading failures. For example, under the setting in Section~\ref{sec: evaluations}, CCPA in absence of secured PMUs can cause initial overload-induced tripping at $2$, $1$, and $2$ lines in IEEE 30-bus, 57-bus, and 118-bus systems, respectively. Moreover, the re-distribution of power flows on the initially tripped lines may cause cascading outages. Take IEEE 118-bus system as an example. There is an attack that physically disconnects line $144$ and manipulates the measurements to cause overload-induced tripping at line $109$. These initial outages will trigger a cascade that eventually results in outages at $82$ lines. This observation highlights the importance of defending against such attacks.} \looseness=-1

While attack prevention traditionally aims at eliminating undetectable attacks by deploying secured PMUs to achieve full observability~\cite{sou2019protection}, this approach can require a large number of PMUs. {Little is known about how to achieve a good tradeoff between the efficacy of protection and the cost of PMU placement during the deployment process 
before full observability is achieved. In addition, the operators may be only interested in using secured PMUs to prevent severe consequences, while leaving the defense of less severe attacks to other mechanisms~\cite{dibaji2019systems}. To fill this gap, we lower the goal of PMU placement to \emph{preventing undetectable attacks from causing outages}.}
Specifically, we want to deploy the minimum number of secured PMUs such that the attacker will not be able to cause overload-induced line tripping due to overcurrent protection devices.  
The key novelty of our approach is that
we allow undetectable attacks to exist but prevent them from causing any outages, hence potentially requiring fewer secured PMUs. {For instance, we can prevent overload-induced tripping using $71\%$ fewer secured PMUs compared to the requirement of full observability in IEEE 118-bus system.} \looseness=-1


\subsection{Related Work}\label{subsec:Related Work}

\textbf{Attacks:}
False data injection (FDI) \cite{ozay2013sparse, alexopoulos2020complementarity} is widely adopted to launch cyber attacks in CCPA to bypass the traditional BDD \cite{deng2017ccpa}. A typical form of FDI is load redistribution attack~\cite{yuan2011modeling}, which together with physical attacks \cite{deng2017ccpa, lakshminarayana2021moving, che2018false} that alter grid topology, aims to mislead SCED by injecting false data for economic loss or severe physical consequences such as sequential outages~\cite{che2018false}. {Bi-level optimization is widely adopted for analyzing the impact of CCPA on state deviation \cite{liu2016masking} or line flow changes \cite{li2015bilevel}.}
In this work, we extend them into a stronger attacker that jointly optimizes the location of physical attacks and the attack target. \rev{Besides misleading SCED, similar physical consequences can also be achieved by attacking the commands issued by the control center \cite{lin2016runtime, garcia2017hey}, which is not the focus of this work.}\looseness=-1

\textbf{Defenses:}
\rev{Defending against CCPA requires a systematic mechanism \cite{dibaji2019systems}, which can be decomposed into three modules: \emph{prevention} that postpones the onset of attacks \cite{lakshminarayana2021moving}, \emph{detection} that identifies the attack before it starts affecting the system \cite{deng2015defending, liu2016optimal, tian2019multilevel, kim2011strategic, hao2015sparse, yang2017optimal, sou2019protection, lin2016runtime}, and \emph{resilience} which limits the impact of the attacks that successfully bypass the detection \cite{xiang2017game, tian2019multilevel, wu2016efficient, yao2007trilevel}. Our focus is on an intermediate stage of PMU deployment where not enough PMUs are installed to achieve perfect {detection} of all FDI attacks.} \looseness=-1

\rev{To eliminate the existence of FDI by detection, different strategies have been studied, such as directly protecting meters \cite{deng2015defending, liu2016optimal, bobba2010detecting, tian2019multilevel, kim2011strategic, hao2015sparse} or deploying secured PMUs \cite{yang2017optimal, sou2019protection}.
Due to the connection between observability of the grid and FDI \cite{liu2016optimal}, solutions on achieving full observability through PMUs \cite{chakrabarti2008placement}
can also be leveraged to defend against FDI.}
Unlike the aforementioned works, our work only aims to prevent attacks from causing outages, which can significantly reduce the required number of secured PMUs \rev{while maintaining the system \emph{resilience}}. \looseness=-1

Tri-level optimization is widely used for modeling interactions among the defender, the attacker and the operator in smart grid. To name a few, a tri-level model is proposed in \cite{wu2016efficient} to find the optimal set of lines to protect from physical attacks to minimize load shedding. {In \cite{xiang2017game, yao2007trilevel, tian2019multilevel}, the  measurements to protect were chosen by solving a budget-constrained optimization problem, which was also adopted in \cite{yuan2016robust} for distribution networks. However, existing works are limited in the following aspects. From the formulation perspective, their solution may become sub-optimal if the cost vector in SCED changes due to the dependence of their methods on the KKT conditions of linear programming. Such dependence also limits the extension of their formulation to the AC power flow model. From the computational perspective, the method in \cite{tian2019multilevel} solves a MIP for each possible physical attack and thus is not scalable to multi-line physical attacks. The method in \cite{wu2016efficient} introduces bilinear terms, which leads to a high computational cost. To overcome such limitations, we will develop a formulation for CCPA that can (i) model multi-line physical attacks without bilinear terms, and (ii) be extended to the AC power flow model. Moreover, the PMU placement obtained from our solution can prevent overloading-induced line tripping regardless of the cost vector in SCED. Furthermore, securing PMU measurements instead of (legacy) measurements for individual nodes/lines has the advantage that it aligns with the ongoing trend of deploying PMU-based power grid monitoring systems.} \looseness=-1
%


{\textbf{Power flow models:} Due to the nonlinear and nonconvex nature of AC power flow equations, it is a common practice \cite{liang2015vulnerability} to develop FDI/CCPA or its countermeasure under the DC power flow model and validate the solutions under the AC power flow model. Although much efforts  \cite{jin2018power, chung2018local, chu2021n} have been devoted into directly formulating FDI under the AC model, most of them targeted at causing erroneous state estimation, with very limited results on load redistribution attack aiming at causing outages. 
The works \cite{jin2018power} formulated FDI under the AC model  through convex relaxation, but did not accurately model the impact of FDI on SCED. In \cite{chu2020vulnerability, liang2015vulnerability, bobba2010detecting}, the design of FDI was based on the DC model, although the feasibility of the attack was tested under the AC model. In \cite{chung2018local, chu2021n}, a formulation based on convex relaxation was proposed to model load redistribution attack under the AC model. They adopted DC-based \emph{line outage distribution factors (LODF)} to infer the impact of attacks on SCED, which leads to the use of active power flows as the criterion to determine overloading. This is inaccurate as the true criterion should be the magnitude of current. To the best of our knowledge, it remains an open problem to compute the optimal load redistribution attack under the AC power flow model. Our approach is to circumvent this problem by (i) first finding a PMU placement to prevent load redistribution attack from causing outages under the DC model, (ii) then 
developing a method to test the feasibility of the found PMU placement under the AC model based on a recently developed approximation of AC power flow equations~\cite{yang2018general}, and (iii) finally refining the PMU placement to prevent outages under the AC model. } \looseness=-1


\subsection{Summary of Contributions}
We summarize our contributions as follows:
\begin{enumerate}
    \item Instead of eliminating the existence of FDI, we investigate the optimal secured \emph{PMU Placement for Outage Prevention (PPOP)} problem to defend against CCPA, where we formulate a strong attacker that jointly optimizes physical attack locations and target lines. The proposed approach can potentially require fewer PMUs than approaches that eliminate FDI.\looseness=-1
    \item We propose an alternating optimization algorithm to solve PPOP by generating additional constraints from each infeasible PMU placement. Specifically, we demonstrate how to generate \quotes{No-Good} constraints and \quotes{Attack-Denial} constraints to solve PPOP optimally.
    \item We develop a heuristic algorithm for PPOP to produce a possibly suboptimal solution. The complexity of the proposed heuristic is polynomial in the grid size, which makes it scalable to large networks.
    \item {We develop an algorithm to test whether a given PMU placement can achieve our defense goal under the AC power flow model. In addition, we propose a heuristic to augment the given PMU placement to pass the test.}
    \item We systematically evaluate the proposed solution on IEEE 30-bus, IEEE 57-bus, IEEE 118-bus, and IEEE 300-bus systems. The results demonstrate that the proposed solution can {substantially reduce the number of required PMUs while preventing CCPA from causing outages, even with the AC-based augmentation.}
\end{enumerate}

\textbf{Roadmap: }We formulate the PPOP problem {under the DC model} in Section~\ref{sec: problem_formulation} and present both optimal algorithms and heuristics to solve PPOP in Section~\ref{sec: solve_PPOP_optimally}. {We then show how the DC-based solution can be refined to work under the AC model in Section~\ref{sec:Extension to AC}.} We evaluate the performance of PPOP in Section~\ref{sec: evaluations} and conclude the paper in Section~\ref{sec: conclusion}. Additional contents and proofs are given in the appendices. 
\looseness=-1


\section{Problem formulation}\label{sec: problem_formulation}

\textbf{Notations: } For a matrix $\vA$, we denote by $\va_i$ its $i$-th column and $\vA_k$ its $k$-th row. We slightly abuse the notation $|\cdot|$ in that $|A|$ indicates the cardinality if $A$ is a set and the element-wise absolute value if $A$ is a vector or matrix. Logical expression $\leftrightarrow$ indicates the \quotes{if and only if} logic, while $\rightarrow$ denotes the \quotes{if then} logic. When the operators $\ge,\le, =$ are applied to two vectors, they indicate element-wise operations. Let $\va \in \mathbb{R}^{n_a},\vb\in \mathbb{R}^{n_b}$ be two vectors, then $\va \oplus \vb \in \mathbb{R}^{n_a+n_b}$ indicates the vertical concatenation of $\va$ and $\vb$. Let $\lceil \va\rceil$ denote the element-wise ceiling. If $n_a=n_b=n$, then $\va\odot \vb := (a_i b_i)_{i=1}^{n}$ denotes the Hadamard product, i.e., the element-wise product. We use $\vLambda_{(\cdot)}\in \{0,1\}^{m\times n}$ with one nonzero element in each row to select entries from a vector such that $\vLambda_{(\cdot)} \vx$ is a subvector of $\vx$. \looseness=-1

\subsection{Power Grid Modeling}\label{subsec:Power Grid Modeling}
We model the power grid as a connected undirected graph ${\mathcal{N}} = (V, E)$, where $E$ denotes the set of lines (lines) and $V$ the set of nodes (buses).
{Majority of our results will be based on the DC power flow model, which is an approximation widely adopted for studying security issues in power grids \cite{sou2019protection, yuan2011modeling, deng2017ccpa, lakshminarayana2021moving, che2018false, liu2016masking, li2015bilevel, tian2019multilevel};
extension to the AC power flow model is deferred to Section~\ref{sec:Extension to AC}.} Under this approximation, each line $e = (s,t)$ is characterized by reactance $r_e = r_{st} = r_{ts}$.
The grid topology can be represented by the \emph{admittance matrix} $\bfB:=(B_{uv})_{u,v\in V} \in \mathbb{R}^{|V|\times|V|}$, defined as
\begin{align}
    B_{uv} &=\left\{\begin{array}{ll}
    0 & \mbox{if }u\neq v, (u,v)\not\in E,\\
    -1/r_{uv} &  \mbox{if } u\neq v, (u,v)\in E,\\
    -\sum_{w\in V\setminus \{u\}}B_{uw} &\mbox{if }u=v.
    \end{array}\right.
\end{align}
Besides $\vB$, the grid topology can also be described by \emph{incidence matrix} $\vD \in \{-1,0,1\}^{|V|\times|E|}$, which is defined as follows:
\begin{align}
    D_{ij} &= \left\{\begin{array}{ll}
    1 & \mbox{if line }e_j\mbox{ comes out of node }v_i,\\
    -1 & \mbox{if line }e_j\mbox{ goes into node }v_i,\\
    0 & \mbox{otherwise,}
    \end{array}\right.
\end{align}
where the orientation of each line is assigned arbitrarily. By defining $\vGamma \in \mathbb{R}^{|E| \times |E|}$ as a diagonal matrix with diagonal entries $\Gamma_e = \frac{1}{r_e}$ ($e\in E$), we have $\vB = \vD \vGamma \vD^T$ and $\vf = \vGamma \vD^T \vtheta \in \mathbb{R}^{|E| }$ where $\vf$ denotes the line flows. By defining network states as phase angles $\vtheta := (\theta_u)_{u\in V}$ and active powers as $\vp = (p_u)_{u\in V}$, the relationship between $\vp,\vtheta$ and $\vf$ is given as \looseness=-1
\begin{align}\label{eq:B theta = p}
    \vp =  \vB \vtheta =  \vD \vf,
\end{align}
The CC will periodically conduct state estimation, whose results will be used for SCED to re-plan the power generation \cite{yuan2011modeling, che2018false}. Formally, let $\vz = [\vz_N^T,\vz_L^T]^T \in \mathbb{R}^{m}$ denote the unsecured meter measurements, where $\vz_N\in \mathbb{R}^{m_N}$ denotes the power injection measurements over (a subset of) nodes and $\vz_L\in \mathbb{R}^{m_L}$ denotes the power flow measurements over (a subset of) lines. Let $\vLambda_N$ and $\vLambda_p$ be two row selection matrices such that $\vz_N = \vLambda_N \vz  = \vLambda_p \vp$. 
Similarly, we define row selection matrices $\vLambda_L$ and $\vLambda_f$ such that $\vz_L = \vLambda_L\vz = \vLambda_f\vf$. Then, we have
\begin{align}\label{eq: z_H_theta}
\vz = \vH \vtheta+\vepsilon ~\mbox{ for } \vH := \left[\begin{array}{c}
     \vLambda_p\vB  \\
     \vLambda_f\vGamma\vD^T
\end{array}\right],
\end{align}
where $\vH$ is the measurement matrix based on the meter locations and the reported breaker status, and $\vepsilon$ is the measurement noise. 
{In the rest of the paper, we assume that the measurements satisfy the conditions of \cite[Theorem 5]{krumpholz1980power} such that $\vH$ has full column rank to support unique recovery of $\vtheta$ from \eqref{eq: z_H_theta} (before attack).} If $\bar{\vtheta}$ is the estimated phase angle from $\vz$ and $\vH$, then BDD will raise alarm if $\lp{\vz - \vH \bar{\vtheta}}$ is greater than a predefined threshold. \looseness = -1


Given $\vp_0 := \vB \bar{\vtheta}$, the CC will conduct SCED to calculate new generation to meet the demand with minimal cost. Specifically, let $\vLambda_g\in \{0,1\}^{|V_g|\times |V|}$, $\vLambda_d\in \{0,1\}^{|V_d|\times |V|}$ be row selection matrices for generator/load buses in $\vp$, where $V_d$ and $V_g$ denote the sets of load buses and generator buses, respectively. Denote $\hat{\vtheta}$ as the decision variable where $\vB \hat{\vtheta}$ represents the new power injection after SCED, and $\vphi\in \mathbb{R}^{|V_g|}$ as the cost vector for power generation. Then, SCED can be formulated  as follows \cite{che2018false}: \looseness=-1
\begin{subequations}\label{eq:reformulate_OPF}
\allowdisplaybreaks
\begin{alignat}{2}
\psi_s(\vp_0, \vD)=\arg&\min_{\hat{\vtheta}} \quad  \vphi^T(\vLambda_g\vB\hat{\vtheta}) \\
\mbox{s.t.} \quad
&\vLambda_d\vB\hat{\vtheta} = \vLambda_d\vp_0, \label{eq: cc_load_meet}  \\
&\vGamma \vD^T\hat{\vtheta} \in [-\vf_{\tiny\text{max}},\vf_{\tiny\text{max}}], \label{eq: cc_fmaxp} \\
&\vLambda_g\vB\hat{\vtheta} \in [\vp_{g,min}, \vp_{g,max}], \label{eq: cc_pgp}
\end{alignat}
\end{subequations}
where $\vf_{\tiny\text{max}}\in \mathbb{R}^{|E|}$ indicates the normal line flow limits, $\vp_{g,min}$ and $\vp_{g,max}$ denote lower/upper bounds on generation, and \eqref{eq: cc_load_meet} indicates that demands on all load buses are satisfied.

\subsection{Modeling Coordinated Cyber-Physical Attack (CCPA)} \label{sec: model_ccpa}

In this section, we formulate the attack model according to  a load redistribution attack~\cite{yuan2011modeling} that aims at causing the maximum outages, so that a defense against this attack can prevent outage under any attack under the same constraints. In the sequel, ``ground truth'' means the estimated value based on unmanipulated measurements, which may contain noise. 

For ease of presentation, we summarize the timeline of the entire attack process, as shown in Fig~\ref{fig: timeline_attack}. Specifically,
\begin{itemize}
    \item At $t_0$, the attacker estimates $\vtheta_0$ and $\vp_0:=\tilde{\vB}\vtheta_0$ by eavesdropping on $\vz_0$ and $\tilde{\vH}$.
    \item At $t_1$, CCPA is deployed to change the ground-truth from $\vz_0, \tilde{\vH}, \vtheta_0$ to $\vz_1, \vH$ and $\vtheta_1$, respectively.
    \item At $t_2$, the CC receives falsified information, i.e., $\tilde{\vH}$ and $\tilde{\vz}_2$, which leads to $\tilde{\vtheta}_2$. Then the CC will deploy a new dispatch {of power generation} as $\tilde{\vp}_3:=\tilde{\vB}\tilde{\vtheta}_3$, where $\tilde{\vtheta}_3$ denotes the associated predicted phase angles.
    \item At $t_3$, the new dispatch takes effect and reaches steady state, with the true phase angles $\vtheta_3$ and power flows $\vf_3$.
\end{itemize}
Key notations at different time instances are summarized in Table~\ref{tab: notation_timeline}, where ``---'' means that the information is not available to the CC at the given time instance.

\begin{figure}[ht]
\vspace{-.5em}
\centering
\includegraphics[width=.6\linewidth]{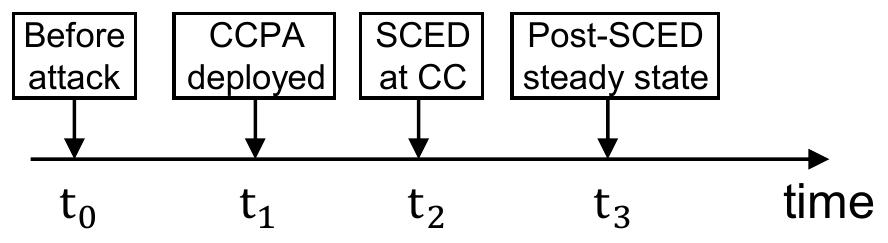}
\vspace{-1em}
\caption{Timeline of an instance of CCPA } \label{fig: timeline_attack}
 \vspace{-1em}
\end{figure}

\begin{table}[htbp]
\renewcommand\arraystretch{1.3}
\small
\vspace{-1em}
\caption{Notations v.s. Timeline}
\centering
\vspace{-.5em}
\begin{tabular}{|c| c |c |c| c|}
\hline
\textbf{time} & $t_0$ & $t_1$ & $t_2$ & $t_3$\\
\hline
\textbf{True measurement matrix} & $\tilde{\vH}$ & $\vH$& $\vH$& $\vH$ \\
\hline
\textbf{Measurement matrix at CC} & --- & --- & $\tilde{\vH}$ & $\tilde{\vH}$ \\
\hline
\textbf{True phase angle} & $\vtheta_0$ & $\vtheta_1$& $\vtheta_2 = \vtheta_1$& $\vtheta_3$ \\
\hline
\textbf{Phase angle at CC} & --- & --- & $\tilde{\vtheta}_2$& $\tilde{\vtheta}_3$\\
\hline
\textbf{True measurement} & $\vz_0$ & $\vz_1$& $\vz_2 = \vz_1$& $\vz_3$\\
\hline
\textbf{measurement at CC} & --- & --- & $\tilde{\vz}_2$&  --- \\
\hline
\end{tabular}
\vspace{-.5em}
\label{tab: notation_timeline}
\end{table}
\normalsize

First, we model the influence of attacks on SCED.  We define $\va_c \in \mathbb{R}^{m}$ to be the \emph{cyber-attack vector}, which changes the measurements received by the CC to $\tilde{\vz}_2 = \vz_2 + \va_c$, and $\va_p \in \{0,1\}^{|E|}$ the \emph{physical-attack vector}, where $a_{p,e} = 1$ indicates that line $e$ is disconnected by the physical attack. As the physical attack changes the topology, we use $\tilde{\mathcal{N}}$ to denote the pre-attack topology and $\mathcal{N}$ the post-attack topology. Accordingly, $\tilde{\vB}, \tilde{\vD}, \tilde{\vH}$ denote the pre-attack admittance, incidence, and measurement matrices, and $\vB, \vD, \vH$ their (true) post-attack counterparts,
related by
\begin{align}\label{eq: attack_BD}
\vB &= \tilde{\vB} - \tilde{\vD}\vGamma\text{diag}(\va_p)\tilde{\vD}^T, ~~~\vD = \tilde{\vD} - \tilde{\vD}\text{diag}(\va_p),
\end{align}
and $\vH = \tilde{\vH} - [(\vLambda_p\tilde{\vD}\vGamma\text{diag}(\va_p)\tilde{\vD}^T)^T, (\vLambda_f \tilde{\vD}\text{diag}(\va_p))^T]^T$.  Falsified measurements in $\tilde{\vz}_2$ and breaker status  will mislead CC to an incorrect state estimation and thus falsified SCED decisions.
Hence, overload-induced line tripping can happen at $t_3$. \looseness=-1

To bypass BDD, the attacker has to manipulate breaker status information to mask the physical attack, misleading the CC to believe that the measurement matrix is $\tilde{\vH}$ instead of $\vH$. Also, measurements have to be modified into $\tilde{\vz}_2$ such that BDD with $\tilde{\vz}_2$ and $\tilde{\vH}$ as input will not raise any alarm. Below, we will derive constraints on $\va_p$ and $\va_c$ such that the modified data can pass BDD under the assumption that the pre-attack data can pass BDD as  assumed in FDI~\cite{deng2017ccpa}. Considering that $\tilde{\vz}_2 = \vz_2 + \va_c$, $\va_c$ should be constructed such that
\begin{align}\label{eq: derivation_ac}
\lp{\tilde{\vz}_2 - \tilde{\vH} \tilde{\vtheta}_2} &= \lp{\vz_0 - \tilde{\vH} \vtheta_0 + \vz_2 + \va_c - \vz_0 + \tilde{\vH} \vtheta_0 - \tilde{\vH} \tilde{\vtheta}_2} \nonumber \\
&=\lp{\vz_0 - \tilde{\vH} \vtheta_0}, ~~~\mbox{(pre-attack residual)}
\end{align}
which leads to the following construction of $\va_c$:
\begin{align}\label{eq: construction_ac}
\va_c &= \vz_0 - \vz_2 + \tilde{\vH}( \tilde{\vtheta}_2 - \vtheta_0)\\
&= \tilde{\vH}\vtheta_0 + \epsilon_0 - \left(\vH\vtheta_2 + \epsilon_0 \right) + \tilde{\vH}( \tilde{\vtheta}_2 - \vtheta_0) \\
&= \left[\begin{array}{c}
     \vLambda_p \tilde{\vB}  \\
     \vLambda_f \vGamma \tilde{\vD}^T
\end{array}\right] \tilde{\vtheta}_2 - \left[\begin{array}{c}
     \vLambda_p \vB  \\
     \vLambda_f \vGamma \vD^T
\end{array}\right]\vtheta_2. \label{eq: constr_ac_theta}
\end{align}
Besides \eqref{eq: construction_ac}, there may be additional constraints on $\va_c$ to avoid causing suspicion.  Specifically,
following \cite{yuan2011modeling}, we assume that all the power injections at generator buses are measured and not subject to attacks, i.e.,
\begin{align}\label{eq: constr_gen_no_attack}
\vLambda_g \tilde{\vD}\tilde{\vf}_2 = \vLambda_g\tilde{\vB}\tilde{\vtheta}_2 = \vLambda_g\vB\vtheta_2 =\vLambda_g\vD\vf_2 =\vLambda_g\vp_0,
\end{align}
recalling that 
{$\vLambda_g$ is the row selection matrix corresponding to generator buses}.
Moreover, by representing the maximum normal load fluctuation through $\alpha\geq 0$, the magnitude of falsification at load buses needs to be constrained due to load forecasting \cite{yuan2011modeling,che2018false}, which can be modeled by \footnote{In contrast to \cite{yudi21SmartGridComm} that only imposes the magnitude constraint on measured buses, constraint \eqref{eq: attack_FDIrange} is imposed on all buses (although subsumed by \eqref{eq: constr_gen_no_attack} for generator buses). This is because under the assumption of full-rank measurement matrix (Section~\ref{subsec:Power Grid Modeling}), the CC can recover all the phase angles and hence the power injections at all the buses, and thus the attacker needs to avoid causing too much deviation in the power injections at all the buses. }
\begin{align}\label{eq: attack_FDIrange}
-\alpha |\vp_0| \leq \tilde{\vB}\tilde{\vtheta}_2 - \vp_0 \leq \alpha |\vp_0|.
\end{align}

Following the convention in \cite{wu2016efficient, yuan2011modeling}, the attack is constrained by a predefined constant $\xi_p$ denoting the maximum number of attacked lines and another constant $\xi_c$ denoting the maximum number of manipulated measurements, i.e.,  
\begin{align}\label{eq: const_phy}
    \lp{\va_p}_0 \le \xi_p, ~~~\lp{\va_c}_0 \le \xi_c.
\end{align}

In addition, we constrain $\va_p$ so that the graph after physical attack remains connected, which is needed for stealth of the attack according to \cite{che2018false, liu2016masking}. Specifically, defining $\vf_{con} \in \mathbb{R}^{|E|}$ as a pseudo flow and $u_0$ as the reference node, we can guarantee network connectivity at $t_2$ by ensuring
\begin{subequations}
\allowdisplaybreaks
\label{eq:connectivity_constr}
\begin{align}
\tilde{\vD}_u\vf_{con} &= \begin{cases}|V|-1, &\mbox{ if } u = u_0,\\
-1, &~\mbox{if } u\in V\setminus\{u_0\},
\end{cases} \label{eq:flow conservation}\\
-|V|\cdot (1-a_{p,e}) &\le f_{con,e} \le |V|\cdot (1-a_{p,e}). \label{eq:a_p constraint}
\end{align}
\end{subequations}
With lines oriented as in $\tilde{\vD}$, \eqref{eq:flow conservation} (flow conservation constraint) and \eqref{eq:a_p constraint} (line capacity constraint) ensure the existence of a unit pseudo flow from $u_0$ to every other node in the post-attack grid and hence the connectivity of the post-attack grid, where $f_{con,e}>0$ if the flow on $e$ is in the same direction of the line and $f_{con,e}<0$ otherwise.

{In practice, transmission lines are equipped with overcurrent protection devices, which will trip the lines when the power flow exceeds the tripping threshold. Thus, heavy overloading caused by the SCED misled by cyber attacks can lead to initial outages at $t_3$, which can create cascading outages \cite{che2018false}. Specifically, let $\vf_{max}\in \mathbb{R}^{|E|}$ be the normal power flow limits imposed in SCED \cite{NERCrating}. Then, a line $e\in E$ will be tripped by protection devices (i.e., having an outage) if
\begin{align}\label{eq: def_outage}
    |f_e| > \gamma_e f_{\tiny\text{max},e},
\end{align}
where $\gamma_e$ denotes the tripping threshold based on the thermal limit of the line. In practice, although \cite{NERCtripping} suggests $\gamma_e \ge 1.5$, the operator may choose higher $f_{\tiny\text{max},e}$, which leads to a smaller $\gamma_e$. As discussed in \cite{che2018false, athari2017impacts}, a small $\gamma_e$ implies that the system is operating with a low margin of overload. A large $\gamma_e$ may contribute to robustness to cascading failure \cite{athari2017impacts}, but leads to underutilization of transmission lines.}

\subsection{Modeling the Protection Effect of Secured PMUs}\label{subsec:Modeling Protection Effect of Secured PMUs}

Let $\vbeta\in\{0,1\}^{|V|}$ be the indicator vector for PMU placement such that $\beta_u = 1$ if and only if a secured PMU is installed at node $u$. We define $\Omega(\vbeta):= \{u|\beta_u > 0\}$ and the inverse process $\vbeta(\Omega): \beta_u = 1$ if $u\in \Omega$ and $\beta_u = 0$ otherwise. Let $\mathcal{V}_u$ be the node set containing neighbors of node $u$ (including $u$) and $E_u$ be the line set composed of lines incident on $u$. According to \cite{yang2017optimal}, by measuring both voltage and current phasor, a PMU on node $u$ can guarantee the correctness of phase angles in $\mathcal{V}_u$ and protect lines in $E_u$ from both cyber and physical attacks. Formally, we define $\vx_N \in \{0,1\}^{|V|}$ such that $(x_{N,u} = 1) \leftrightarrow ( \exists v \in \mathcal{V}_u \mbox{ such that }\beta_v=1)$, which can be modeled as \looseness=-1
\begin{align}\label{eq: const_beta_x_N}
\vDelta^{-1}\underline{\vA} \vbeta\le \vx_N \le \vDelta^{-1}\underline{\vA} \vbeta + \frac{\lp{\vDelta}_{\infty}-1}{\lp{\vDelta}_{\infty}},
\end{align}
where $\vDelta \in \mathbb{Z}^{|V|\times|V|}$ is a diagonal matrix with $\Delta_{uu}=|\mathcal{V}_u|$, while $\underline{\vA}:=\vA + \vI$ is the adjacency matrix of the grid with added self-loops at all nodes. Similarly, we define $\zeta$ to be any constant within $[0.5,1)$ and $\vx_L \in \{0,1\}^{|E|}$ satisfying $(x_{L,e} = 1)\leftrightarrow (\exists v \mbox{ with }e\in E_v \mbox{ and } \beta_v = 1 )$, which can be linearlized as
\begin{align}\label{eq: const_beta_x_L}
0.5|\vD|^T  \vbeta \le \vx_L \le 0.5|\vD|^T  \vbeta + \zeta.
\end{align}

We assume that the PMU locations are known to the attacker, thus the cyber attack is constrained as follows:
\begin{subequations}\label{eq: PMU_defend}
\begin{alignat}{2}
&x_{N,u}  = 1 \rightarrow \tilde{\vtheta}_{2,u} = \vtheta_{2,u}, &\forall u \in V, \label{eq: beta_on_theta} \\
&x_{L,e}  = 1 \rightarrow a_{p,e} = 0, &\forall e \in E. \label{eq: beta_on_ap}
\end{alignat}
\end{subequations}
Note that \eqref{eq: const_beta_x_N}-\eqref{eq: PMU_defend} implicitly protect the power flow measurements on lines incident to a PMU. To see this, suppose that $e=(s,t)$ and $\beta_s = 1$. Then we must have $x_{N,s} = x_{N,t} = x_{L,e} = 1$ due to \eqref{eq: const_beta_x_N}-\eqref{eq: const_beta_x_L}. By \eqref{eq: PMU_defend}, it is guaranteed that $\tilde{z}_{2,e}:= (\tilde{\theta}_{2,s} - \tilde{\theta}_{2,t}) / r_{st}= (\theta_{2,s} - \theta_{2,t}) / r_{st} =: z_{2,e}$. \rev{In addition, PMU data are usually collected at a high frequency (e.g., around $60$-$200$ samples per second). Thus, the PMUs can \quotes{instantly} detect any attack violating \eqref{eq: PMU_defend} even though they cannot prevent the attack from happening. In this way, the PMUs can reduce the potential damage by restricting the attacker's choices of attack vectors.}
\looseness=-1

\subsection{Optimal PMU Placement Problem}
Our main problem, named \emph{PMU Placement for Outage Prevention (PPOP)}, aims at {placing the minimum number of secured PMUs so that no undetectable CCPA can cause overload-induced tripping}.
To achieve this, we model the problem as a tri-level optimization problem (an overview of PPOP is given in Fig.~\ref{fig: ppop} in Appendix~A). 
\looseness=-1

The \emph{middle-level} optimization is the attacker's problem, which aims to maximize the number of overloaded lines without being detected. Instead of using $\va_c$ as decision variable, we propose to formulate over $\tilde{\vf}_i,\vf_i$ and $\tilde{\vtheta}_i, \vtheta_i$ where $i\in\{2,3\}$. In the rest of the paper, we will apply big-M modeling technique that introduces sufficiently large constants denoted as $M_{(\cdot)}$ for linearization. The calculation of $M_{(\cdot)}$ is given in Appendix~\ref{appendix:Calculation of Big-M}. 
Specifically, the constraints on $\vtheta_2$ and $\vf_2$ are:
\begin{subequations}\label{eq: dc_t2_gt}
\allowdisplaybreaks
\begin{align}
&{-M_{2,a,e}\lrbrackets{\bm{1}-\va_{p}}\le \vf_2  \le M_{2,a,e} \lrbrackets{\bm{1}-\va_{p}},} \label{f2_range}\\
&\tilde{\vD} \vf_2 = \vp_0, \label{f2_flow_conservation} \\
&-M_{2,f}\va_{p}\le \vGamma \vD \vtheta_2-\vf_2 \le M_{2,f}\va_p.\label{f2_theta_valid}
\end{align}
\end{subequations}
The constraints \eqref{f2_range} and \eqref{f2_flow_conservation} guarantee the consistency between $\vf_2$ and $\vp_0$ given $\va_p$, where $a_{p,e} = 1$ will force $f_{2,e} = 0$. The role of \eqref{f2_theta_valid} is to force the consistency between $\vf_2$ and $\vtheta_2$ for all $e$ with $a_{p,e} = 0$, which is necessary for the uniqueness of $\vf_2$. {Similarly, we can transform \eqref{eq: derivation_ac}-\eqref{eq: const_phy} into constraints over $\tilde{\vf}_2$, $\tilde{\vtheta}_2$, and $\va_p$, which are} 
\begin{subequations}\label{constr: f2_fdi}
\allowdisplaybreaks
\begin{align}
&-\vf_{\tiny\text{max}}\le \tilde{\vf}_{2} \le \vf_{\tiny\text{max}}, \label{eq: f2_fdi_valid_range} \\
&\vGamma \tilde{\vD}^T \tilde{\vtheta}_2-\tilde{\vf}_{2}  = 0, \label{eq: f2_fdi_valid_unique}\\
&\tilde{\theta}_{2,u} - \theta_{2,u} \le M_{2,\theta}\cdot (1 - \vx_{N,u}), \label{eq: constr_theta2_pmu_p} \\
&\tilde{\theta}_{2,u} - \theta_{2,u} \ge -M_{2,\theta}\cdot (1 - \vx_{N,u})\label{eq: constr_theta2_pmu_n},\\
&-\alpha |\vp_0| \le \tilde{\vD}\tilde{\vf}_2 - \vp_0 \le \alpha |\vp_0|, \label{eq: fdi_range}\\
&\vLambda_g \tilde{\vD}\tilde{\vf}_2 =\vLambda_g\vp_0, \label{eq: no_attack_generator}\\
&\lp{\vLambda_f \lrbrackets{\tilde{\vf}_2 - \vf_2}}_0  + \lp{ \vLambda_p \lrbrackets{\tilde{\vD}\tilde{\vf}_2 - \vp_0} }_0 \le \xi_c, \label{eq: f2_fdi_xi_c}\\
&\lp{\va_p}_0 \le \xi_p, \label{eq:xi_p}
\end{align}
\end{subequations}
{where \eqref{eq: f2_fdi_valid_range}-\eqref{eq: f2_fdi_valid_unique} guarantee the validity of $\tilde{\vf}_2$ as in \eqref{f2_range}-\eqref{f2_theta_valid}, \eqref{eq: constr_theta2_pmu_p}-\eqref{eq: constr_theta2_pmu_n} linearize \eqref{eq: beta_on_theta} ($M_{2,\theta}$ defined in Appendix~\ref{appendix:Calculation of Big-M}), 
while \eqref{eq: fdi_range}, \eqref{eq: no_attack_generator}, and \eqref{eq: f2_fdi_xi_c}--\eqref{eq:xi_p} correspond to \eqref{eq: attack_FDIrange}, \eqref{eq: constr_gen_no_attack}, and \eqref{eq: const_phy}, respectively. It is worth noting that there exists an $\va_c$ in the form of \eqref{eq: constr_ac_theta} for any $\tilde{\vf}_2$ and $\tilde{\vtheta}_2$ satisfying \eqref{constr: f2_fdi} due to the relationship between $\tilde{\vf}_2$, $\tilde{\vtheta}_2$ and $\va_c$ shown in \eqref{eq: constr_ac_theta} and \eqref{eq: f2_fdi_valid_unique}.}
Moreover, the constraints on $\vtheta_3,\tilde{\vtheta}_3$, and $\vf_3$ are
\begin{subequations}\label{eq: constr_f3}
\allowdisplaybreaks
\begin{align}
&\vp_{g,min} \le \vLambda_g \tilde{\vB} \tilde{\vtheta}_3 \le \vp_{g,max} \label{eq: pg_theta_3_fdi}\\
&-\vf_{\tiny\text{max}} \le \vGamma\vD^T \tilde{\vtheta}_3 \le \vf_{\tiny\text{max}}, \label{eq: fmax_theta_3_fdi}\\
&\vLambda_d \tilde{\vB} \tilde{\vtheta}_3 = \vLambda_d\tilde{\vD}\tilde{\vf}_2\label{eq: pd_theta_3_fdi}\\
&-M_{3,a}(\bm{1}-\va_{p})\le \vf_3  \le M_{3,a}(\bm{1}-\va_{p}),\label{eq:f_3, a_p}\\
&\vLambda_d\tilde{\vD}\vf_3 = \vLambda_d\vp_0,~~~ \vLambda_g\tilde{\vD}\vf_3 = \vLambda_g\tilde{\vB}\tilde{\vtheta}_3,~ \label{eq: attack_load}\\
&-M_{3,f}\va_p\le \vGamma \tilde{\vD}^T \vtheta_3-\vf_3 \le M_{3,f} \va_p, \label{eq:theta_3, f_3}
\end{align}
\end{subequations}
where \eqref{eq: pg_theta_3_fdi}-\eqref{eq: pd_theta_3_fdi} describe the feasible region of $\tilde{\vtheta}_3$ under false data injection, and \eqref{eq:f_3, a_p}--\eqref{eq:theta_3, f_3} are used to enforce the power flow equation \eqref{eq:B theta = p} at $t_3$, where $\vLambda_g\tilde{\vB}\tilde{\vtheta}_3$ is the post-SCED generation predicted by the attacker. While a straightforward formulation of the power flow equation should be
\begin{align}\label{eq: bilinear_terms}
    \vGamma \vD^T \vtheta_3=\vf_3, ~ \vLambda_d \vD\vf_3 = \vLambda_d\vp_0,~ \vLambda_g \vD\vf_3 = \vLambda_g\tilde{\vB}\tilde{\vtheta}_3,
\end{align}
such a formulation will introduce bilinear terms $\vD^T\vtheta_3$ and $\vD\vf_3$, as the post-attack incidence matrix $\vD$ is a function of the physical-attack vector $\va_p$ that is also a decision variable for the attacker. To avoid the bilinear terms, we use \eqref{eq:f_3, a_p} to force $f_{3,e} = 0$ when $a_{p,e} = 1$ (line $e$ is disconnected), and \eqref{eq:theta_3, f_3} to force $\Gamma_e \tilde{\vd}_e^T \vtheta_3 = \Gamma_e \vd_e^T \vtheta_3 = f_{3,e}$ when $a_{p,e} = 0$. Moreover, under \eqref{eq:f_3, a_p}, we observe that $\vD\vf_3 = \sum_{e\in E} \vd_e f_{3,e} = \sum_{e\in E} \tilde{\vd}_e f_{3,e} = \tilde{\vD}\vf_3$, as $\vd_e = \tilde{\vd}_e$ if $a_{p,e}=0$ and $\vd_e f_{3,e} = \tilde{\vd}_e f_{3,e} = \bm{0}$ if $a_{p,e}=1$, which explains \eqref{eq: attack_load}. \looseness=-1

Thus, the attacker's problem, which defines 
the optimal attack strategy, can be formulated as:
\begin{maxi!}|s|[2]<b>
    {}
    {\lp{\vpi}_0}
    {\label{eq:reformulate_attacker}}
    {\psi_a(\vbeta):=}
    \addConstraint{\eqref{eq:connectivity_constr}, \eqref{eq: const_beta_x_N}-\eqref{eq: constr_f3}}{}\label{eq: fdi_bypass}
    \addConstraint{\theta_{2,u_0}=\theta_{3,u_0} =  \tilde{\theta}_{2,u_0} = \tilde{\theta}_{3,u_0}=0}{}\label{eq: zero_ref_theta}
    \addConstraint{\tilde{\vtheta}_3 = \psi_s(\tilde{\vB}\tilde{\vtheta}_2, \tilde{\vD})}{}\label{eq: false_sced}
    \addConstraint{\frac{|f_{3,e}|}{f_{\tiny\text{max},e}} > \gamma_e  \leftrightarrow \pi_e = 1, \forall e\in E,}{}\label{eq: attack_overflow}
\end{maxi!}
where $\vy_{c}:= \tilde{\vtheta}_2 \oplus \tilde{\vtheta}_3 \oplus \vtheta_2 \oplus \vtheta_3 \oplus \vf_2\oplus \vf_3 \oplus \tilde{\vf}_2 \oplus \vf_{con} $ 
and $\vy_{b}:= \vpi \oplus \va_p \oplus \vx_N \oplus \vx_L$ are continuous and binary decision variables, respectively. Here, $\pi_e=1$ if and only if line $e$ is overloaded to be tripped, which is ensured by \eqref{eq: attack_overflow}. Thus, the objective is to maximize the number of overload-induced tripped lines due to the attack-induced load redistribution. The constraints \eqref{eq: zero_ref_theta} fixes the phase angle at the reference node, denoted as node $u_0$. The constraint \eqref{eq: false_sced} incorporates the \emph{lower-level} optimization of SCED \eqref{eq:reformulate_OPF} by specifying the post-SCED generation, determined by $\tilde{\vtheta}_3$.

We formulate the \emph{upper-level} PMU placement problem as
\begin{mini!}|s|<b>
{}{ \lp{\vbeta}_0 }
{\label{eq:PMU_placement}}
{}
\addConstraint{\psi_a(\vbeta) = 0}{}\label{eq: PMU_no_overflow}
\end{mini!}
where the decision variable is $\vbeta\in \{0,\: 1\}^{|V|}$, and $\psi_a(\vx)$ defined in \eqref{eq:reformulate_attacker} denotes the maximum number of lines that will be tripped according to \eqref{eq: def_outage} at $t_3$. In the sequel, we call $(\va_p,\va_c,e)$ an \emph{attack tuple}, which is called ``successful'' under PMU placement $\vbeta$ if there exists a feasible solution to \eqref{eq:reformulate_attacker} with physical attack $\va_p$ and cyber attack $\va_c$ such that $\pi_e = 1$. Moreover, {we call $(\va_p,e)$ a successful \emph{attack pair} under $\vbeta$ if it can form a successful attack tuple under $\vbeta$.} 

\emph{Remark~1:}
While the above formulation treats the load profile $\vp_0$ as a constant, it can be easily extended to handle the fluctuations in loads. This can be modeled by treating $\vp_0$ as a decision variable in the attacker's optimization, constrained by the expected range of fluctuation, e.g., $\vp_0 \in [\underline{\kappa}\vp^{(0)}, \overline{\kappa}\vp^{(0)}]$, or the union of ranges around multiple operating points:
\begin{align}\label{eq: varying_load_profile_p0}
    \vp_0 \in \bigcup_{i=1}^{i_0} \{  \underline{\kappa}_i \vp^{(i)}\le \vp \le \overline{\kappa}_i \vp^{(i)}\}.
\end{align}
This enlarges the solution space for the attacker, which changes the meaning of $\psi_a(\vbeta)$ to the \emph{maximum number of tripped lines under the worst load profile and the worst attack under this load profile}. Clearly, a PMU placement that avoids overload-induced tripping in this worst scenario can avoid overload-induced tripping in any scenario encountered during operation, as long as the load profile stays within the predicted range.

\emph{Remark~2:}
{In practice, PMUs are often deployed in stages. Thus, it may be desirable that a temporary PMU placement designed to prevent outages can be augmented into an optimal PMU placement $\vbeta^{opt}$ in the long run (e.g., a minimum placement that provides full observability). This can be modeled by adding a constraint in \eqref{eq:PMU_placement} that requires $\vbeta\leq\vbeta^{opt}$.}

\section{Solving PPOP}\label{sec: solve_PPOP_optimally}

The PPOP problem \eqref{eq:reformulate_attacker}-\eqref{eq:PMU_placement} is a tri-level non-linear mixed integer problem, which is notoriously hard \cite{liu2016masking}. In this section, we first formally prove that the problem is NP-hard, and then demonstrate how to transform it into a \emph{bi-level mixed-integer linear programming (MILP)} problem. Next, we propose an alternating optimization framework based on constraint generation to solve the problem optimally. Finally, to accelerate the computation, we develop a polynomial-time heuristic. 

\subsection{Hardness and Conversion to Bi-Level MILP}

Although multi-level non-linear mixed integer programming is generally hard, PPOP is only a special case and hence needs to be analyzed separately. Nevertheless, we show that PPOP is NP-hard (see proof in Appendix~\ref{appendix: additional_proof}). 



\begin{theorem}\label{theo: NP-hardness}
The PPOP problem \eqref{eq:PMU_placement} is NP-hard.
\end{theorem}
The attacker's problem \eqref{eq:reformulate_attacker} can be linearized into a MILP (see details in Appendix~\ref{appendix: bi_level_milp}), 
which implies that PPOP can be converted into a bi-level MILP.

\subsection{An Alternating Optimization Framework}

\begin{algorithm}\label{alg: alter_opt_v1}
\SetAlgoLined
\SetKwFunction{Fmain}{FailEdgeDetection}
\SetKwInOut{Input}{input}\SetKwInOut{Output}{output}
\textbf{Initialization: }$k=1$, $\hat{\vbeta}^{(k)} = \bm{0}$\;
\While{True}{
    Solve \eqref{eq:reformulate_attacker} under $\hat{\vbeta}^{(k)}$ to obtain $\psi_a(\hat{\vbeta}^{(k)})$\; \label{alter_opt:attacker}
    \uIf{$\psi_a(\hat{\vbeta}^{(k)})>0$}{
    Add constraints to \eqref{eq:PMU_placement}\label{Line: update_beta}\;
    $k\gets k+1$, obtain $\hat{\vbeta}^{(k)}$ by solving \eqref{eq:PMU_placement}, {with \eqref{eq: PMU_no_overflow} replaced by the generated constraints}\label{Line: update_beta_new_constr}
    }
    \lElse{
    break
    }
}
Return $\hat{\vbeta}^{(k)}$, indicators of the selected PMU placement\; 
\caption{Alternating Optimization}
\end{algorithm}
As a bi-level MILP, PPOP is still difficult to solve due to the integer variables in \eqref{eq:reformulate_attacker} and \eqref{eq:PMU_placement}.
Since one of the fundamental challenges in solving bi-level MILPs is the lack of explicit description of the upper-level optimization's feasible region, we propose an alternating optimization framework shown in Alg.~\ref{alg: alter_opt_v1} {to solve PPOP by gradually approximating the feasible region of the upper-level optimization through constraint generation}. In Sections~\ref{subsec:AONG}--\ref{subsec:AODC}, we will give two concrete constraint generation methods for Line~\ref{Line: update_beta} of Alg.~\ref{alg: alter_opt_v1} based on the results of \eqref{eq:reformulate_attacker}. \looseness=-1

In the sequel, we assume that solving \eqref{eq:reformulate_attacker} returns a successful attack tuple $(\va_p^{(k)}, \va_c^{(k)}, e^{(k)})$ if $\psi_a(\hat{\vbeta}^{(k)})>0$.

\subsection{Alternating Optimization with No-Good Constraints (AONG)}\label{subsec:AONG}

In this section, we give the first specific algorithm under the framework of Alg.~\ref{alg: alter_opt_v1}, in which the added constraints in Line~\ref{Line: update_beta} are motivated by the following observation:
\begin{lemma}\label{lem: ob1}
Given $\hat{\vbeta}$ and $\Omega(\hat{\vbeta}):=\{u\in V:\: \hat{\beta}_u>0\}$, if there exists a successful attack tuple $(\va_p, \va_c, e)$, then for all $\vbeta$ with $\Omega(\vbeta)\subseteq\Omega(\hat{\vbeta})$, there exists a successful attack tuple.
\end{lemma}
\begin{proof}
For any $\vbeta$ with $\Omega(\vbeta)\subseteq\Omega(\hat{\vbeta})$, $(\va_p, \va_c, e)$ remains a successful attack tuple.
\end{proof}
The above observation indicates that at least one PMU must be placed in $\Omega(\hat{\vbeta})^c:= V\setminus \Omega(\hat{\vbeta})$. Therefore, the optimal $\vbeta$ can be obtained in an iterative manner: during each iteration, we 
use the PMU placement $\hat{\vbeta}$ from the previous iteration (initially, $\hat{\vbeta}=\bm{0}$) to solve \eqref{eq:reformulate_attacker} for $\psi_a(\hat{\vbeta})$. If $\psi_a(\hat{\vbeta}) = 0$, $\hat{\vbeta}$ is the final solution; otherwise, we will add the following ``No-Good'' constraint: {$\sum_{i:\hat{\beta}_i=0}\beta_i \ge 1$}
to \eqref{eq:PMU_placement} for the next iteration to rule out the infeasible solution $\hat{\vbeta}$.

However, the above procedure will converge very slowly as $|\Omega(\hat{\vbeta})^c|$ is usually large. To speed up convergence, we augment each discovered infeasible solution $\hat{\vbeta}$ into a maximal infeasible solution $\hat{\vbeta}'$ to narrow down candidate solutions. 
This can be achieved by solving the following problem:
\begin{maxi!}|s|<b>
{}{ \lp{\hat{\vbeta}'}_0 }
{\label{eq: opt_max_nogood_but}}
{}
\addConstraint{\psi_a(\hat{\vbeta}') \ge 1}{}
\addConstraint{\hat{\beta}_u' = 1, \forall u\in V \mbox{ with }\hat{\beta}_u = 1,}{}
\end{maxi!}
which has the same decision variables as \eqref{eq:reformulate_attacker} and the additional $\hat{\vbeta}'$. Algorithm AONG adds the following \quotes{No-Good} constraint in Line~\ref{Line: update_beta} of Alg.~\ref{alg: alter_opt_v1}:
\begin{align}\label{eq: no_good_cut}
\sum_{i:\hat{\beta}'_i=0}\beta_i \ge 1.
\end{align}

AONG solves PPOP optimally, as proved in Appendix~\ref{appendix: additional_proof}. 

\begin{theorem}\label{thm:Optimality of alternating optimization}
AONG converges in finite time to an optimal solution to \eqref{eq:PMU_placement}.
\end{theorem}
Given the MILP formulation of \eqref{eq:reformulate_attacker} in Appendix~\ref{appendix: bi_level_milp}, 
it is easy to write \eqref{eq: opt_max_nogood_but} as a MILP and solve it by existing MILP solvers.
It is worth noting that solving \eqref{eq: opt_max_nogood_but} suboptimally does not affect the optimality of AONG. Thus, we can also apply heuristic algorithms (e.g., LP relaxation with rounding).

\subsection{Alternating Optimization with Double Constraints (AODC)}\label{subsec:AODC}
Building on AONG, we develop an additional constraint as a complement of \eqref{eq: no_good_cut} to accelerate convergence, in the special case where $\xi_c = \infty$ and $\psi_s(\vp, \vD)$ returns the set of $\vtheta$'s satisfying \eqref{eq: cc_load_meet}-\eqref{eq: cc_pgp}, i.e., it returns the feasible region of SCED 
rather than a single solution.  Such a special case is worth consideration because (i) $\xi_c = \infty$ represents the strongest cyber attack, and (ii) relaxing the optimality requirement in \eqref{eq: false_sced} means that the attacker is allowed to pick a solution for SCED within its feasible region, both making the attack stronger and hence the resulting PMU placement more robust in preventing outages.\looseness=-1 


Below we will first introduce the new constraints, called \quotes{Attack-Denial} constraints, and then give the AODC algorithm, in which both \quotes{No-Good} constraints and \quotes{Attack-Denial} constraints are added in Line~\ref{Line: update_beta} of Alg.~\ref{alg: alter_opt_v1}. The new constraints are motivated by the following observations about AONG: many PMU placements enumerated by AONG are vulnerable to attacks formed from the same attack pair $(\va_p, e)$, indicating that it is more efficient to generate constraints that can invalidate the identified attack pairs. More discussions are given in Appendix~\ref{appendix: effi_no_good}. 
\looseness=-1

The above observations motivate the following idea of \quotes{Attack-Denial} constraints: \emph{given a successful attack pair $(\va_p^{(k)}, e^{(k)})$ under $\vbeta^{(k)}$, the added constraints 
should guarantee that any PMU placement satisfying the constraints can prevent attacks that fail lines according to $\va_p^{(k)}$ from causing overload-induced tripping at line $e^{(k)}$.} We focus on $(\va_p^{(k)}, e^{(k)})$ instead of $(\va_p^{(k)}, \va_c^{(k)}, e^{(k)})$ due to the following observations:
\begin{enumerate}
    \item The number of $(\va_p^{(k)}, \va_c^{(k)}, e^{(k)})$'s is infinite since $\va_c^{(k)}$ is continuous, but the number of $(\va_p^{(k)}, e^{(k)})$'s is finite.
    \item Given $\vx_N$ and $(\va_p^{(k)}, e^{(k)})$, \eqref{eq: fdi_bypass}-\eqref{eq: attack_overflow} reduce to a linear system with only the continuous variables contained in $\vy_c$ under the assumptions that $\xi_c = \infty$ and $\psi_s(\vp, \vD)$ returns the set of $\vtheta$'s satisfying \eqref{eq: cc_load_meet}-\eqref{eq: cc_pgp}. The linear system can be summarized as
    \begin{subequations}\label{eq: LP_given_pair}
    \begin{align}
    \vF_1^{(k)}\vy_c &= \vs_1^{(k)}, \label{eq: LP_given_pair_eq}\\
    \vF_2^{(k)} \vy_c &\le \vs_2^{(k)} + \vF_3 \vx_N, \label{eq: LP_given_pair_ineq}
    \end{align}
    \end{subequations}
    where $\vF_1^{(k)}$, $\vF_2^{(k)}$, $\vF_3$, $\vs_1^{(k)}$, $\vs_2^{(k)}$ are constant matrices/vectors defined in Appendix~\ref{appendix: expansion_attack_denial_primal}. 
    An attack pair $(\va_p^{(k)},e^{(k)})$ can form a successful attack if {and only if} \eqref{eq: LP_given_pair} has a feasible solution.\looseness=-1
\end{enumerate}
{The above assumptions (i.e., $\xi_c=\infty$ and $\psi_s(\vp, \vD)$ returns all the $\vtheta$'s satisfying \eqref{eq: cc_load_meet}-\eqref{eq: cc_pgp}) are needed because: (i) $\xi_c = \infty$ implies that we no longer need {the binary variables used to linearize \eqref{eq: f2_fdi_xi_c} (i.e., $\vw_f$ and $\vw_p$ in \eqref{eq: linearization_xi_c} in Appendix~\ref{appendix: bi_level_milp})};
(ii) when the lower-level optimization returns the feasible region of \eqref{eq:reformulate_OPF}, \eqref{eq: false_sced} can be replaced by \eqref{eq: cc_load_meet}-\eqref{eq: cc_pgp} without introducing binary variables required for transforming \eqref{eq:reformulate_OPF} into its KKT conditions~\cite{yuan2011modeling}.} \looseness=-1

Our key observation is that a PMU placement $\vbeta$ can defend against an attack pair $(\va_p^{(k)}, e^{(k)})$ by either preventing the physical attack $\va_p^{(k)}$ or making \eqref{eq: LP_given_pair} infeasible. The former can be achieved by adding constraint $\sum_{l:a_{p,l}^{(k)}=1}x_{L,l} \ge 1$ (i.e., at least one attacked line must be incident to a PMU). The latter holds according to Gale's theorem of alternative \cite{mangasarian1994nonlinear} if and only if there exists $\vq_1^{(k)}$ and $\vq_2^{(k)} \ge \bm{0}$ satisfying
\begin{subequations}\label{eq: gale_formulation}
\begin{align}
(\vF_1^{(k)})^T \vq_1^{(k)} + (\vF_2^{(k)})^T \vq_2^{(k)} &= \bm{0},\\
(\vs_1^{(k)})^T\vq_1^{(k)} + (\vs_2^{(k)} + \vF_3 \vx_N)^T\vq_2^{(k)} &< 0, \label{eq: gale_strict_ineq}
\end{align}
\end{subequations}
where $\vq_1^{(k)} \in \mathbb{R}^{m_1}$ and $\vq_2^{(k)}\in \mathbb{R}^{m_2}$ can be treated as the dual variables for \eqref{eq: LP_given_pair_eq} and \eqref{eq: LP_given_pair_ineq}, respectively.

Based on the above observation, the \quotes{Attack-Denial} constraints for  defending against  $(\va_p^{(k)}, e^{(k)})$ are:
\begin{subequations}\label{eq: attack_denial_constr}
\allowdisplaybreaks
\begin{align}
&(\vF_1^{(k)})^T \vq_1^{(k)} + (\vF_2^{(k)})^T \vq_2^{(k)} = \bm{0}, \label{eq: gale_eq}\\
&(\vs_1^{(k)})^T\vq_1^{(k)} + (\vs_2^{(k)}+\vF_3 \vx_N)^T\vq_2^{(k)} \le w_{a,k}-1, \label{eq: gale_ineq} \\
&\sum_{l:a_{p,l}^{(k)}=1}x_{L,l} \ge w_{a,k}, \label{eq: attack_denial_xL} \\
&\vq_2^{(k)} \ge \bm{0}, w_{a,k}\in\{0,1\},
\end{align}
\end{subequations}
where {$\vq_1^{(k)}$,} $\vq_2^{(k)}$, and $w_{a,k}$ are newly introduced variables. Note that \eqref{eq: gale_strict_ineq} and \eqref{eq: gale_ineq} are equivalent when $w_k = 0$ since we can scale $\vq_1^{(k)}$ and $\vq_2^{(k)}$ to satisfy \eqref{eq: gale_ineq} if \eqref{eq: gale_strict_ineq} holds. 
The binary variable $w_{a,k}$ indicates which approach to use for defending against $(\va_p^{(k)}, e^{(k)})$. When $w_{a,k} = 0$, \eqref{eq: attack_denial_xL} holds trivially, in which case $\vbeta$ defends against $(\va_p^{(k)}, e^{(k)})$ by satisfying \eqref{eq: gale_formulation}, i.e., preventing the cyber attack from causing overload-induced tripping at line $e^{(k)}$. When $w_{a,k} = 1$, $\vq_1^{(k)} = \bm{0}$ and $\vq_2^{(k)} = \bm{0}$ will satisfy the constraints \eqref{eq: gale_eq}-\eqref{eq: gale_ineq}, in which case $\vbeta$ defends against $(\va_p^{(k)}, e^{(k)})$ by preventing the physical attack $\va_p^{(k)}$.

Now, we are ready to present the AODC algorithm, where $\hat{\vbeta}^{(K+1)}$ in Line~\ref{Line: update_beta_new_constr} of Alg.~\ref{alg: alter_opt_v1} is obtained by solving:
\begin{mini!}|s|<b>
{}{ \lp{\vbeta}_0 }
{\label{eq: vcg_formulation}}
{}
\addConstraint{\eqref{eq: const_beta_x_N}-\eqref{eq: const_beta_x_L}, \eqref{eq: attack_denial_constr} \mbox{ for } k = 1,\cdots, K}{} \label{eq:x_N,x_L}
\addConstraint{\sum_{i:\hat{\beta}'^{(k)}_i=0}\beta_i \ge 1, k = 1,\cdots, K}{}
\addConstraint{\vbeta\in \{0,\: 1\}^{|V|},}{}
\end{mini!}
where the decision variables are $\vbeta$, $\vx_N$, $\vx_L$, $\vq_1^{(k)}$, $\vq_2^{(k)}$, and $w_{a,k}$ for $k=1,\cdots,K$.

To convert \eqref{eq: vcg_formulation} to a MILP, we linearize $(\vF_3\vx_N)^T \vq_2^{(k)}$ using McCormick's relaxation. Concretely, note that \looseness=-1
\begin{align}
(\vF_3\vx_N)^T \vq_2^{(k)} = \sum_{u\in V} \vx_{N,u} \left( \sum_{i=1}^{m_2} F_{3,i,u} q_{2,i}^{(k)} \right), \forall k.
\end{align}
Assuming that $\sum_i F_{3,i,u} q_{2,i}^{(k)} \in [\underline{M}_{F}, \overline{M}_{F}]$, we introduce a continuous auxiliary variable $y_u$
and the following constraints: \looseness=-1
\begin{subequations}\label{eq: McCormick_relaxation}
\allowdisplaybreaks
\begin{align}
\underline{M}_{F} x_{N,u} &\le y_u \le \overline{M}_{F} x_{N,u},\\
y_u &\le \left( \sum_{i=1}^{m_2} F_{3,i,u} q_{2,i}^{(k)} \right) +\underline{M}_{F}x_{N,u}-\underline{M}_{F},\\
y_u &\ge \left( \sum_{i=1}^{m_2} F_{3,i,u} q_{2,i}^{(k)} \right) +\overline{M}_{F}x_{N,u}-\overline{M}_{F}.
\end{align}
\end{subequations}
{Note that $y_u = \sum_{i=1}^{m_2} F_{3,i,u} q_{2,i}^{(k)}$ if $x_{N,u} = 1$ and $y_u = 0$ otherwise, i.e., $y_u = \vx_{N,u} \left( \sum_{i=1}^{m_2} F_{3,i,u} q_{2,i}^{(k)} \right)$.} Then, $(\vF_3\vx_N)^T \vq_2^{(k)}$ in \eqref{eq: gale_ineq} can be replaced by $\sum_{u\in V}y_u$ subject to \eqref{eq: McCormick_relaxation}. \looseness=-1

AODC guarantees an optimal solution at convergence in the considered special case (see proof in Appendix~\ref{appendix: additional_proof}). 

\begin{theorem}\label{thm:Optimality of AODC}
If $\xi_c = \infty$ and $\psi_s(\vp, \vD)$ returns the feasible region of \eqref{eq:reformulate_OPF}, then AODC will converge in finite time to an optimal solution to \eqref{eq:PMU_placement}.
\end{theorem}

Although in the worst case AODC may still enumerate all the attack pairs, which can be exponential in $|E|$, we have observed that in practice it usually converges after identifying a relatively small set of \quotes{typical attack pairs}, as shown in Table~\ref{tab: computational_iterations}.\looseness=-1

\subsection{Efficient Heuristics} 

Although Alg.~\ref{alg: alter_opt_v1} is guaranteed to find the optimal solution, the computational complexity 
can grow exponentially with the network size due to the requirement of solving MILPs in each iteration, which motivates us to develop polynomial-time heuristics. A scenario of particular interest is when $\xi_p$ is small, i.e., $\xi_p=\mathcal{O}(1)$. In this case, the total number of attack pairs is polynomial in $|E|$, and thus the number of iterations in AODC and the complexity of computing a new attack pair in each iteration are both polynomial in $|E|$. Our focus in this case is thus on solving \eqref{eq: vcg_formulation} approximately in polynomial time.

\emph{Relaxation:}
One idea is to directly relax the MILP version of \eqref{eq: vcg_formulation} into an LP. However, simple LP relaxation will not work:\looseness=-1
\begin{enumerate}
    \item The LP relaxation will invalidate the McCormick relaxation \eqref{eq: McCormick_relaxation} for the bilinear term $(\vF_3\vx_N)^T \vq_2^{(k)}$.
    \item The feasible region is significantly extended by the LP relaxation due to the adopted big-M modeling technique. 
    \item Given a continuous solution $\tilde{\vbeta}$ 
    obtained from the LP relaxation, it is non-trivial to determine which subset of $\Omega(\tilde{\vbeta})$, if any, can achieve our defense goal.
\end{enumerate}

We have developed a polynomial-time heuristic that can find a better PMU placement. The core of our heuristic is a different ``LP relaxation'' of \eqref{eq: vcg_formulation}. 
Recall that the main challenge in directly relaxing the MILP version of \eqref{eq: vcg_formulation} is the invalidation of \eqref{eq: McCormick_relaxation} for linearizing $(\vF_3\vx_N)^T \vq_2^{(k)}$. 
To overcome this issue, we make the following observation (see proof in Appendix~\ref{appendix: additional_proof}): 
\begin{lemma}\label{lem: lp_relax_attack_denial}
Define $\vLambda_{x,p}, \vLambda_{x,n} \in \{0,1\}^{|V|\times m_2}$ such that $(\vLambda_{x,p} \vq_2)_u$ is the dual variable for \eqref{eq: constr_theta2_pmu_p} and $(\vLambda_{x,n} \vq_2)_u$ is the dual variable for \eqref{eq: constr_theta2_pmu_n}. Suppose that the linear system
\begin{subequations}\label{eq: attack_denial_lp}
\allowdisplaybreaks
\begin{align}
&(\vF_1^{(k)})^T \vq_1^{(k)} + (\vF_2^{(k)})^T \vq_2^{(k)} = \bm{0},\\
&(\vs_1^{(k)})^T\vq_1^{(k)} + (\vs_2^{(k)} + \vF_3)^T\vq_2^{(k)} \le -1, \label{eq: lp_gale_ineq} \\
&(\vLambda_{x,p}+\vLambda_{x,n}) \vq_2 \le M_{q}\underline{\vA} \vbeta, \label{eq: q2_beta_relationship} \\
&\vq_2^{(k)} \ge \bm{0},~ \bm{1}\ge  \vbeta \ge \bm{0}
\end{align}
\end{subequations}
for attack pair $(\va_p^{(k)},e^{(k)})$ is feasible under $\vbeta = \check{\vbeta}$, where $M_{q}$ is a large constant (defined in Appendix~\ref{appendix:Calculation of Big-M}).  
Then, $\vbeta = \ceil{\check{\vbeta}}$ satisfies \eqref{eq: const_beta_x_N}--\eqref{eq: const_beta_x_L} and \eqref{eq: attack_denial_constr} with $w_{a,k} = 0$ for the attack pair $(\va_p^{(k)},e^{(k)})$. \looseness=-1
%
\end{lemma}

Lemma~\ref{lem: lp_relax_attack_denial} suggests that given an attack pair $(\va_p^{(k)},e^{(k)})$, we can relax the mixed integer \quotes{Attach-Denial} constraints \eqref{eq: attack_denial_constr} into the linear constraints \eqref{eq: attack_denial_lp} and round up the fractional solution to obtain a valid PMU placement, which is guaranteed to prevent the given attack pair from forming successful attack tuples.
{According to Gale's theorem of alternative, $\lrbrackets{(\vLambda_{x,p}+\vLambda_{x,n}) \vq_2^{(k)}}_u > 0$ only if at least one of \eqref{eq: constr_theta2_pmu_p} and \eqref{eq: constr_theta2_pmu_n} is \emph{effective} for making \eqref{eq: LP_given_pair} infeasible\footnote{We say that an inequality in \eqref{eq: LP_given_pair} is \emph{effective} for making \eqref{eq: LP_given_pair} infeasible if removing it will change the feasibility of \eqref{eq: LP_given_pair}.}. Since \eqref{eq: constr_theta2_pmu_p}-\eqref{eq: constr_theta2_pmu_n} is effective if and only if $x_{N,u} = 1$ (under the constraint of $x_{N,u} \in \{0,1\}$), we use $(\vLambda_{x,p}+\vLambda_{x,n}) \vq_2^{(k)}$ as a proxy of $\vx_N$ 
in Lemma~\ref{lem: lp_relax_attack_denial}. } \looseness=-1


Lemma~\ref{lem: lp_relax_attack_denial} motivates us to formulate the following LP based on a given set $\mathcal{C}$ 
of infeasible PMU placements and a given set $\{(\va_p^{(k)}, e^{(k)})\}_{k=1}^K$ of attack pairs:
\begin{mini!}|s|<b>
{}{ \sum_{u\in V}\beta_u }
{\label{eq: vcg_LP_relax}}
{}
\addConstraint{ \eqref{eq: attack_denial_lp} \mbox{ for } k=1,\cdots,K }{}\label{eq: constr_q_beta}
\addConstraint{\sum_{i:\hat{\beta}_i=0}\beta_i \ge 1, \forall \hat{\vbeta} \in \mathcal{C},}{}\label{eq: constr_no_good_relax}
\end{mini!}
where \eqref{eq: constr_q_beta} models relaxed \quotes{Attack-Denial} constraints and \eqref{eq: constr_no_good_relax} models relaxed \quotes{No-Good} constraints. In this sense, \eqref{eq: vcg_LP_relax} is a \quotes{LP relaxation} of \eqref{eq: vcg_formulation}.
However, instead of directly computing a PMU placement from \eqref{eq: vcg_LP_relax} which still faces some of the issues for simple LP relaxation, our idea is to use the result of \eqref{eq: vcg_LP_relax} to identify important nodes for PMU placement to defend against the given attack pairs in the case of $w_{a,k} = 0$ in \eqref{eq: attack_denial_constr}. We will account for the case of $w_{a,k} = 1$ separately in the proposed algorithm to avoid scaling and numerical issues.
\begin{algorithm}\label{alg: heuristic_3_phase}
\SetAlgoLined
\SetKwFunction{Fmain}{FailEdgeDetection}
\SetKwInOut{Input}{input}\SetKwInOut{Output}{output}
\tcc{Phase-1: find a set $\mathcal{A}_0$ of attack pairs}
\textbf{Initialization: }$k=1$,
$\hat{\vbeta}^{(k)} = \bm{0}$, $\mathcal{A}_0 = \emptyset$, $\mathcal{C} = \emptyset$\;
\While{$\psi_a(\hat{\vbeta}^{(k)}) > 0$}{
    $\mathcal{A}_0\leftarrow \mathcal{A}_0\cup \{(\va_p^{(k)}, e^{(k)})\}$, where $(\va_p^{(k)}, e^{(k)})$ is obtained by solving \eqref{eq:reformulate_attacker} under $\hat{\vbeta}^{(k)}$\; \label{Line: ph1_atttacker_v1}
    $\mathcal{C} \gets \mathcal{C} \cup \{\hat{\vbeta}^{(k)}\}$, $k\gets k+1$\;
    obtain {$\check{\vbeta}^{(k)}$} by solving \eqref{eq: vcg_LP_relax} over $\mathcal{C}$ and $\mathcal{A}_0$\;\label{Line: ph1_vcg_lp}
    Rounding: $\hat{\vbeta}^{(k)} \gets \ceil{\check{\vbeta}^{(k)}}$\;
}
\tcc{Phase-2: find candidate placements $\{\Omega_i\}_{i=1}^{K_{c}}$ to defend against $\mathcal{A}_0$}
Set $\Omega_i := \{u_i\}, i=1,\cdots,K_{c}$, where $\{u_i\}_{i=1}^{K_{c}}$ are the indices of the largest $K_c$ elements of $\check{\vbeta}^{(k)}$ {that is obtained in the last iteration of phase-1}\;
$\{\Omega_i\}_{i=1}^{K_{c}}, \mathcal{C} \gets \mbox{UpdateCandidate}\lrbrackets{\{\Omega_i\}_{i=1}^{K_{c}}, \mathcal{A}_0, \mathcal{C}}$\;\label{Line: init_candidates_lp}
\tcc{Phase-3: augment $\{\Omega_i\}_{i=1}^{K_{c}}$ to find a placement $\Omega$ with $\psi_a\left(\vbeta(\Omega)\right) = 0$}
\While{True}{
    $\mathcal{A} \gets \emptyset$\;
    \For{$i\gets1$ \KwTo $K_{c}$}{
    \lIf{$\psi_a\lrbrackets{\vbeta(\Omega_i)} > 0$ \label{Line: attack_generate_phase3}}
    {Generate $(\va_p^{(i)}, e^{(i)})$ and $\mathcal{A} \gets \mathcal{A}\cup (\va_p^{(i)}, e^{(i)})$}
    \lElse{Return $\Omega^* = \arg\min_{\Omega_j: \psi_a(\vbeta(\Omega_j))=0} |\Omega_j|$ if $|\Omega^*| \le 1+ \min_{\Omega_j: \psi_a(\vbeta(\Omega_j))>0} |\Omega_j|$}
    }
    $\{\Omega_i\}_{i=1}^{K_{c}}, \mathcal{C} \gets \mbox{UpdateCandidate}\lrbrackets{\{\Omega_i\}_{i=1}^{K_{c}}, \mathcal{A}, \mathcal{C}}$\; \label{Line: uc_A3}
}
\caption{3-phase Secured PMU Placement}
\end{algorithm}
\emph{Algorithm:}
The details of the proposed heuristic is given in Alg.~\ref{alg: heuristic_3_phase}, which relies on the function \emph{UpdateCandidate($\cdot$)} shown in Alg.~\ref{alg: primitive_find_candidate}. The logic behind the heuristic is similar to that in AODC, i.e., iteratively updating PMU placements 
based on newly found attack pairs. The questions are: (\romannumeral1) how to generate initial placements, (\romannumeral2) how to find attack pairs that can cause outages under given placements, and (\romannumeral3) how to update the given placements 
to defend against the newly found attack pairs, all in polynomial time. Since this algorithm is designed for the case of $\xi_p=\mathcal{O}(1)$, under which question (\romannumeral2) is easily solvable, our focus will be on questions (\romannumeral1) and (\romannumeral3).

We answer question (\romannumeral1) in two phases. Specifically, in \emph{phase-1}, we iteratively find a set of attack pairs $\mathcal{A}_0$ such that solving \eqref{eq: vcg_LP_relax} over $\mathcal{A}_0$ leads to a fractional solution $\check{\vbeta}$ with $\psi_a\lrbrackets{\ceil{\check{\vbeta}}} = 0$. Then in \emph{phase-2}, we search for a set of candidate PMU placements $\{\Omega_i\}_{i=1}^{K_{c}}$ to defend against $\mathcal{A}_0$ in the hope that $|\Omega_i| < |\Omega(\ceil{\check{\vbeta}})|$. The motivation for maintaining $K_{c} > 1$ candidates is to avoid the situation where the computed placement is effective in defending against the given attacks but ineffective for other attacks.

We answer (\romannumeral3) in Alg.~\ref{alg: primitive_find_candidate}, which iteratively augments a given set of candidate placements $\{\Omega_i\}_{i=1}^{K_{c}}$ to defend against a given set $\mathcal{A}$ of attack pairs. For each candidate placement not effective against all the attack pairs in $\mathcal{A}$, Alg.~\ref{alg: primitive_find_candidate} will generate $K_L$ and $K_A$ new candidate placements in Line~\ref{line: alg3_candidate_phy_attack} and Lines~\ref{line: alg3_sol_LP}-\ref{line: alg3_up_Q}, respectively. Then, Line~\ref{line: alg3_new_can_2} will select the $K_{c}$ placements most effective in defending against the attack pairs in $\mathcal{A}$ from the pool of $K_{c}\cdot (K_A + K_L)$ candidate placements. We now characterize the complexity of Alg.~\ref{alg: heuristic_3_phase} (see proof in Appendix~\ref{appendix: additional_proof}). 
\looseness=-1

\begin{theorem}\label{theo: polynomial_time_heuristic}
If $\xi_p = \mathcal{O}(1)$, then the complexity of Alg.~\ref{alg: heuristic_3_phase} is polynomial in $|V|$, $|E|$, and $K_c$.
\end{theorem}

\begin{algorithm}\label{alg: primitive_find_candidate}
\SetAlgoLined
\textbf{Initialization: }$\mathcal{A}_i = \mathcal{A}, i=1,\cdots, K_c$\;
\While{$\exists i$ such that $\mathcal{A}_i \neq \emptyset$}{
    $\mathcal{Q} \gets \emptyset$\;
    \For{$i\gets 1$ \KwTo $K_{c}$\label{line: alg3_for_LP}}{
    \lIf{$\mathcal{A}_i=\emptyset$}{$\mathcal{Q} \gets \mathcal{Q} \bigcup \{\Omega_i\}$ and continue}
    \lElse{ $\mathcal{C} \gets \mathcal{C} \cup \{ \vbeta(\Omega_i) \}$ }
    $\mathcal{Q} \gets \mathcal{Q} \bigcup ( \Omega_i \cup \{ v_j\} )$ for $j=1,\cdots,K_L$, where $v_j$ can prevent the j-th most physical attacks  in $\mathcal{A}_i$ \label{line: alg3_candidate_phy_attack}\;
    Solve \eqref{eq: vcg_LP_relax} over $\mathcal{A}$, $\mathcal{C}$, and the constraints $\beta_u = 1, \forall u\in \Omega_i$, which results in $\check{\vbeta}$\;\label{line: alg3_sol_LP}
    $\mathcal{Q} \gets \mathcal{Q} \bigcup ( \Omega_i \cup \{ u_j \} )$ for $j=1,\cdots,K_A$, where $u_j$ is the index of the $j$-th largest element in $\{\check{\beta}_u\}_{u\in V\setminus \Omega_i}$\label{line: alg3_up_Q}\;
    }
    Update $\{\Omega_i\}_{i=1}^{K_{c}}$ as the $K_c$ elements in $\mathcal{Q}$ that can defend against the most attack pairs in $\mathcal{A}$\label{line: alg3_new_can_2}\;
    $\mathcal{A}_i \gets \{(\va_p, e) \in \mathcal{A} \arrowvert \Omega_i \mbox{ cannot defend against }$ $(\va_p, e)\}$, $\forall i=1,\ldots,K_c$\;
}
Return $\{\Omega_i\}_{i=1}^{K_{c}}$ and $\mathcal{C}$\;
\caption{UpdateCandidate\big($\{\Omega_i\}_{i=1}^{K_{c}}$, $\mathcal{A}$, $\mathcal{C}$\big)}
\vspace{-0em}
\end{algorithm}

{
\section{Extension to AC Power Flow Model}\label{sec:Extension to AC}
So far we have assumed the DC power flow approximation for the power grid given in Section~\ref{subsec:Power Grid Modeling}. It remains to validate the resulting PMU placement under the AC power flow model that describes the grid state more accurately.
To this end, we will address the following questions: given a PMU placement $\Omega_{\mbox{\tiny DC}} \subseteq V$ obtained under the DC power flow model,
 (i) how to test the feasibility of $\Omega_{\mbox{\tiny DC}}$ in preventing outages under the AC power flow model, and
(ii) how to refine $\Omega_{\mbox{\tiny DC}}$ if needed to achieve our defense goal under the AC power flow model. \looseness=-1 
}

{
\subsection{Testing a PMU Placement under AC Model}
One challenge to answer the first question is the nonlinear and nonconvex nature of \emph{AC power flow based SCED (AC-SCED)}, which invalidates the transformation of \eqref{eq:reformulate_attacker} into a single-level MILP through KKT conditions. Another challenge lies in formulating a single optimization to maximize the overloading of a target line after SCED (at $t_3$ in Fig.~\ref{fig: timeline_attack}). Specifically, since solving nonlinear AC power flow equations usually requires iterative methods (e.g., Newton-Raphson method \cite{monticelli2000electric}), we cannot directly formulate the AC-SCED at $t_2$ and the corresponding ground-truth grid state at $t_3$ in an optimization problem.
Existing works handled this challenge by approximating the grid state at $t_3$ by the DC power flow model \cite{chu2020vulnerability, liang2015vulnerability} or DC-based line outage distribution factors \cite{chung2018local, chu2021n}. However, such DC-based approximations cannot be directly used to compute the magnitude of currents, which determines the overloading and related tripping of lines.
}


{In the following, we provide a method, as shown in Alg.~\ref{alg: check_AC}, to check the existence of an AC-based CCPA that can cause overloading 
under a given PMU placement. To overcome the challenges discussed before, we first remove the optimality requirement in AC-SCED, similar to our derivation of \quotes{Attack-Denial} constraints {in Section~\ref{subsec:AODC}}. Omitting this optimality requirement is equivalent to allowing the attacker to choose the objective for AC-SCED, which enlarges the feasible region for the attacker's optimization. To jointly model the current at $t_3$ and the AC-SCED at $t_2$, we adopt the \emph{linearized approximation of AC power flow equations} \cite{yang2018general}.}
{Based on these two strategies, we formulate the following optimization problem for the attacker to maximize the magnitude of current on a given target line $e_t$ under a given physical attack (i.e., $\va_p$): 
}
{\begin{maxi!}|s|<b>
{}{ |\hat{I}_{3,e_t}|^2 }
{\label{eq: ac_formulate_overload}}
{}
\addConstraint{\text{Constraints on } \tilde{\vv}_2, \tilde{\vtheta}_2 \text{ to bypass BDD}}{} \label{ac: fdi_t2}
\addConstraint{\text{ACOPF constraints on }\tilde{\vv}_3, \tilde{\vtheta}_3}{}\label{ac: t3_acopf}
\addConstraint{\text{Constraints to solve }\hat{\vv}_3, \hat{\vtheta}_3, |\hat{\vI}_3|,}{}\label{ac: pf_linearze_t3}
\end{maxi!}}
{where $\tilde{\vv}_2, \tilde{\vtheta}_2$ denote the voltage magnitudes and phase angles estimated at $t_2$ by the control center based on falsified measurements, $\tilde{\vv}_3, \tilde{\vtheta}_3$ denote the same variables predicted by AC-SCED for $t_3$ (computed at $t_2$), and $\hat{\vv}_3, \hat{\vtheta}_3, |\hat{\vI}_3|$ denote the approximated ground-truth of voltage magnitudes, phase angles and line current magnitude at $t_3$.}
{The details of \eqref{eq: ac_formulate_overload} are given in Appendix~\ref{appendix: ac_attacker_detail}. Similar to Table~\ref{tab: notation_timeline}, for a given variable $x$, we use $\tilde{x}_{2}$ to denote its estimate based on falsified measurements at $t_2$, $x_2$ to denote its ground-truth value at $t_2$, $\tilde{x}_{3}$ to denote the value predicted by AC-SCED (at $t_2$) for $t_3$, and $x_3$ to denote the ground-truth value at $t_3$. Given the voltage magnitudes $\tilde{\vv}_3$ and the phase angles $\tilde{\vtheta}_3$, the approximated values of $x$ at $t_3$ is denoted as $\hat{x}_3$.}

{
In \eqref{eq: ac_formulate_overload}, we have the following three types of constraints and decision variables:
\begin{enumerate}
    \item Constraint \eqref{ac: fdi_t2} is the counterpart of \eqref{constr: f2_fdi} under the AC power flow model, in which the main decision variables are $\tilde{\vv}_2$ and $\tilde{\vtheta}_2$. Similar to \eqref{constr: f2_fdi},  we use $\tilde{\vv}_2$ and $\tilde{\vtheta}_2$ as the decision variables to model the cyber attack that can bypass the BDD. Following \cite{chung2018local}, we adopt the quadratic convex (QC) relaxation \cite{coffrin2015qc} in \eqref{ac: fdi_t2} to model the AC power flow equations.
    \item As the counterpart of \eqref{eq: pg_theta_3_fdi}-\eqref{eq: pd_theta_3_fdi} under the AC power flow model, \eqref{ac: t3_acopf} models the reaction of AC-SCED to the falsified measurements based on the QC relaxation.
    \item The real grid state at $t_3$ is formulated in \eqref{ac: pf_linearze_t3} as the counterpart of \eqref{eq:f_3, a_p}-\eqref{eq:theta_3, f_3}, based on the approximation of AC power flow equations proposed in \cite{yang2018general}.
\end{enumerate}
As we have enlarged the feasible region for the attacker in \eqref{ac: fdi_t2}-\eqref{ac: t3_acopf} by using the QC relaxation, \eqref{eq: ac_formulate_overload} models a stronger attack, and hence a PMU placement that prevents overloading under this attack can prevent overloading under the original attack. We will use $x^*$ to denote the value of decision variable $x$ in the optimal solution to \eqref{eq: ac_formulate_overload}.
}

{
Based on \eqref{eq: ac_formulate_overload}, we develop an algorithm to check the feasibility of a PMU placement $\Omega\subseteq V$ in preventing outages under AC-based CCPA, shown in Algorithm~\ref{alg: check_AC}. Specifically, at Lines~\ref{Line: ac_I2}, we compute $\vv_2, \vtheta_2, |\vI_2|$ by solving power flow equations. Thus, the counterpart of \eqref{eq: dc_t2_gt} is no longer needed to compute the real grid states after physical attacks.
Then, at Line~\ref{Line: solve_attacker_ac}, we obtain the optimal solution ($|\hat{I}_{3,e_t}^*|$, $\tilde{\vv}_3^*, \tilde{\vtheta}_3^*$) to \eqref{eq: ac_formulate_overload} for the given attack pair $(\va_p,e_t)$ (recall that $|\hat{I}_{3,e_t}^*|$ is the approximated current magnitude on line $e_t$ at time $t_3$ while $|I_{3,e_t}^*|$ is the corresponding true value). 
Alg.~\ref{alg: check_AC} considers the PMU placement $\Omega$ to successfully defend against $(\va_p,e_t)$ (i.e., preventing overloading at line $e_t$ under physical attack $\va_p$) if one of the following conditions hold:
\begin{enumerate}
    \item no cyber attack $\va_c$ can bypass the BDD, i.e., \eqref{eq: ac_formulate_overload} is infeasible, as checked in Line~\ref{Line: infeasible_fdi}, or
    \item $|\hat{I}_{3,e_t}^*| \le \hat{I}_{max,e_t}$ and $|I_{3,e_t}^*| \le \gamma_e I_{max,e_t}$, as checked in Lines~\ref{Line: if_linear_no_overload}--\ref{Line: continue_no_of}, where $\hat{I}_{max,e_t}$ (derived in Theorem~\ref{theorem: ac_thr_set_linear_guarantee}) is the threshold used by Alg.~\ref{alg: check_AC} to detect the tripping of line $e_t$ based on the approximated current magnitude $|\hat{I}_{3,e_t}^*|$.
\end{enumerate}
}

{
The use of $\hat{I}_{max,e}$ rather than $\gamma_e I_{max,e}$ allows us to compensate for the approximation error at $t_3$. As stated in Theorem~\ref{theorem: ac_thr_set_linear_guarantee}, under a properly-set $\hat{I}_{max,e}$, a PMU placement $\Omega$ is guaranteed to achieve our defense goal under the AC model if $\Omega$ can pass the test of Alg.~\ref{alg: check_AC}, i.e., no overloading is reported.
How to bound the approximation errors as assumed in Theorem~\ref{theorem: ac_thr_set_linear_guarantee} is not the focus of this work; we refer interested readers to \cite{yang2018general} for details.
\begin{theorem}\label{theorem: ac_thr_set_linear_guarantee}
Assume that the approximation used in \eqref{ac: pf_linearze_t3} satisfies $|\hat{v}_{3,u} - v_{3,u}| \le \epsilon_{v,u}, |\hat{\theta}_{3,u} - \theta_{3,u}| \le \epsilon_{\theta,u}$, $\forall u \in V$ and $|\hat{p}_{3,f,e} - p_{3,f,e}| \le \epsilon_{p,e}$, $|\hat{q}_{3,f,e} - q_{3,f,e}| \le \epsilon_{q,e}$, $\forall e \in E$. Then, there exists $\epsilon_{I,e}, \forall e\in E$ {(see proof in Appendix~\ref{appendix: additional_proof} for details)} and $\hat{I}_{max,e}:=\gamma_e I_{max,e} - \epsilon_{I,e}$ such that any PMU placement passing the test of Alg.~\ref{alg: check_AC} can {prevent overload-induced tripping under the AC power flow model}.
\end{theorem}
}
\begin{algorithm}\label{alg: check_AC}
\SetAlgoLined
\SetKwFunction{Fmain}{FailEdgeDetection}
\SetKwInOut{Input}{input}\SetKwInOut{Output}{output}
\For{each possible attack pair $(\va_p,e_t)$ {under the given PMU placement $\Omega$}}{
    Obtain $\vv_2, \vtheta_2, |\vI_2|$ from AC power flow equations\; \label{Line: ac_I2}
    Solve \eqref{eq: ac_formulate_overload} to obtain $|\hat{I}_{3,e_t}^*|$, $\tilde{\vv}_3^*, \tilde{\vtheta}_3^*$\;\label{Line: solve_attacker_ac}
    \uIf{\eqref{eq: ac_formulate_overload} is feasible AND $|\hat{I}_{3,e_t}^*| \leq \hat{I}_{max,e_t}$\label{Line: if_linear_no_overload}}{ Compute $|I_{3,e_t}^*|$ from AC power flow equations\; \label{Line: ac_pf_t3}
    \uIf{$|I_{3,e_t}^*| \leq \gamma_e I_{max,e_t}$}{Continue\label{Line: continue_no_of}\;}\lElse{Terminate and report overloading}}
    \uElseIf{\eqref{eq: ac_formulate_overload} is infeasible\label{Line: infeasible_fdi}}{Continue\;}
    \lElse{Terminate and report overloading}
}
\caption{Test Feasibility of $\Omega$ under AC Model 
}
\end{algorithm}

{
\subsection{Refining PMU Placement}
In the case that the DC-based PMU placement $\Omega_{\mbox{\tiny DC}}$ fails the test by Alg.~\ref{alg: check_AC}, we provide a simple heuristic to augment it into a new placement $\Omega_{\mbox{\tiny AC}}$ that can achieve our defense goal under the AC model. The intuition is to iteratively augment $\Omega_{\mbox{\tiny DC}}$ by placing more PMUs until the resulting placement $\Omega_{\mbox{\tiny AC}}$ can pass the test of Alg.~\ref{alg: check_AC}. The key question is which node to add. To answer this question, we first augment $\Omega_{\mbox{\tiny DC}}$ into a PMU placement $\Omega_C := \Omega(\vbeta_C)$ that can achieve full observability 
by solving \eqref{eq: full_ob_augment_Omega_c}: \looseness=-1
\begin{mini!}|s|<b>
{\vbeta_C\in \{0,1\}^{|V|}}{ \lp{\vbeta_C}_1 }
{\label{eq: full_ob_augment_Omega_c}}
{}
\addConstraint{ \vbeta_C\ge \vbeta(\Omega_{\mbox{\tiny DC}}) }{}\label{constr: containt_Omega_DC}
\addConstraint{\underline{\vA}\vbeta_C \ge 1,}{}\label{constr: full_ob_Omega_c}
\end{mini!}
where \eqref{constr: containt_Omega_DC} guarantees $\Omega_{\mbox{\tiny DC}}\subseteq\Omega_{C}$, and \eqref{constr: full_ob_Omega_c} forces $\Omega_C$ to achieve full observability.
Then equipped with $\Omega_C$, 
we augment $\Omega_{\mbox{\tiny DC}}$ into $\Omega_{\mbox{\tiny AC}}$ by Alg.~\ref{alg: ac_augment_pmu_dc}. If a PMU placement cannot defend against an attack pair $(\va_p,e_t)$ (Line~\ref{Line: overload_detected}), then we update the PMU placement by the following rules:
\begin{enumerate}
    \item If there exists a node $u\in \Omega_C$ that can prevent the physical attack $\va_p$ as in \eqref{eq: beta_on_ap}, we add node $u$ to the current PMU placement (Line~\ref{Line: update_Omega_prevent_ap}).
    \item Otherwise, we add the node {in $\Omega_C$} with the maximum deviation in phase angle due to false data injection (Line~\ref{Line: u_max_VA_deviate}), with ties broken arbitrarily.
\end{enumerate}
}
\begin{algorithm}\label{alg: ac_augment_pmu_dc}
\SetAlgoLined
\SetKwFunction{Fmain}{FailEdgeDetection}
\SetKwInOut{Input}{input}\SetKwInOut{Output}{output}
\textbf{Initialization: }$\Omega_{\mbox{\tiny AC}}=\Omega_{\mbox{\tiny DC}}$\;
\While{True}{
Test $\Omega_{\mbox{\tiny AC}}$ through Alg.~\ref{alg: check_AC}\;
\lIf{No overloading is reported}{Return $\Omega_{\mbox{\tiny AC}}$}
\Else
{
    Let $(\va_p,e_t)$ be the attack pair under which overloading is reported, and $U := \{u\in\hspace{-.1em} V\hspace{-.1em}: \exists e \mbox{ with } a_{p,e}=1, D_{u,e}\neq 0 \}$ (all end-nodes of physically-attacked lines)\label{Line: overload_detected}\;
    \uIf{$\Omega_C\cap U \neq \emptyset$}{Arbitrarily choose a node $u\in \Omega_C\cap U$\label{Line: update_Omega_prevent_ap}\;}
    \uElse{Let $\tilde{\vtheta}_2, \vtheta_2$ be the falsified/true phase angles at $t_2$ under attack pair $(\va_p,e_t)$\;
    Set $u:=\arg\max_{v\in \Omega_C} |\tilde{\theta}_{2,v}-\theta_{2,v}|$\label{Line: u_max_VA_deviate}\;}
    $\Omega_{\mbox{\tiny AC}}\leftarrow \Omega_{\mbox{\tiny AC}}\cup \{u\}$\;
}
}
\caption{Augment PMU Placement for AC Model}
\end{algorithm}

\section{Numerical Experiments}\label{sec: evaluations}
\emph{Simulation Settings:} We evaluate our solution against benchmarks in several standard systems: IEEE 30-bus, IEEE 57-bus, IEEE 118-bus, and IEEE 300-bus system, where the system parameters as well as load profiles are obtained from \cite{babaeinejadsarookolaee2019power}. {The parameters for our evaluation are set as follows unless specified otherwise:} 
We set $\alpha = 0.25$ according to \cite{che2018false}. We allow 
$\tilde{\vtheta}_3$ to take any value specified by the attacker subject to \eqref{eq: cc_load_meet}-\eqref{eq: cc_pgp}, which makes our defense effective under any SCED cost vector.
The attacker's capability is set as $\xi_p = 2$, $\xi_c = \infty$ (no constraint on the number of manipulated meters). We set the overload-induced tripping threshold to $\gamma_e = 1.2, \forall e \in E$, which is slightly smaller than the one used in \cite{che2018false} to make the solution more robust. For Alg.~\ref{alg: heuristic_3_phase}, we set $K_c = K_A = K_L = 10$.

In the rest of this section, we will compare the performance of Alg.~\ref{alg: alter_opt_v1} (AONG or AODC) and Alg.~\ref{alg: heuristic_3_phase} with the following benchmarks: (\romannumeral1) PMU placement to achieve full observability as proposed in \cite{chakrabarti2008placement}; (\romannumeral2) greedily placing PMUs in the descending order of  node degrees until attack-induced overload-induced tripping is prevented, referred to as \quotes{GreedyDegree}. Benchmark (i) represents the current approach, and benchmark (ii) represents a baseline solution under the lowered goal of defense. \looseness=-1


\emph{Savings in the Number of PMUs:} 
In  Table~\ref{tab: res_matpower_pmu}, we compare the number of secured PMUs required by the proposed algorithms (Alg.~\ref{alg: alter_opt_v1}, Alg.~\ref{alg: heuristic_3_phase}) with the benchmarks
under the nominal operating point~\cite{babaeinejadsarookolaee2019power}. The minimum number of PMUs required to avoid outages, given by Alg.~\ref{alg: alter_opt_v1} (either AONG or AODC), 
is significantly smaller than what is required to achieve full observability. Alg.~\ref{alg: heuristic_3_phase} closely approximates the minimum for the tested systems, but a simple heuristic such as GreedyDegree does not.  For IEEE 300-bus system, we have skipped Alg.~\ref{alg: alter_opt_v1} as neither AODC nor AONG can converge within $72$ hours. The details of PMU locations are given in Appendix~\ref{appendix: pmu_locations}. \looseness=-1

\begin{table}[htbp]
\small
\vspace{-1em}
\centering
\caption{Comparison of the Required Number of PMUs }
\begin{tabular}{|c| c |c|c |c |}
\hline
                   & 30-bus & 57-bus & 118-bus & 300-bus \\
                   \hline
Alg.~\ref{alg: alter_opt_v1} (optimal)         & 2      & 3      & 9 & --- \\
\hline
Alg.~\ref{alg: heuristic_3_phase}  &   2    &  3     & 10 & 31 \\
\hline
GreedyDegree  &   3    &   3    & 14 & 85 \\
\hline
Full observability & 10  & 17  & 32 & 87 \\
\hline
\end{tabular}
\vspace{-.5em}
\label{tab: res_matpower_pmu}
\end{table}

{
Then, we evaluate the scenario when the solution by PPOP is used as a temporary PMU placement that will eventually be augmented into a placement achieving full observability, as discussed at the end of Section~\ref{sec: problem_formulation} (Remark 2). To this end, we evaluate the following metrics: (\romannumeral1) the minimum number of PMUs required by PPOP $|\Omega_{\mbox{\tiny PPOP}}|$, (\romannumeral2) the minimum number of PMUs for achieving full observability $|\Omega_{\mbox{\tiny FO}}|$, (\romannumeral3) the size of a full-observability placement $\Omega_{C}$ augmented from $\Omega_{\mbox{\tiny PPOP}}$ given by \eqref{eq: full_ob_augment_Omega_c}, and (\romannumeral4) the size of the optimal solution $\Omega'_{\mbox{\tiny PPOP}}$ to a variation of PPOP with the additional constraint that $\Omega'_{\mbox{\tiny PPOP}} \subseteq\Omega_{\mbox{\tiny FO}}$. In Table~\ref{tab: fullob_matpower_pmu}, we observe that (i) $|\Omega'_{\mbox{\tiny PPOP}}|$ is only slightly larger than $|\Omega_{\mbox{\tiny PPOP}}|$, i.e., most of the cost savings by PPOP is still achievable when its solution is required to be consistent with the optimal long-term solution that achieves full observability, but (ii) $|\Omega_{C}|$ can be notably larger than $|\Omega_{\mbox{\tiny FO}}|$ for large systems, i.e., augmenting an arbitrary solution to PPOP to achieve full observability may require notably more PMUs compared to a clean-slate solution. }

\begin{table}[htbp]
\small
\vspace{-1em}
\centering
\caption{Comparison of \#PMUs under Temporary/Long-term Placement }
\begin{tabular}{|c| c |c|c |c |}
\hline
                   & 30-bus & 57-bus & 118-bus & 300-bus \\
                   \hline
$|\Omega_{\mbox{\tiny PPOP}}|$         & 2      & 3      & 9 & 31 \\
\hline
$|\Omega'_{\mbox{\tiny PPOP}}|$  &   2    &  3     & 10 & 34 \\
\hline
$|\Omega_{C}|$  &   10    &   17    & 33 & 95 \\
\hline
$|\Omega_{\mbox{\tiny FO}}|$ & 10  & 17  & 32 & 87 \\
\hline
\end{tabular}
\vspace{-.5em}
\label{tab: fullob_matpower_pmu}
\end{table}

\emph{Impact of System Parameters:}
We evaluate the impact of various system parameters on the number of PMUs required by PPOP, given by Alg.~\ref{alg: alter_opt_v1} (by Alg.~\ref{alg: heuristic_3_phase} for the 300-bus system).\looseness=0 

First, we study the effect of $\alpha$ introduced in \eqref{eq: attack_FDIrange}, where a larger $\alpha$ implies a larger feasible region for the attacker. 
It can be seen from Table~\ref{tab: res_varying_limits} that (i) PPOP can still significantly reduce the required number of PMUs compared to ``Full observability'' (see Table~\ref{tab: res_matpower_pmu}) even if $\alpha$ is large, and (ii) PPOP benefits from a small value of $\alpha$, which 
signifies the importance of precise load forecasting in defending against CCPA. 

\begin{table}[htbp]
\small
\vspace{-1em}
\centering
\caption{Number of PMUs in PPOP under varying $\alpha$}
\begin{tabular}{|c| c |c|c |c |}
\hline
& 30-bus & 57-bus & 118-bus & 300-bus \\
\hline
$\alpha = 0.01$         &  1     &    1   &   4  & 24  \\
\hline
$\alpha = 0.10$         &  1    &    2   &    6  & 30  \\
\hline
$\alpha = 0.25$         & 2      & 3      & 9  & 31 \\
\hline
$\alpha = 0.50$         & 3      &   3    &   11 & 34    \\
\hline
\end{tabular}
\vspace{-.5em}
\label{tab: res_varying_limits}
\end{table}

Then, we vary $\xi_p$ and $\xi_c$ to evaluate the impact of the attacker's capability. As shown in Figure~\ref{fig:xi_p xi_c}, (i) defending against a stronger attacker requires more PMUs as expected, (ii) PPOP still requires much fewer PMUs than ``Full observability'' when the attacker can disconnect multiple lines and manipulate all the meters (except for the secured PMUs), which is stronger than the attack model considered in \cite{tian2019multilevel, che2018false}, {and (\romannumeral3) PPOP can save a larger fraction of PMUs in IEEE 57-bus system since $\vf_{\tiny\text{max}}$ given in \cite{babaeinejadsarookolaee2019power} is large}. \looseness=-1
\begin{figure}
\begin{minipage}{.495\linewidth}
  \centerline{
  \includegraphics[width=1\columnwidth]{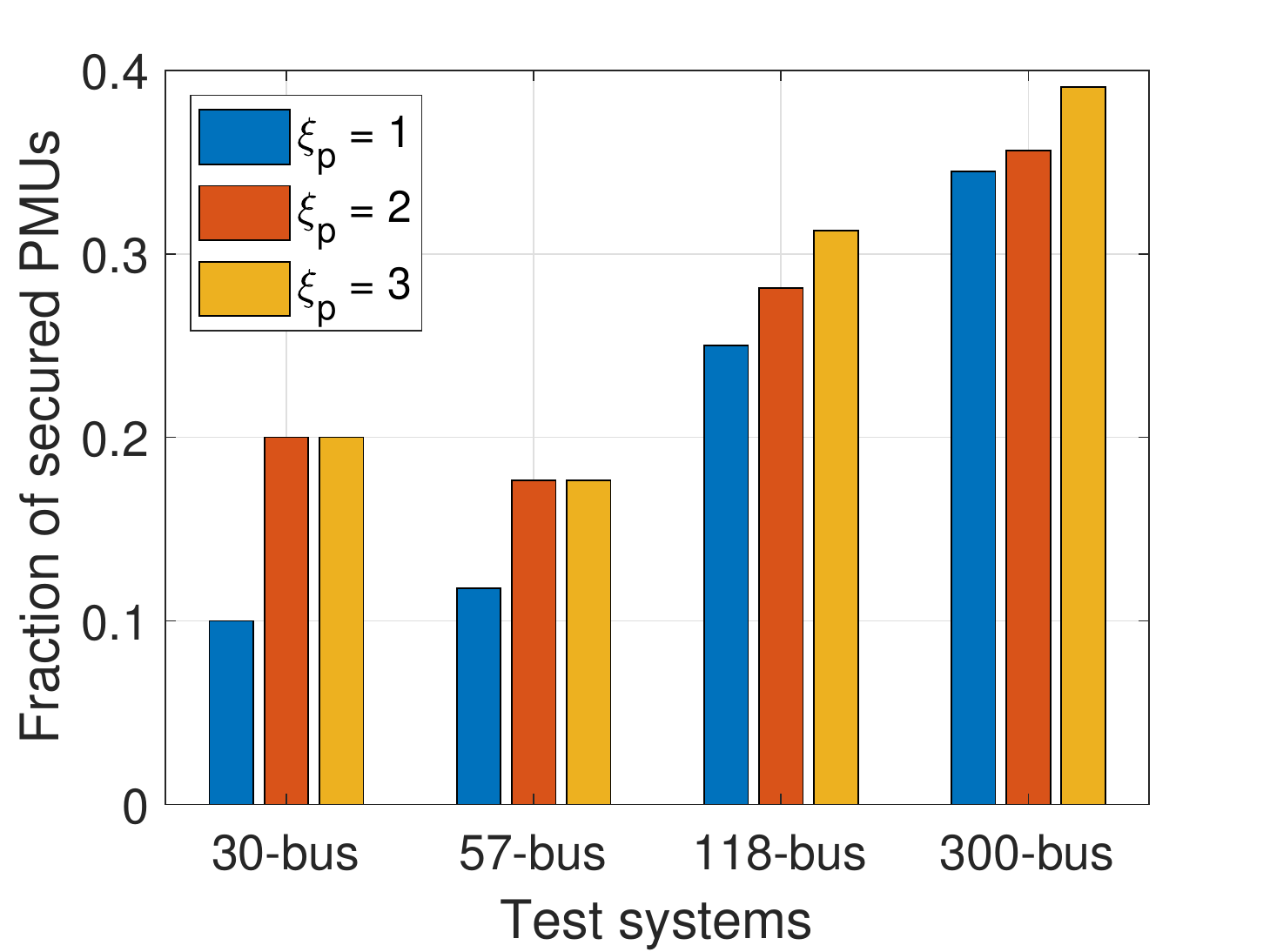}}
  \centerline{\small (a) Effect of $\xi_p$ }
\end{minipage}\hfill
\begin{minipage}{.495\linewidth}
 \centerline{
  \includegraphics[width=1\columnwidth]{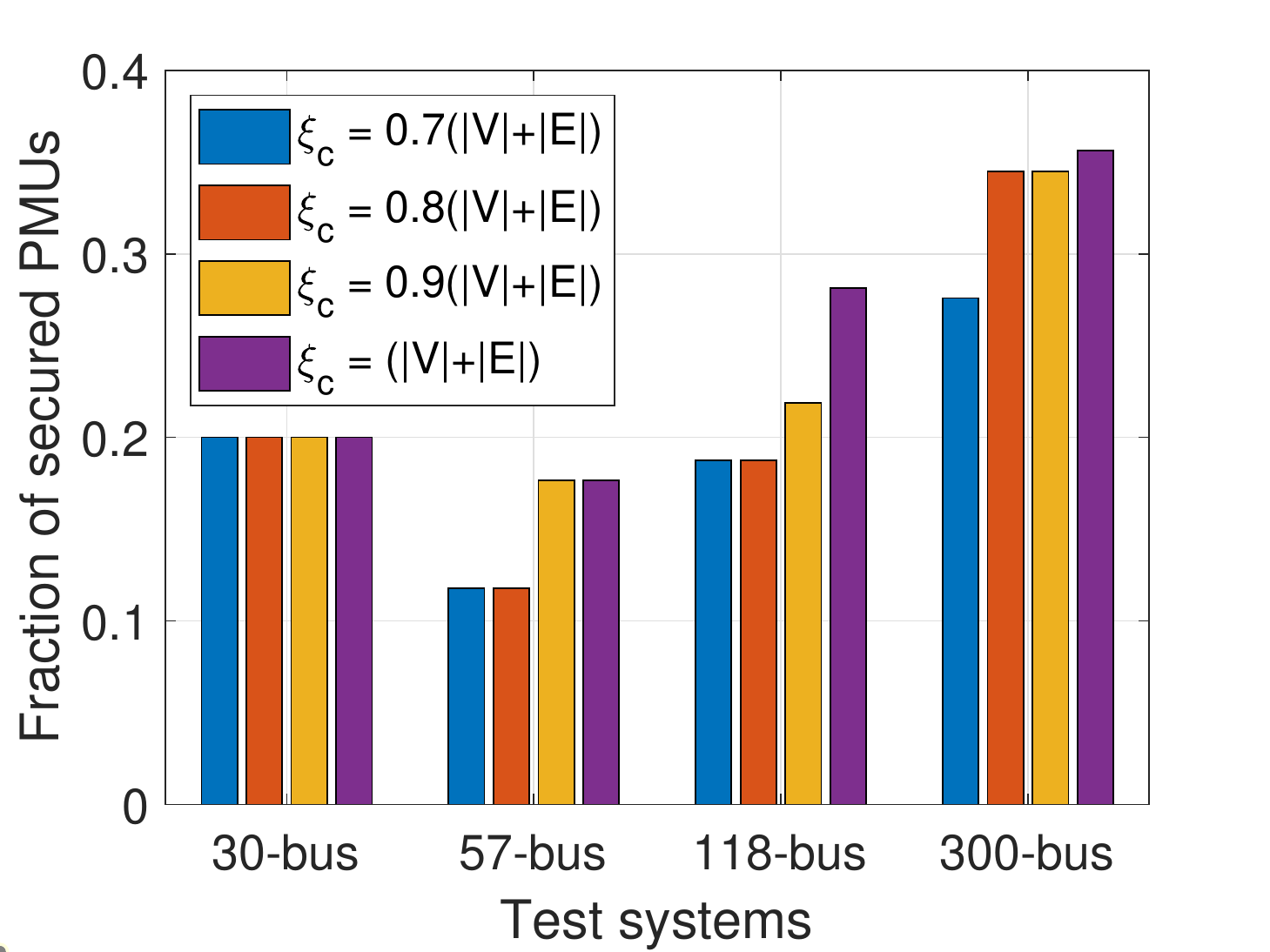}}
  \centerline{\small (b) Effect of $\xi_c$ }
\end{minipage}
  \caption{ $\frac{\text{\#PMUs required by PPOP}}{\text{\#PMUs required by full observability}}$ ($\xi_c\hspace{-.2em}=\hspace{-.2em}|V|\hspace{-.2em}+\hspace{-.2em}|E|$ means no $\xi_c$-constraint).}
  \label{fig:xi_p xi_c}
\vspace{-1.5em}
\end{figure}


In addition, we consider the case that the load profile $\vp_0$ can vary as shown in \eqref{eq: varying_load_profile_p0}. We assume $\vp_0 \in [\underline{\kappa}\vp^{(0)}, \overline{\kappa}\vp^{(0)}]$, where $\vp^{(0)}$ is the nominal load profile from \cite{babaeinejadsarookolaee2019power},  $\underline{\kappa} = 0.5$ and $\overline{\kappa}$ is set to the maximum value that keeps \eqref{eq:reformulate_OPF} feasible under $\overline{\kappa}\vp^{(0)}$. 
{In our evaluations, we set $\overline{\kappa}$ as $1.95, 2.69$, $2.41$ and $1.61$ for IEEE 30-bus, 57-bus, 118-bus and 300-bus systems, respectively.} {For the given range, PPOP requires 3, 4, 19, and 33 PMUs for the 30-bus, 57-bus, 118-bus, and 300-bus systems, which is more than what is required under a single load profile as expected. Nevertheless, PPOP can still save PMUs compared to \quotes{Full observability} as shown in Table~\ref{tab: res_matpower_pmu}.} \looseness=-1

\emph{Computational Efficiency:} We compare AODC and AONG in terms of the number of iterations ({which is also the number of examined attack pairs}) and the running time, which is evaluated in a platform with Intel i7-8700 CPU with Gurobi as the solver. 
Since any feasible solution to \eqref{eq: opt_max_nogood_but} can form an \quotes{No-Good} constraint, we set an upper-bound on the time for solving \eqref{eq: opt_max_nogood_but}, which is $1200$ seconds. As shown in Table~\ref{tab: computational_iterations}, while the two algorithms perform similarly for small systems, AODC converges notably faster for larger systems such as the 118-bus system thanks to its reduced solution space due to the adoption of both \quotes{No-Good} and \quotes{Attack-Denial} constraints. {Note that both algorithms converge after examining a small fraction of possible attack pairs (the total number of attack pairs is 33620, 252800, and 3200130 for these systems, respectively).}


\begin{table}[htbp]
\small
\vspace{-1em}
\centering
\caption{Number of iterations/Convergence time ($10^3$ sec) }
\begin{tabular}{|c| c |c|c |}
\hline
                   & 30-bus & 57-bus & 118-bus \\
                   \hline
AODC         &  8/0.021     &   3/2.188    &   16/26.64    \\
\hline
AONG  &   7/0.014 &   4/2.163    &  78/74.44   \\
\hline
\end{tabular}
\vspace{-.5em}
\label{tab: computational_iterations}
\end{table}

Moreover, we use IEEE 118-bus system as an example to demonstrate the trade-off in tuning the parameters $K_c, K_A, K_L$ for Alg.~\ref{alg: heuristic_3_phase} (assuming $K_A = K_L$). 
We run Alg.~\ref{alg: heuristic_3_phase} for $5$ times under each setting due to the randomness in solving \eqref{eq:reformulate_attacker} and breaking ties. 
The results are given in Fig.~\ref{fig: effect_Kc_KA}, where the bar denotes the mean and the error bar denotes the minimum/maximum. In Fig.~\ref{fig: effect_Kc_KA}~(b), we show the speedup of the heuristic compared to AODC in convergence time, i.e., $(\mbox{time of AODC})/(\mbox{time of heuristic})$.
We observe that (\romannumeral1) Alg.~\ref{alg: heuristic_3_phase} can return a good solution when $K_c\ge \%10\cdot|V|$ and $K_A=K_L\ge K_c$, and (\romannumeral2) under this configuration, Alg.~\ref{alg: heuristic_3_phase} is significantly faster than AODC at a small cost of requiring a couple of more PMUs.\looseness=-1 
\begin{figure}
\begin{minipage}{.495\linewidth}
  \centerline{
  \includegraphics[width=1\columnwidth]{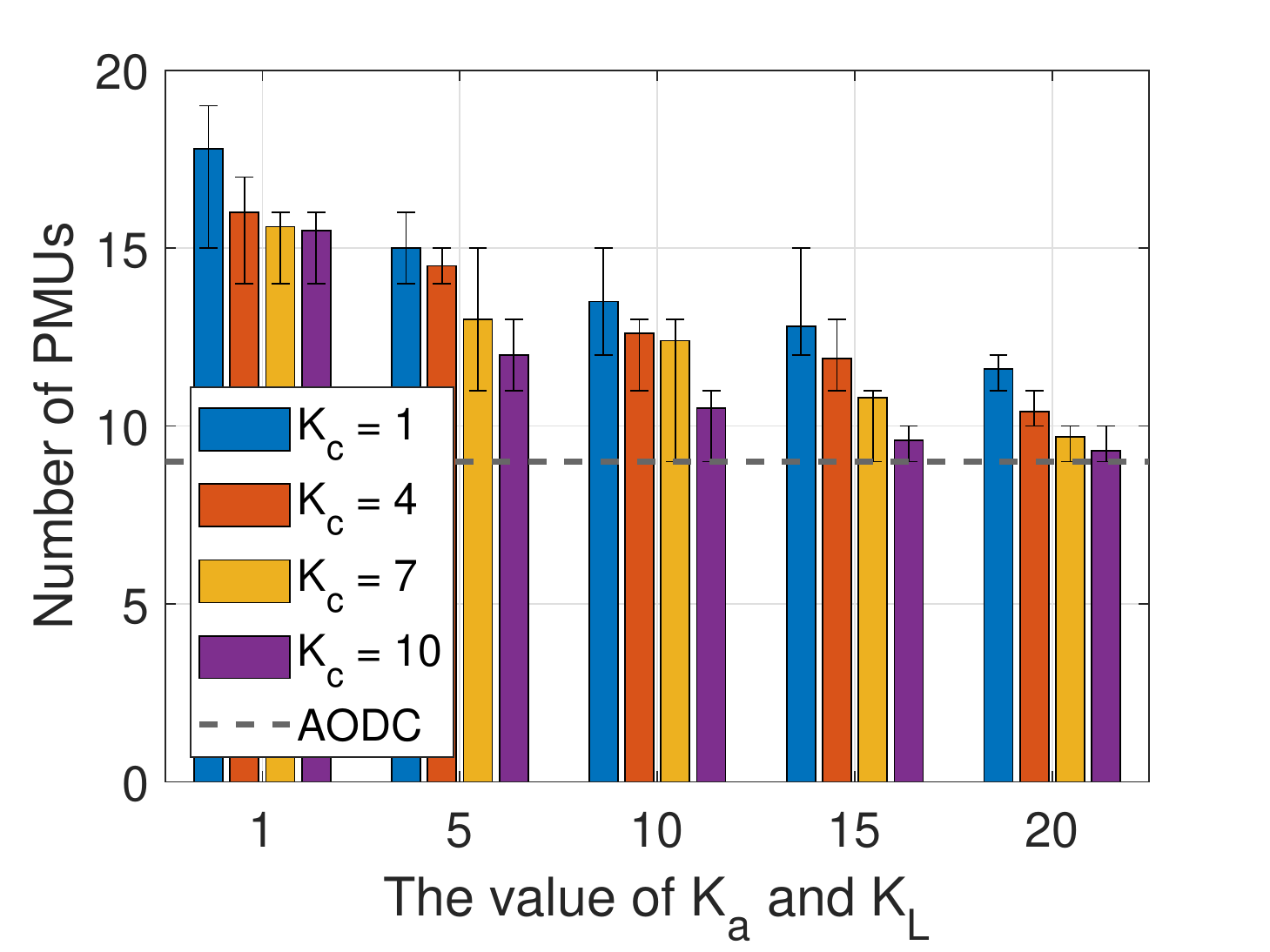}}
  \centerline{\small (a) Number of secured PMUs }
\end{minipage}\hfill
\begin{minipage}{.495\linewidth}
 \centerline{
  \includegraphics[width=1\columnwidth]{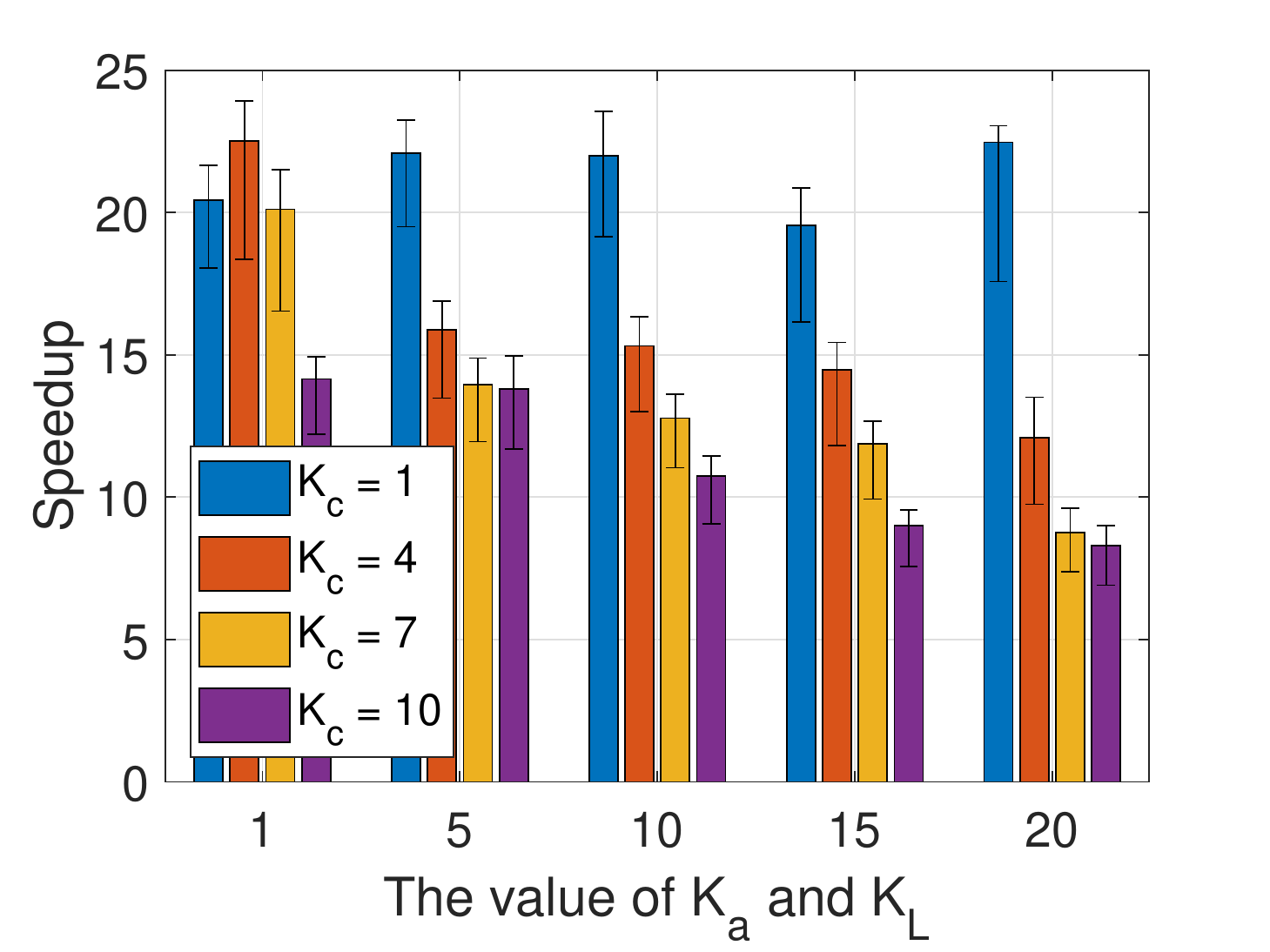}}
  \centerline{\small (b) Computation time }
\end{minipage}
  \caption{The performance of Alg.~\ref{alg: heuristic_3_phase} under different $K_c$, $K_A$, and $K_L$.} \label{fig: effect_Kc_KA}
  \vspace{-1em}
\end{figure}

{
\emph{Extension to AC model: }We compare the solution $\Omega_{\mbox{\tiny AC}}$ obtained by Alg.~\ref{alg: ac_augment_pmu_dc} with the best previous solution $\Omega_{\mbox{\tiny DC}}$ obtained under the DC approximation. As shown in Table~\ref{tab: ac_matpower_pmu}, although the DC-based solution may need augmentation to defend against AC-based CCPA, the gap (i.e., $|\Omega_{\mbox{\tiny AC}}| - |\Omega_{\mbox{\tiny DC}}|$) is small. More importantly, $|\Omega_{\mbox{\tiny AC}}|$ is still much smaller (by $60$--$80\%$) than the number of PMUs $|\Omega_{\mbox{\tiny FO}}|$ required to achieve full observability (see Table~\ref{tab: fullob_matpower_pmu}), indicating the efficacy of our approach of first computing an initial solution under the DC approximation and then augmenting it to achieve our defense goal under the AC model. We note that the values of $|\Omega_{\mbox{\tiny AC}}|$ in Table~\ref{tab: ac_matpower_pmu} are only upper bounds on the number of PMUs required to prevent outages under AC-based CCPA, suggesting great potential of saving PMUs by adopting the proposed defense goal.
}
\begin{table}[htbp]
\small
\vspace{-0.5em}
\centering
\caption{Number of PMUs Under AC Power Flow Model}
\begin{tabular}{|c| c |c|c |c |}
\hline
                   & 30-bus & 57-bus & 118-bus & 300-bus \\
                   \hline
$|\Omega_{\mbox{\tiny AC}}|$         & 3      & 3      & 10 & 34 \\
\hline
$|\Omega_{\mbox{\tiny DC}}|$  &   2    &  3     & 9 & 31 \\
\hline
\end{tabular}
\vspace{-.5em}
\label{tab: ac_matpower_pmu}
\end{table}

\section{Conclusion}\label{sec: conclusion}
We formulate a tri-level optimization problem {under the DC power flow model} to find the optimal secured PMU placement to defend against the coordinated cyber-physical attack (CCPA) in the smart grid. Rather than completely eliminating the attack, we propose to limit the impact of the attack by preventing overload-induced outages. To solve the proposed problem, we first transform it into a bi-level MILP and then propose an alternating optimization algorithm framework to obtain optimal solutions. The core of the proposed algorithm framework is constraint generation based on infeasible placements, for which we develop two constraint generation approaches. Furthermore, we propose a polynomial-time heuristic algorithm that can scale to large-scale grids. {In addition, we demonstrate how to extend the obtained PMU placement to achieve our defense goal under the AC power flow model.} Our experimental results on standard test systems demonstrate great promise of the proposed approach in reducing the requirement of PMUs.
{Our work lays the foundation for tackling a number of further questions in future work, e.g., how to characterize the optimal attack without solving MILPs, how to directly optimize the PMU placement for outage prevention under the AC model, and how to improve the robustness of the solution against the failures of PMUs themselves.}

\bibliographystyle{IEEEtran}
\bibliography{myBib}

\newpage
\ifisappendix
\begin{appendices}
\section{MILP Formulation of Attacker's Problem}\label{appendix: bi_level_milp}
In this section, we will demonstrate how to transform \eqref{eq:reformulate_attacker} into a MILP, which can be efficiently solved by existing solver such as Gurobi. 

To begin with, we give an overview of the PPOP, as shown in Fig.~\ref{fig: ppop}.
\begin{figure}[ht]
\vspace{-.5em}
\centering
\includegraphics[width=.6\linewidth]{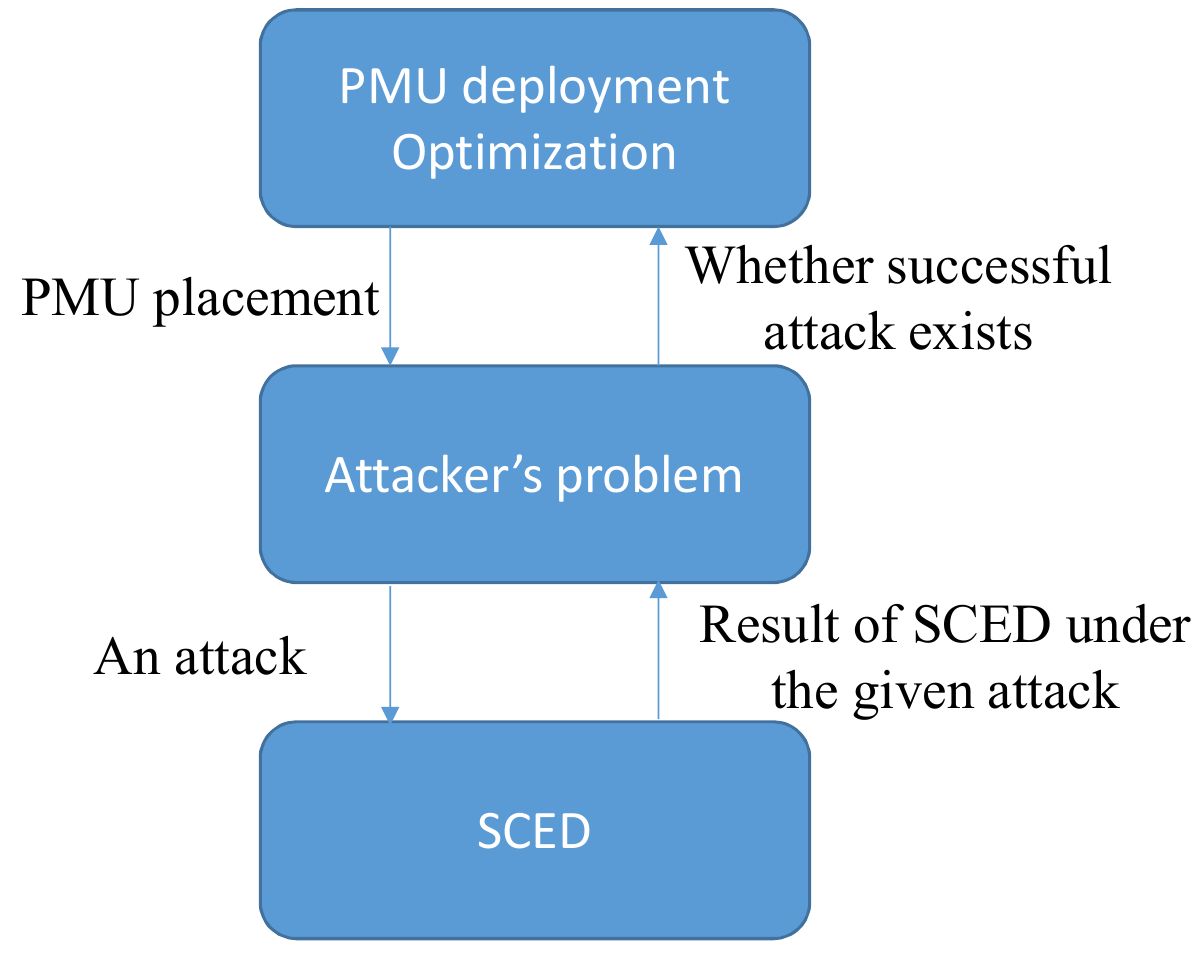}
\vspace{-1em}
\caption{Overview of the PPOP} \label{fig: ppop}
 \vspace{-1em}
\end{figure}

{We first consider the case that lower-level optimization \eqref{eq:reformulate_OPF} returns the set of $\vtheta$'s satisfying \eqref{eq: cc_load_meet}-\eqref{eq: cc_pgp}, i.e., it returns the feasible region of SCED rather than a single solution. In this case, \eqref{eq:reformulate_attacker} becomes a single-level problem.}


Below, we show how to convert the single-level formulation of \eqref{eq:reformulate_attacker} into a MILP. To convert \eqref{eq: PMU_defend} and \eqref{eq: attack_overflow} into linear constraints, we introduce a constant $M_{2,\theta}$ (defined in Appendix~\ref{appendix:Calculation of Big-M}) such that \eqref{eq: beta_on_theta} holds if and only if the following holds:\looseness=-1
\begin{subequations}\label{eq: linear_ifthen}
\begin{align}
\tilde{\theta}_{2,u} - \theta_{2,u} &\le M_{2,\theta}\cdot (1 - \vx_{N,u}), \label{eq: constr_theta2_pmu_p_appendix} \\
\tilde{\theta}_{2,u} - \theta_{2,u} &\ge -M_{2,\theta}\cdot (1 - \vx_{N,u})\label{eq: constr_theta2_pmu_n_appendix},
\end{align}
\end{subequations}
and similar conversion applies to \eqref{eq: beta_on_ap}.
As for \eqref{eq: attack_overflow}, by defining a sufficiently large constant $M_{\pi,e}$ (see Appendix~\ref{appendix:Calculation of Big-M}) and two binary auxiliary variables $\pi_{n,e}, \pi_{p,e}$ to get rid of the absolute value operation, \eqref{eq: attack_overflow} is transformed into { \begin{subequations}\label{eq: linear_overflow_v2}
\begin{align}
-M_{\pi,e}\cdot (1-\pi_{p,e}) < \frac{f_{3,e}}{f_{\tiny\text{max},e}}  - \gamma_e   \le M_{\pi,e}\cdot \pi_{p,e}, \\
-M_{\pi,e}\cdot (1-\pi_{n,e}) < \frac{-f_{3,e}}{f_{\tiny\text{max},e}}  - \gamma_e    \le M_{\pi,e}\cdot \pi_{n,e}.
\end{align}
\end{subequations}
We claim that $\pi_e = \pi_{n,e}+\pi_{p,e}$. To see this, suppose that $f_{3,e} \ge 0$. Then, we must have $-\frac{f_{3,e}}{f_{\tiny\text{max},e}}-\gamma_e \le 0$ and thus $\pi_{n,e} = 0$, while $\frac{|f_{3,e}|}{f_{\tiny\text{max},e}}-\gamma_e=\frac{f_{3,e}}{f_{\tiny\text{max},e}}-\gamma_e$ and thus $\pi_{p,e}=\pi_e$. Notice that we must have $\pi_e = 1$ if $|f_{3,e}| - \gamma_e\cdot f_{\tiny\text{max},e} > 0$, while $|f_{3,e}| - \gamma_e\cdot f_{\tiny\text{max},e} \le 0$ leads to $\pi_e = 0$. }


{To linearize \eqref{eq: f2_fdi_xi_c}, we introduce binary variables $\vw_{f} \in \{0,1\}^{|m_L|}$ and $\vw_{p} \in \{0,1\}^{|m_N|}$ for data injection on line measurements and node measurements, respectively. Then, \eqref{eq: f2_fdi_xi_c} can be transformed into (see definitions of $M_{c,f},\: M_{c,p}$ in  Appendix~\ref{appendix:Calculation of Big-M}) 
\begin{subequations}\label{eq: linearization_xi_c}
\begin{align}
&-M_{c,f} \vw_f \le \vLambda_f \lrbrackets{\tilde{\vf}_2 - \vf_2} \le M_{c,f} \vw_f, \\
&-M_{c,p} \vw_p \le \vLambda_p \lrbrackets{\tilde{\vB}\tilde{\vtheta}_2 - \vp_0} \le M_{c,p} \vw_p, \\
&\bm{1}^T \vw_f + \bm{1}^T \vw_p \le \xi_c.
\end{align}
\end{subequations}}

Together, the above techniques transform \eqref{eq:reformulate_attacker} into a MILP. Specifically, the binary decision variables are $\{\vpi_n,\vpi_p, \va_p, \vw_f,\vw_p\}$, continuous variables are $\{\tilde{\vtheta}_2 , \tilde{\vtheta}_3 , \vtheta_2 , \vtheta_3 , \vf_2, \vf_3 , \tilde{\vf}_2, \vf_{con}\}$, where $\vw_f,\vw_p$ are introduced auxiliary variables. Then, the full formulation without considering the optimality of \eqref{eq:reformulate_OPF} is given as follows.
\begin{maxi!}|s|[3]<b>
{}
{\lp{\vpi_p + \vpi_n}_0}    
{\label{eq:attacker_complete}}
{}
\addConstraint{\vDelta^{-1}\underline{\vA} \vbeta\le \vx_N \le \vDelta^{-1}\underline{\vA} \vbeta + \frac{\lp{\vDelta}_{\infty}-1}{\lp{\vDelta}_{\infty}}}{} \label{milp: xn}
\addConstraint{\frac{1}{2}|\vD|^T  \vbeta \le \vx_L \le \frac{1}{2}|\vD|^T  \vbeta + \zeta}{}\label{milp: xl}
\addConstraint{-(\bm{1}-\va_{p})\le \mbox{diag}\left(\vgamma \odot \vf_{\tiny\text{max}}\right)^{-1} \vf_2  \le \bm{1}-\va_{p}}{}\label{milp: theta2_start}
\addConstraint{\tilde{\vD} \vf_2 = \vp_0, -M_{2,f}\va_{p}\le \vGamma \vD \vtheta_2-\vf_2 \le M_{2,f}\va_p}{}\label{milp: theta2_end}
\addConstraint{-\vf_{\tiny\text{max}}\le \tilde{\vf}_{2} \le \vf_{\tiny\text{max}},~ \vGamma \tilde{\vD}^T \tilde{\vtheta}_2-\tilde{\vf}_{2}  = 0}{}\label{milp: theta2_fdi_start}
\addConstraint{-\alpha |\vp_0| \le \tilde{\vD}\tilde{\vf}_2 - \vp_0 \le \alpha |\vp_0|}{}
\addConstraint{\vLambda_g \tilde{\vD}\tilde{\vf}_2 =\vLambda_g\vp_0}{}
\addConstraint{-M_{c,f} \vw_f \le \vLambda_f \lrbrackets{\tilde{\vf}_2 - \vf_2} \le M_{c,f} \vw_f}{}
\addConstraint{-M_{c,p} \vw_p \le \tilde{\vB}\tilde{\vtheta}_2 - \vp_0 \le M_{c,p} \vw_p}{}
\addConstraint{\bm{1}^T \vw_f + \bm{1}^T \vw_p \le \xi_c,~\lp{\va_p}_0 \le \xi_p}{}\label{milp: theta2_fdi_end}
\addConstraint{\tilde{\theta}_{2,u} - \theta_{2,u} \le M_{2,\theta}\cdot (1 - \vx_{N,u})}{}\label{milp: theta2_protect_p}
\addConstraint{\tilde{\theta}_{2,u} - \theta_{2,u} \ge -M_{2,\theta}\cdot (1 - \vx_{N,u})}{}\label{milp: theta2_protect_n}
\addConstraint{\vp_{g,min} \le \vLambda_g \tilde{\vB} \tilde{\vtheta}_3 \le \vp_{g,max}}{}\label{milp: theta3_start}
\addConstraint{-\vf_{\tiny\text{max}} \le \vGamma\vD^T \tilde{\vtheta}_3 \le \vf_{\tiny\text{max}}}{}
\addConstraint{\vLambda_d \tilde{\vB} \tilde{\vtheta}_3 = \vLambda_d\tilde{\vD}\tilde{\vf}_2}{}
\addConstraint{-M_{3,a}(\bm{1}-\va_{p})\le \vf_3  \le M_{3,a}(\bm{1}-\va_{p})}{}
\addConstraint{\vLambda_d\tilde{\vD}\vf_3 = \vLambda_d\vp_0,~~~ \vLambda_g\tilde{\vD}\vf_3 = \vLambda_g\tilde{\vB}\tilde{\vtheta}_3}{}
\addConstraint{-M_{3,f}\va_p\le \vGamma \tilde{\vD}^T \vtheta_3-\vf_3 \le M_{3,f} \va_p}{}\label{milp: theta3_end}
\addConstraint{\theta_{2,u_0}=\theta_{3,u_0} =  \tilde{\theta}_{2,u_0} = \tilde{\theta}_{3,u_0}=0}{}\label{milp: zero_ref_theta}
\addConstraint{-\cdot (1-\pi_{p,e}) < \frac{f_{3,e}}{M_{\pi,e}f_{\tiny\text{max},e}}  - \frac{\gamma_e }{M_{\pi,e}} \le \cdot \pi_{p,e}, \forall e}{}\label{milp: outage_1}
\addConstraint{-\cdot (1-\pi_{n,e}) < \frac{-f_{3,e}}{M_{\pi,e}f_{\tiny\text{max},e}}  - \frac{\gamma_e }{M_{\pi,e}} \le \cdot \pi_{n,e}, \forall e}{}\label{milp: outage_2}
\addConstraint{\tilde{\vD}_u\vf_{con} = \begin{cases}|V|-1, &\mbox{ if } u = u_0,\\
-1, &~\mbox{if } u\in V\setminus\{u_0\},
\end{cases}}{}\label{milp: fcon_1}
\addConstraint{-|V|\cdot (1-a_{p,e}) \le f_{con,e} \le |V|\cdot (1-a_{p,e})}{}\label{milp: fcon_2}
\end{maxi!} 
The constraints \eqref{milp: xn}-\eqref{milp: xl} correspond to \eqref{eq: const_beta_x_N}-\eqref{eq: const_beta_x_L}, \eqref{milp: theta2_start}-\eqref{milp: theta2_end} correspond to \eqref{f2_range}-\eqref{f2_theta_valid}, \eqref{milp: theta2_fdi_start}-\eqref{milp: theta2_fdi_end} correspond to \eqref{constr: f2_fdi}, \eqref{milp: theta2_protect_p}-\eqref{milp: theta2_protect_n} correspond to \eqref{eq: linear_ifthen}, \eqref{milp: theta3_start}-\eqref{milp: theta3_end} correspond to \eqref{eq: constr_f3}, \eqref{milp: zero_ref_theta} corresponds to \eqref{eq: zero_ref_theta}, \eqref{milp: outage_1}-\eqref{milp: outage_2} correspond to \eqref{eq: attack_overflow}, \eqref{milp: fcon_1}-\eqref{milp: fcon_2} correspond to \eqref{eq:connectivity_constr}.

If we do not relax the optimality requirements in \eqref{eq:reformulate_OPF}, we need to introduce additional binary variables $\{\vr_{fl},\vr_{fu},\vr_{gl},\vr_{gu}\}$ and continuous dual variables $\{\vmu_b,\vmu_c,\vmu_d,\vmu_e,\vmu_g\}$ to transform \eqref{eq:reformulate_OPF} into a linear system by using its KKT conditions \cite{yuan2011modeling}. Specifically, we add the following linear system into \eqref{eq:attacker_complete} for the completeness of KKT conditions of \eqref{eq:reformulate_OPF}:
\begin{subequations}
\allowdisplaybreaks
\begin{align}
&\tilde{\vB}\vLambda_d^T \vmu_b + \tilde{\vD}\vGamma \vmu_c + \tilde{\vD}\vGamma \vmu_d + \tilde{\vB}\vLambda_g^T \vmu_e - \tilde{\vB}\vLambda_g^T \vmu_g = -\tilde{\vB}\vLambda_g^T \vphi \\
&\vmu_c - M \vr_{fl} \le \bm{0}, \\
&\vGamma\tilde{\vD}^T\tilde{\vtheta_3} + M \vr_{fl}\le M - \vf_{\tiny\text{max}}\\
&\vmu_d - M \vr_{fu} \le \bm{0}, \\
&-\vGamma\tilde{\vD}^T\tilde{\vtheta_3} + M \vr_{fl}\le M - \vf_{\tiny\text{max}}\\
&\vmu_e - M \vr_{gl} \le \bm{0}, \\
&\vLambda_g\tilde{\vB}\tilde{\vtheta}_3 + M\vr_{gl}\le \vp_{g,min}+M\bm{1} \\
&\vmu_g - M \vr_{gu} \le \bm{0}, \\
&-\vLambda_g\tilde{\vB}\tilde{\vtheta}_3 + M\vr_{gu}\le -\vp_{g,max}+M\bm{1}\\
&\vr_{gl}+\vr_{gu}\le \bm{1}\\
&\vr_{fl}+\vr_{fu}\le \bm{1} \\
&\vmu_c, \vmu_d, \vmu_e, \vmu_g \ge \bm{0}
\end{align}
\end{subequations}
{Compared to the attacker's formulations in \cite{li2015bilevel, yudi21SmartGridComm} that also optimize the location of physical attacks, the key advantage of \eqref{eq:reformulate_attacker} is avoiding McCormick's relaxation for bilinear terms \eqref{eq: bilinear_terms} and reducing the numbers of variables and constraints. Specifically, McCormick's relaxation in \cite{yudi21SmartGridComm} will introduce $2|E||V|$ additional continuous variables and $8|E||V|$ additional constrains. The cost of avoiding bilinear term in \eqref{eq:reformulate_attacker} is the additional variables $\vf_{2},\vf_3, \tilde{\vf}_2$ and the associated constraints, although the benefit usually outweighs the cost. For example, for the IEEE 118-bus system, the formulation in \cite{yudi21SmartGridComm} has $44436$ continuous variables and $178,596$ constraints, while \eqref{eq:reformulate_attacker} only has $1216$ continuous variables and $5,802$ constraints.}

\section{Calculation of Big-M}\label{appendix:Calculation of Big-M}
In this section, we will explain how to calculate {$M_{2,a,e}$ in \eqref{f2_range}}, $M_{2,f}$ in \eqref{f2_theta_valid}, $M_{3,a}$ in \eqref{eq:f_3, a_p}, $M_{2,\theta}$ in \eqref{eq: linear_ifthen}, $M_{3,f}$ in \eqref{eq:theta_3, f_3}, $M_{\pi,e}$ in \eqref{eq: linear_overflow_v2}, $\overline{M}_{F}, \underline{M}_{F}$ in \eqref{eq: McCormick_relaxation}, $M_{c,f}, M_{c,p}$ in \eqref{eq: linearization_xi_c} and $M_{q}$ in \eqref{eq: q2_beta_relationship}. In this section, we denote $\mathcal{N}=(V,E)$ as the graph before physical attack while $\mathcal{N}'=(V,E')$ as the graph after attack.

We first show how to calculate {$M_{2,a}$}, $M_{2,f}$ and $M_{2,\theta}$. {Suppose that the power grid is designed to be robust to $N-k$ contingency. Then, the value of $M_{2,a}$ depends on $\xi_p-k$. If $\xi_p-k\leq 0$, then we can set $M_{2,a,e}:=f_{\tiny\text{max},e}$ or $M_{2,a,e}:=\gamma_e f_{\tiny\text{max},e}$, since no $\va_p$ can cause overloading. Otherwise, we set $M_{2,a,e}:=C_{2,a}\gamma_e f_{\tiny\text{max},e}$ with a parameter $C_{2,a}>1$. In our simulations, we find that $C_{2,a}:=3$ suffices since $\xi_p-k$ is usually small.} Next, we bound $|\vtheta_2|$ by defining $M_{\theta_2,u} \ge \max_{\va_p} |\theta_{2,u}|$ and $M_{\theta_2} \ge \max_{\va_p}\max_u |\theta_{2,u}|$ since the value of $\vtheta_2$ depends on $\va_p$. An intuitive way of obtaining $M_{\theta_2,u}$ is enumerating all possible values of $\va_p$, whose time complexity is polynomial in $|E|$ and $|V|$ if $\xi_p = \mathcal{O}(1)$. Here we provide another way of bounding $M_{\theta_2,u}$. Due to our assumption of the connected $\mathcal{N}$, there exists at least one path in $\mathcal{N}$ connecting the reference node $u_0$ to each node $u\in V$. Moreover, for each path connecting $u_0$ and $u$, say $Pa(u_0,u) := (e_0, e_1,\cdots,e_J)$ where $e_0=(u_0,v_1)$, $e_j = (v_j,v_{j+1})$ and $e_J = (v_J,u)$, we have $\theta_{2,u} - \theta_{2,u_0} = \theta_{2,s}-\theta_{2,v_1}+\theta_{2,v_1}-\cdots+\theta_{2,v_J} - \theta_{2,t}$, which leads to 
\begin{align}\label{append: max_theta_2_u}
\max_{\va_p} |\theta_{2,u}| &= \max_{\va_p} |\theta_{2,u} - \theta_{2,u_0}|  \nonumber \\
&\le \sum_{j = 0}^J{ r_{e_j} M_{2,a,e_j} := M_{Pa}(\theta_{2,u})}
\end{align}
since $\theta_{2,u_0} = 0$ in our assumption and $|\theta_{2,v_j} - \theta_{2,v_{j+1}}| \le {r_{e_j} M_{2,a,e_j}}$ due to \eqref{f2_range}. Denote $n_p$ as the number of different paths connecting $u$ and $u_0$. Then, since the physical attack will disconnect at most $\xi_p$ lines, we set $M_{\theta_2,u} := \max \{M_{Pa_i}(\theta_{2,u})\}_{i=1}^{\min\{\xi_p+1, n_p\}}$.

Equipped with $M_{\theta_2,u}, u\in V$, we can calculate $M_{2,f}$ and $M_{2,\theta}$. We define an intermediate constant $M_{2,f,e}$ for each line such that $M_{2,f} = \max_{e\in E} M_{2,f,e}$. Then, for $e=(u,v)$ we can set $M_{2,f,e} := r_e (M_{\theta_2,u} + M_{\theta_2,v})$ since $|\Gamma_e \tilde{\vd}_e^T\vtheta_2 - f_{2,e}| > 0$ only if $a_{p,e} = 1$ and $f_{2,e}=0$. 

To obtain $M_{2,\theta}$, we first bound $M_{\tilde{\theta}_2,u} \ge \max_{\va_p,\va_c} |\tilde{\theta}_{2,u}|, u\in V$ in a similar way as that in \eqref{append: max_theta_2_u}. Specifically, since $\tilde{\vtheta}_2$ is estimated by CC based on the topology $\mathcal{N}$, we can arbitrarily choose one path $(e_0, e_1,\cdots,e_J)$ in $\mathcal{N}$ that connects $u$ and $u_0$ and set 
\begin{align}\label{append: max_theta_2_fdi_u}
M_{\tilde{\theta}_2,u} := \sum_{j = 0}^J r_{e_j} f_{max,e_j} \ge \max_{\va_p,\va_c} |\tilde{\theta}_{2,u}|.
\end{align}
Then, we can set $M_{2,\theta} := \max_{u\in V} ( M_{\tilde{\theta}_2,u} + M_{\theta_2,u} )$.


Now, we are ready to demonstrate the calculation of $M_{3,a}$ and $M_{3,f}$. As for $M_{3,a}$, we only require $M_{3,a,e} > \gamma_e f_{\tiny\text{max},e}$ and $M_{3,a} \ge \max_{e\in E}\gamma_e f_{\tiny\text{max},e}$ so that the attacker can cause outages over any lines. In practice, we can set $M_{3,a} := c \max_{e\in E}\gamma_e f_{\tiny\text{max},e}$ with $c > 1$. As for $M_{3,f}$, we again first show that we can bound $|\theta_{3,u}| \le M_{\theta_3,u}$ without hurting the attacker's objective. We notice that the topology of grid at $t_3$ before lines facing outage automatically disconnect themselves is still $\mathcal{N}'$. Thus, we can set $M_{\theta_3,u}$ similarly as $M_{\theta_2,u}$, except that \eqref{append: max_theta_2_u} becomes:
\begin{align}
\max_{\va_p} |\theta_{3,u}| \le \sum_{j = 0}^J r_{e_j} M_{a,e_j} := M_{Pa}(\theta_{3,u}).
\end{align}
Then, we can set $M_{\theta_3,u} := \max \{M_{Pa_i}(\theta_{3,u})\}_{i=1}^{\min\{\xi_p+1, n_p\}}$, $M_{3,f,e} := r_e (M_{\theta_3,u} + M_{\theta_3,v})$ for $e=(u,v) \in E$, and $M_{3,f} = \max_{e\in E} M_{3,f,e}$.

Equipped with $M_{3,a,e}$, $M_{\pi,e}$ can be easily set as $c\cdot ( \frac{M_{3,a,e}}{f_{\tiny\text{max},e}} + \gamma_e )$ with any constant $c > 1$.

We can set $\overline{M}_{F}$ as $0$ since $\vq_2 \ge \bm{0}$ and $F_{3,i,u} \in \{0,-M_{2,\theta}\}$, $\forall i,u$. There is no simple guidelines for $\underline{M}_{F}$ in \eqref{eq: McCormick_relaxation} since it is the bound for dual variables. In practice, we can initialize $\underline{M}_{F}$ to a given value and solve \eqref{eq: vcg_formulation} for each attack pair separately. Then, we iteratively decrease $\underline{M}_{F}$ until \eqref{eq: vcg_formulation} is feasible under each attack pair separately. In our simulations, we set $\underline{M}_{F} := -M_{2,\theta}^2$. Equipped with $\underline{M}_{F}$, we can set $M_{q} := \frac{2\underline{M}_{F}}{M_{2,\theta}}$.

Finally, we demonstrate how to set $M_{c,f}$ and $M_{c,p}$. Due to \eqref{f2_range} and \eqref{eq: f2_fdi_valid_range}, we have $|\tilde{f}_{2,e}-f_{2,e}| \le (1+\gamma_e) f_{\tiny\text{max},e}$, which implies that we can set $M_{c,f}:= \max_{e\in E} (1+\gamma_e) f_{\tiny\text{max},e}$. Similarly, we can set $M_{c,p}:= \alpha \lp{\vp_0}_{\infty}$ due to \eqref{eq: fdi_range}.

\section{Efficiency Analysis of \quotes{No-Good} Constraints}\label{appendix: effi_no_good}
We have the following observations about AONG:
\begin{enumerate}
    \item \emph{Cold start}. The efficiency of \eqref{eq: no_good_cut} can be characterized by the number of infeasible $\vbeta$'s that are cut out. Let $\{\hat{\vbeta}^{(k)}\}_{k=1}^K$ be the PMU placements obtained in the first $K$ iterations of Alg.~\ref{alg: alter_opt_v1} and $\{\hat{\vbeta}^{'(k)}\}_{k=1}^K$ the corresponding augmented placements obtained from \eqref{eq: opt_max_nogood_but}. Then, the number of feasible $\vbeta$'s for the next iteration is at least
    \begin{align}
    \left( 2^{|\bigcap_{k=1}^K \Omega(\hat{\vbeta}^{'(k)})^c|} - 1 \right) \cdot 2^{ |V| - |\bigcap_{k=1}^K \Omega(\hat{\vbeta}^{'(k)})^c| }
    \end{align}
    if $\bigcap_{k=1}^K \Omega(\hat{\vbeta}^{'(k)})^c \ne \emptyset$, as placing at least one PMU in $\bigcap_{k=1}^K \Omega(\hat{\vbeta}^{'(k)})^c$ will satisfy \eqref{eq: no_good_cut} for every placement in $\{\hat{\vbeta}^{'(k)}\}_{k=1}^K$. This implies that the number of $\vbeta$'s that are cut out is at most $2^{ |V| - |\bigcap_{k=1}^K \Omega(\hat{\vbeta}^{'(k)})^c| }$. Therefore, the first $K$ ``No-Good'' constraints \eqref{eq: no_good_cut} added in Alg.~\ref{alg: alter_opt_v1} will be inefficient if $|\bigcap_{k=1}^K \Omega(\hat{\vbeta}^{'(k)})^c|$ is large. 
    We observe that $|\bigcap_{k=1}^K \Omega(\hat{\vbeta}^{'(k)})^c|$ is large at the beginning of Alg.~\ref{alg: alter_opt_v1} and decreases quickly as $\lp{\hat{\vbeta}^{(k)}}_0$ increases. 
    \item \emph{Repeated successful attacks}. Another cause of inefficiency 
    is that for many PMU placements 
    enumerated by AONG, there exist successful attacks based on the same attack pair $(\va_p, e)$, indicating that new constraints are needed to better defend against identified attacks. 
\end{enumerate}


\section{The Details of Coefficients in Attacker's Problem} \label{appendix: expansion_attack_denial_primal}
The linear system \eqref{eq: LP_given_pair_eq} is the composition of \eqref{eq: no_attack_generator}, \eqref{eq: attack_load} and \eqref{eq: pd_theta_3_fdi}, which can be expanded into:
\begin{align}
\left[ \begin{array}{ccc}
   \vLambda_g \tilde{\vB}  & \bm{0}  & \bm{0} \\
   \bm{0}  & \bm{0}  & \vLambda_d \vB \\
   \bm{0}  & -\vLambda_g \tilde{\vB} & \vLambda_g \vB  \\
   \vLambda_d \tilde{\vB} & -\vLambda_d \tilde{\vB} & \bm{0}
\end{array} 
\right]  
\left[\begin{array}{c}
     \tilde{\vtheta}_2  \\
     \tilde{\vtheta}_3 \\
     \vtheta_3
\end{array}
\right] = 
\left[\begin{array}{c}
     \vLambda_g \vp_0  \\
     \vLambda_d \vp_0 \\
     \bm{0} \\
      \bm{0}
\end{array}
\right]
\end{align}
as well as $\tilde{\theta}_{2,u_0} = \tilde{\theta}_{3,u_0} = \theta_{3,u_0} = 0$. For a given attack pair $(\va_p, e)$ and the corresponding $\vtheta_2$, the expansion of \eqref{eq: LP_given_pair_ineq} is 
\begin{align}\label{eq: expand_primal_ineq}
\begin{blockarray}{cccc}
{\tilde{\vtheta}_2}&{\tilde{\vtheta}_3} & {\vtheta_3} & {\vs_2 + \vF_3 \vx_N}  \\
\begin{block}{[cccc]}
\tilde{\vB}  & \bm{0}  & \bm{0} & \vp_0 + \alpha |\vp_0| \\
-\tilde{\vB}  & \bm{0}  & \bm{0} & -\vp_0 + \alpha |\vp_0| \\
\vI_{|V|}  & \bm{0}  & \bm{0} & \vtheta_2 + M_{\theta}(1-\vx_N) \\
-\vI_{|V|}  & \bm{0}  & \bm{0} & -\vtheta_2 + M_{\theta}(1-\vx_N) \\
0  &  0  & -\Gamma_e\vd_e^T & -\gamma_e f_{\tiny\text{max},e} \\
\bm{0} & \tilde{\vD}^T \vGamma & \bm{0} &\vf_{\tiny\text{max}} \\
\bm{0} & -\tilde{\vD}^T \vGamma & \bm{0} &\vf_{\tiny\text{max}} \\
\bm{0}  & \vLambda_g \tilde{\vB} & \bm{0} & \vp_{g,max}   \\
\bm{0}  & -\vLambda_g \tilde{\vB} & \bm{0} & -\vp_{g,min}  \\
\tilde{\vD}^T\vGamma & \bm{0} & \bm{0} &\vf_{\tiny\text{max}} \\
-\tilde{\vD}^T\vGamma & \bm{0}  & \bm{0} &\vf_{\tiny\text{max}} \\
\end{block}
\end{blockarray}
\end{align}
Specifically, the first two rows of \eqref{eq: expand_primal_ineq} correspond to \eqref{eq: fdi_range}, the next two rows correspond to \eqref{eq: linear_ifthen}, the 5-th row indicates the outage at the target line, the 6-th and 7-th rows correspond to \eqref{eq: fmax_theta_3_fdi}, the 8-th and 9-th rows correspond to \eqref{eq: pg_theta_3_fdi}, and the last two rows correspond to \eqref{eq: f2_fdi_valid_range}.

\section{Details of the Attacker's Problem Under AC Power Flow Model}\label{appendix: ac_attacker_detail}
\begin{table}[tb]
\footnotesize
\renewcommand{\arraystretch}{1.3}
\caption{Notations for AC power flow} \label{tab:notation_ac_appendix}
\vspace{-1em}
\centering
\begin{tabular}{c|l}
  \hline
  Notation & Description  \\
  \hline
 {$\vp/\vq \in \mathbb{C}^{|V|}$} & {Active/reactive power injection} \\
 \hline
 {$\Vec{v}_u = v_u e^{j\cdot \theta_u}$} & {node voltage} \\
 \hline
 {$\tilde{\vY}_{bus} = \tilde{\vG}_{bus} + j\tilde{\vB}_{bus}$} & {Bus admittance matrix} \\
 \hline
 {$\tilde{\vY}_f/ \tilde{\vY}_t \in \mathbb{C}^{|E|\times|V|}$} & {\emph{From}/\emph{to} end admittance matrix} \\
 \hline
 {$\vC_f/ \vC_t \in \{0,1\}^{|E|\times|V|}$ } & {\emph{From}/\emph{to} end incidence matrix} \\
 \hline
 {$\vp_f/ \vp_t \in \mathbb{C}^{|E|}$} & {\emph{From}/\emph{to} end active power flow} \\
 \hline
 {$\vq_f/ \vq_t \in \mathbb{C}^{|E|}$} & {\emph{From}/\emph{to} end reactive power flow} \\
 \hline
 {$|\vI_f|^2/ |\vI_t|^2 \in \mathbb{C}^{|E|}$} & {Square of \emph{from}/\emph{to} end current magnitude} \\
 \hline
 {$\vI_{max} \in \mathbb{R}^{|E|}$} & {Limit on line current magnitude} \\
 \hline
 {$\hat{\vI}_{max} \in \mathbb{R}^{|E|}$} & {Threshold for line tripping}\\
 \hline
 {$\tilde{\vY}_c = \tilde{\vG}_c + j\tilde{\vB}_c \in \mathbb{C}^{|E|}$} & {line charging} \\
 \hline
 {$\tilde{\vZ} = \tilde{\vZ}_R + j\tilde{\vZ}_I \in \mathbb{C}^{|E|}$} & {line impedance} \\
 \hline
 {$\tilde{\vY}_L = \tilde{\vG}_L + j\tilde{\vB}_L \in \mathbb{C}^{|E|}$} & {line admittance} \\
 \hline
 {$\vV_{max}/\vV_{min} \in \mathbb{R}^{|V|}$} & {Limit on node voltage magnitude} \\
 \hline
 {$\vtheta_{max}/\vtheta_{min} \in \mathbb{R}^{|E|}$} & {Limit on phase angle difference for lines} \\
 \hline
 {$\hat{\vp}_{3}/\hat{\vq}_{3} \in \mathbb{R}^{|V|}$} & {approximated power injections at $t_3$} \\
 \hline
 {$\hat{\vp}_{f,3}/\hat{\vq}_{f,3} \in \mathbb{R}^{|E|}$} & {approximated line power flow at $t_3$} \\
 \hline
\end{tabular}
\vspace{-.5em}
\end{table}
\normalsize

For completeness, we summarize the necessary notations for presenting AC power flow model in Table~\ref{tab:notation_ac_appendix}. Specifically, we denote $\vC_f$ as the \emph{From} end incidence matrix, in which $C_{f,e,i} = 1$ if and only if we have $e=(i,k)\in E$. The \emph{To} end incidence matrix $\vC_t$ is defined similarly, where $C_{t,e,k} = 1$ if and only if we have $e=(i,k)\in E$.

We provide details about \eqref{eq: ac_formulate_overload}, where we adopt QC relaxation proposed in \cite{coffrin2015qc} for \eqref{ac: t3_acopf} and linearized approximation proposed in \cite{yang2018general} for \eqref{ac: pf_linearze_t3}. As for the constraint on false data injection to bypass BDD \eqref{ac: fdi_t2}, we follow \cite{chung2018local} to formulate QC relaxation-based constraints. 

To begin with, we demonstrate the basics on QC relaxation for AC power flow equations. Recall from Table~\ref{tab:notation_ac_appendix} that the complex voltage on node $i$ is $\Vec{v_i}:=v_i e^{j\cdot \theta_i}$. Then, we introduce auxiliary variables $c_{ii}, c_{ik}$ and $s_{ik}$ in the hope that
\begin{subequations}
\begin{align}
c_{ii} &= v_i^2, \\
c_{ik} &= v_i v_k \cos{\theta_{ik}} \\
s_{ik} &= v_i v_k \sin{\theta_{ik}},
\end{align}
\end{subequations}
where $\theta_{ik} = \theta_i - \theta_k$. As proposed in \cite{coffrin2015qc}, we first introduce the notation $\langle x\rangle^{\cdot}$, which indicates an auxiliary variable as well as the associated constraints with $x$ as input. Concretely, $\left\langle x^{2}\right\rangle^{T}$ indicates the auxiliary variable $\breve{x}$ together with the following constraints:
\begin{align}\label{eq: x_T}
\allowdisplaybreaks
\left\langle x^{2}\right\rangle^{T} \equiv\left\{\begin{array}{l}
\breve{x} \geqslant x^{2} \\
\breve{x} \leqslant\left(x_u+x_l\right) x-x_u x_l
\end{array}\right. ,
\end{align}
where $x\in [x_{l}, x_{u}]$ is pre-assigned bound. Similarly, we have
\begin{subequations}\label{eq: x_MSC_conv}
\allowdisplaybreaks
\begin{align}
&\langle x y\rangle^{M} := \left\{\begin{array}{l}
\breve{x y} \geqslant x_l y+y_l x-x_l y_l \\
\breve{x y} \geqslant x_u y+y_u x-x_u y_u \\
\breve{x y} \leqslant x_l y+y_u x-x_l y_u \\
\breve{x y} \leqslant x_u y+y_l x-x_u y_l
\end{array}\right.\\
&\langle\sin x\rangle^{S} :=\left\{\begin{array}{l}
\breve{s x} \leqslant \cos \left(\frac{x_u}{2}\right)\left(x-\frac{x_u}{2}\right)+\sin \left(\frac{x_u}{2}\right) \\
\breve{s x} \geqslant \cos \left(\frac{x_u}{2}\right)\left(x+\frac{x_u}{2}\right)-\sin \left(\frac{x_u}{2}\right)
\end{array}\right.\\
&\langle\cos x\rangle^{C} :=\left\{\begin{array}{l}
c x \leqslant 1-\frac{1-\cos \left(x_u\right)}{\left(x_u\right)^{2}} x^{2} \\
\breve{c x} \geqslant \cos \left(x_u\right)
\end{array}\right.
\end{align}
\end{subequations}
Equipped with \eqref{eq: x_T} and \eqref{eq: x_MSC_conv}, the QC relaxation-based constraints on $c_{ii}$ for each $i\in V$ can be written as $c_{i i} \in\left\langle v_{i}^{2}\right\rangle^{T}$, while the constraints on $c_{ik}$ and $s_{ik}$ for each $e=(i,k)\in E$ are
\begin{subequations}
\allowdisplaybreaks
\begin{align}
&c_{ik} = c_{ki}, \\
&s_{ik} = -s_{ki},\\
&c_{i k}^{2}+s_{i k}^{2} \leq c_{i i} c_{k k},\\
&c_{i k} \in\left\langle\left\langle v_i v_k \right\rangle^{M} \cdot\left\langle\cos{\theta_{ik}}\right\rangle^{C}\right\rangle^{M}, \label{appendix: cik_Mconv}\\
&s_{i k} \in\left\langle\left\langle v_i v_k \right\rangle^{M} \cdot\left\langle\sin{\theta_{ik}} \right\rangle^{S}\right\rangle^{M}. \label{appendix: sik_Mconv}
\end{align}
\end{subequations}

For simplicity, we will omit the auxiliary variables and the associated constraints for modeling \eqref{appendix: cik_Mconv} and \eqref{appendix: sik_Mconv}. We assume that \eqref{appendix: cik_Mconv} and \eqref{appendix: sik_Mconv} are imposed when QC relaxation is adopted. For \eqref{ac: fdi_t2}, the decision variables we focus are $\tilde{c}_{2,ii}, \forall i\in V, \tilde{c}_{2,ik}, \tilde{s}_{2,ik}, \forall e=(i,k)\in E, e= (k,i)\in E$, $\tilde{\vtheta}_{2}, \tilde{\vv}_2$ and $|\tilde{\vI}_{2,f}|^2,|\tilde{\vI}_{2,t}|^2$. Then, the constraints \eqref{ac: fdi_t2} can be written as
\begin{subequations}\label{eq: details_ac_t2_fdi}
\allowdisplaybreaks
\begin{align}
&\vLambda_g (\tilde{\vp}_2 - \vp_0) = 0, \vLambda_g (\tilde{\vq}_2 - \vq_2) = 0\label{eq: ac_no_attack_gen}\\
&\vLambda_g(\tilde{\vv}_2 - \vv_2) = \bm{0}, \\
&\tilde{c}_{2,ii} = v_{2,i}^2, \forall i\in V_g,\\
&-\vLambda_d |\tilde{\vp}_0| \le \alpha\vLambda_d(\tilde{\vp}_{2,i} - \tilde{\vp}_0) \le \alpha\vLambda_d |\tilde{\vp}_0|, \\
&-\vLambda_d |\tilde{\vq}_0| \le \alpha\vLambda_d(\tilde{\vq}_{2,i} - \tilde{\vq}_0) \le \alpha\vLambda_d |\tilde{\vq}_0|, \\
&(1-\eta)\vV_{min} \le \tilde{\vv}_2 \le (1+\eta) \vV_{max} \\
&(1-\eta)\theta_{min,e} \le \tilde{\theta}_{2,e} \le (1+\eta) \theta_{max,e}, \forall e\in E \\
&|\tilde{\vI}_{2,f}| \le \vI_{max}, |\tilde{\vI}_{2,t}| \le \vI_{max}\label{eq: ac_limit_I_t2_fdi} \\
&\tilde{p}_{2,i}=\sum_{k=1, \ldots, n} \tilde{G}_{i k} \tilde{c}_{2,ik}-\tilde{B}_{i k} \tilde{s}_{2,ik}, \label{eq: qc_p_appendix} \\
&\tilde{q}_{2,i}=\sum_{k=1, \ldots, n}-\tilde{B}_{i k} \tilde{c}_{2,ik}-\tilde{G}_{i k} \tilde{s}_{2,ik}, \label{eq: qc_q_appendix} \\
&\tilde{p}_{2,f,e} = \tilde{G}_{f,e,i} \tilde{c}_{2,ii} + \tilde{G}_{f,e,k}\tilde{c}_{2,ik} - \tilde{B}_{f,e,k}\tilde{s}_{2,ik}, \label{eq: p2fe_fdi}\\
&\tilde{q}_{2,f,e} = -\tilde{B}_{f,e,i}\tilde{c}_{2,ii} - \tilde{B}_{f,e,k}\tilde{c}_{2,ik} - \tilde{G}_{f,e,k}\tilde{s}_{2,ik},\\
&\tilde{p}_{2,t,e}=\tilde{G}_{t,e,k}^*\tilde{c}_{2,kk} + \tilde{G}_{t,e,i}\tilde{c}_{2,ik} + \tilde{B}_{t,e,i}\tilde{s}_{2,ik},\\
&\tilde{q}_{2,t,e}=-\tilde{B}_{t,e,k}^*\tilde{c}_{2,kk} - \tilde{B}_{t,e,i}\tilde{c}_{2,ik} + \tilde{G}_{t,e,i}\tilde{s}_{2,ik}, \label{eq: q2te_fdi}\\
&\tilde{p}_{2,i} =\vc_{f,i}^T \tilde{\vp}_{f,2} + \vc_{t,i}^T \tilde{\vp}_{t,2} + \mathcal{R}(Y_{sh,i} \tilde{c}_{2,ii})\\
&\tilde{q}_{2,i} =\vc_{f,i}^T \tilde{\vq}_{f,2} + \vc_{t,i}^T \tilde{\vq}_{t,2} - \mathcal{I}(Y_{sh,i} \tilde{c}_{2,ii}) \\
&\tilde{c}_{2,ii} = v_{2,i}^2, \tilde{v}_{2,i} = v_{2,i}, \tilde{\theta}_{2,i} = \theta_{2,i}, \forall i \text{ with } x_{N,i} = 1, \label{eq: ac_protection_cii}\\
&\tilde{c}_{2,ik} = v_{2,i} v_{2,k} \cos{\theta_{ik}}, \forall e=(i,k) \text{ with } x_{L,e} = 1,\\
&\tilde{p}_{2,e} = p_{2,e}, \tilde{q}_{2,e} = q_{2,e}, \forall e=(i,k) \text{ with } x_{L,e} = 1, \\
&\tilde{I}_{2,f,e} = I_{2,f,e}, \tilde{I}_{2,t,e} = I_{2,t,e}, \forall e=(i,k) \text{ with } x_{L,e} = 1,\label{eq: ac_protection_I}
\end{align}
\end{subequations}
where $\vp_0$ and $\vq_0$ indicates the ground-truth power injections at $t_0$, \eqref{eq: qc_p_appendix}-\eqref{eq: qc_q_appendix} are imposed for each node $i\in V$, \eqref{eq: p2fe_fdi}-\eqref{eq: q2te_fdi} are imposed for all $e=(i,k) \in E$, $\vc_{f,i}/\vc_{t,i}$ is the $i$-th column of $\vC_f/\vC_t$, $\vY_{sh}$ denotes the diagonal matrix of node shunt, $\mathcal{R}(x)/\mathcal{I}(x)$ denotes the real/imaginary part of $x$, \eqref{eq: ac_protection_cii}-\eqref{eq: ac_protection_I} indicates the protection effect of PMUs, and $\eta \in [0,1)$ is a manually assigned factor for $\tilde{v}_2$ and $\tilde{\theta}_2$ not to raise alarms in control center. Besides \eqref{eq: details_ac_t2_fdi}, we impose the following constraints according to \cite[Chapter~5]{coffrin2015distflow} for each $e=(i,k)\in E$ into \eqref{ac: fdi_t2} :
\begin{subequations}\label{eq: ac_details_equiv_I_2}
\allowdisplaybreaks
\begin{align}
&|\tilde{\vI}_{2,f,e}|^2= \frac{1}{|Z_e|^2} \left(\tilde{c}_{2,ii} + \tilde{c}_{2,kk} - 2\tilde{c}_{2,ik}\right)+2\tilde{G}_{c,e}\tilde{p}_{2,f,e} \nonumber \\
&\qquad \qquad - 2\tilde{B}_{c,e}\tilde{q}_{2,f,e} -\left|\boldsymbol{Y}_{c,e}\right|^{2} \tilde{c}_{2,ii},\\
&\tilde{p}_{2,f,e}+\tilde{q}_{2,f,e}=\tilde{Z}_{R,e}\big(|\tilde{\vI}_{2,f,e}|^2-2(\tilde{G}_{c,e}\tilde{p}_{2,f,e} -\tilde{B}_{c,e}\tilde{q}_{2,f,e})\nonumber\\  
&\qquad \quad  +\left|\boldsymbol{Y}_{c,e}\right|^{2} \tilde{c}_{2,ii}\big)+\tilde{G}_{c,e} (\tilde{c}_{2,ii} + \tilde{c}_{2,kk}), \\
&\tilde{p}_{2,f,e}+\tilde{q}_{2,f,e}=\tilde{Z}_{I,e}\big(|\tilde{\vI}_{2,f,e}|^2-2(\tilde{G}_{c,e}\tilde{p}_{2,f,e} - \tilde{B}_{c,e}\tilde{q}_{2,f,e}) \nonumber\\
&\qquad \qquad +\left|\boldsymbol{Y}_{c,e}\right|^{2} \tilde{c}_{2,ii}\big)-\tilde{B}_{c,e} (\tilde{c}_{2,ii} + \tilde{c}_{2,kk}), \\
&\lrbrackets{ 1+2 \tilde{Z}_{R,e}\tilde{G}_{c,e} - 2 \tilde{Z}_{I,e}\tilde{B}_{c,e} } \tilde{c}_{2,ii} - \tilde{c}_{2,kk} = 2(\tilde{Z}_{R,e}\tilde{p}_{2,f,e}   \nonumber\\
&+\tilde{Z}_{I,e}\tilde{q}_{2,f,e})-|\tilde{Z}_e|^2\big(|\tilde{\vI}_{2,f,e}|^2-2(\tilde{G}_{c,e}\tilde{p}_{2,f,e} - \tilde{B}_{c,e}\tilde{q}_{2,f,e})  \nonumber \\ &\qquad \qquad \qquad +\left|\tilde{Y}_{c,e}\right|^{2} \tilde{c}_{2,ii}\big)
\end{align}
\end{subequations}
All equations in \eqref{eq: ac_details_equiv_I_2} should hold simultaneously.

Similarly, the decision variables we will focus on in \eqref{ac: t3_acopf} are $\tilde{c}_{3,ii}, \forall i\in V, \tilde{3}_{2,ik}, \tilde{s}_{3,ik}, \forall e=(i,k)\in E, e= (k,i)\in E$, $\tilde{\vtheta}_{3}, \tilde{\vv}_3$ and $|\tilde{\vI}_{3,f}|^2,|\tilde{\vI}_{3,t}|^2$. Then, the constraints \eqref{ac: t3_acopf} are similar to \eqref{eq: details_ac_t2_fdi} and \eqref{eq: ac_details_equiv_I_2}, with \eqref{eq: ac_no_attack_gen}-\eqref{eq: ac_limit_I_t2_fdi} changed into
\begin{subequations}
\allowdisplaybreaks
\begin{align}
&\vp_{g,min} \le  \vLambda_g \tilde{\vp}_3 \le \vp_{g,max}, \vq_{g,min} \le  \vLambda_g \tilde{\vq}_3 \le \vq_{g,max}, \\
&\vLambda_d(\tilde{\vp}_{3,i} - \tilde{\vp}_{2,i}) = 0, \quad \vLambda_d(\tilde{\vq}_{3,i} - \tilde{\vq}_{2,i}) = 0, \\
&\vLambda_g(\tilde{\vp}_{3,i} - \tilde{\vp}_{2,i}) = 0\\
&\vV_{min} \le \tilde{\vv}_3 \le \vV_{max},\theta_{min,e} \le \tilde{\theta}_{3,e} \le  \theta_{max,e}, \forall e\in E, \\
&|\tilde{\vI}_{3,f}| \le \vI_{max}, |\tilde{\vI}_{3,t}| \le \vI_{max}.
\end{align}
\end{subequations}

Following \cite{yang2018general}, the decision variables in \eqref{ac: pf_linearze_t3} are $\hat{v}_{3,i}^2, \hat{\theta}_{3,i}, \hat{p}_{3,i}, \hat{q}_{3,i}, \forall i\in V$, $\hat{\vp}_{f,3}\in \mathbb{R}^{|E|}, \hat{\vq}_{f,3}\in \mathbb{R}^{|E|}$ and $|\hat{\vI}_3|^2\in \mathbb{R}^{|E|}$. Next, we define $p_{f,3,e}^{L}$ and $q_{f,3,e}^{L}$ for $e=(i,k)\in E$ with $a_{p,e} = 0$ as follows:
\begin{subequations}
\begin{align}
p_{f,3,e}^{L} &= \tilde{G}_{L,e}\bigg( \hat{\theta}_{ik,0}\hat{\theta}_{3,ik} - \frac{\hat{\theta}_{ik,0}^2}{2} + \frac{\hat{v}_{i,0}- \hat{v}_{k,0}}{\hat{v}_{i,0}+ \hat{v}_{k,0}}(\hat{v}_{3,i}^2- \hat{v}_{3,k}^2) \nonumber \\
&\qquad \qquad \quad- \frac{(\hat{v}_{i,0}- \hat{v}_{k,0})^2}{2} \bigg) + \mathcal{R}(\tilde{Y}_{c,e}) \hat{v}_{3,i}^2\\
q_{f,3,e}^{L} &= -\tilde{B}_{L,e}\bigg( \hat{\theta}_{ik,0}\hat{\theta}_{3,ik} - \frac{\hat{\theta}_{ik,0}^2}{2} + \frac{\hat{v}_{i,0}- \hat{v}_{k,0}}{\hat{v}_{i,0}+ \hat{v}_{k,0}}(\hat{v}_{3,i}^2- \hat{v}_{3,k}^2) \nonumber \\
&\qquad \qquad \quad- \frac{(\hat{v}_{i,0}- \hat{v}_{k,0})^2}{2} \bigg) - \mathcal{I}(\tilde{Y}_{c,e}) \hat{v}_{3,i}^2,
\end{align}
\end{subequations}
where $\hat{v}_{ik,0}$ and $\hat{\theta}_{ik,0}$ are obtained from any base case system operating condition. In our work, we set it as $\hat{v}_{ik,0} = v_{2,ik}$ and $\hat{\theta}_{ik,0}=\theta_{2,ik}$ for each given $(\va_p,e_t)$. Then, we have three types of constraints in \eqref{ac: pf_linearze_t3}. Specifically, by appropriately setting $\eta_{3,p,i}$ and $\eta_{3,q,i}$ (see proof of Theorem~\ref{theorem: ac_thr_set_linear_guarantee} for details) to tolerate the approximation error, for each $i\in V$, we have
\begin{align}
-\eta_{3,p,i} \le \vD_{i}\hat{\vp}_{3,f} + \hat{v}_{3,i}^2 \sum_{k=1}^{|V|} \tilde{G}_{ik} - p_{0,i} \le \eta_{3,p,i}. \label{eq: ac_t3_meet_load_p}
\end{align}
For each $i\in V_d$, we have
\begin{align}
-\eta_{3,q,i} \le \vD_{i}\hat{\vq}_{3,f} - \hat{v}_{3,i}^2 \sum_{k=1}^{|V|} \tilde{B}_{ik} - \tilde{q}_{3,i} \le \eta_{3,q,i}, \label{eq: ac_t3_meet_load_q}.
\end{align}
For each $e=(i,k)\in E$ with $a_{p,e} = 0$, we have
\begin{subequations}
\allowdisplaybreaks
\begin{align}
&p_{f,3,e}=\tilde{G}_{L,e} \frac{\hat{v}_{3,i}^{2}-\hat{v}_{3,k}^{2}}{2}-\tilde{B}_{L,e} \hat{\theta}_{ik}+p_{f,3,e}^{L},\label{eq: ac_t3_linear_pf}\\
&q_{f,3,e}=-\tilde{B}_{L,e} \frac{\hat{v}_{3,i}^{2}-\hat{v}_{3,k}^{2}}{2}-\tilde{G}_{L,e} \hat{\theta}_{ik}+q_{f,3,e}^{L},\label{eq: ac_t3_linear_qf}\\
&\lrbrackets{ 1+2 \tilde{Z}_{R,e}\tilde{G}_{c,e} - 2 \tilde{Z}_{I,e}\tilde{B}_{c,e} } \hat{v}_{3,i}^2 - \hat{v}_{3,k}^2 = 2(\tilde{Z}_{R,e}\hat{p}_{3,f,e} \nonumber \\
&\qquad \qquad +\tilde{Z}_{I,e}\hat{q}_{3,f,e}) -|\tilde{Z}_e|^2\big(|\hat{\vI}_{3,f,e}|^2-2(\tilde{G}_{c,e}\hat{p}_{3,f,e} \nonumber \\
&\qquad \qquad \qquad- \tilde{B}_{c,e}\hat{q}_{3,f,e})+\left|\tilde{Y}_{c,e}\right|^{2} \hat{v}_{3,i}^2 \big)
\end{align}\label{eq: ac_t3_linear_If}
\end{subequations}

\section{Details of PMU Locations Obtained in PPOP}\label{appendix: pmu_locations}
{
Here, we present the location of PMUs obtained in the proposed PPOP. First, in Table~\ref{tab: loc_pmu_arbitrary}, we give the PMU locations according to the best proposed solution $\Omega_{\mbox{\tiny PPOP}}$ to PPOP, which is consistent with Table~\ref{tab: res_matpower_pmu}.
\begin{table}[htbp]
\small
\vspace{-1em}
\centering
\caption{PMU Locations of PPOP under DC Model}
\begin{tabular}{|c| c |}
\hline
                   & Location of PMUs \\
                   \hline
IEEE 30-bus system        & 15, 23 \\
\hline
IEEE 57-bus system &   12,13,25 \\
\hline
IEEE 118-bus system  &   17,34,37,42,49,72,85,100,118 \\
\hline
IEEE 300-bus system & \makecell[c]{8,20,22,34,38,43,44,48,49,54,64,68, \\74,77,79,89,90,94,99,109,119,132, \\138,152,185,190,203,216,221,270,271}  \\
\hline
\end{tabular}
\vspace{-.5em}
\label{tab: loc_pmu_arbitrary}
\end{table}

Then, in Table~\ref{tab: ac_pmu_constrained}, we present the PMU locations of the solution that can pass the test of Alg.~\ref{alg: check_AC} under AC power flow model, obtained by Alg.~\ref{alg: ac_augment_pmu_dc}. 
\begin{table}[htbp]
\small
\vspace{-1em}
\centering
\caption{PMU Locations of PPOP under AC Model}
\begin{tabular}{|c| c |}
\hline
                   & Location of PMUs \\
                   \hline
IEEE 30-bus system        & 5,15,23 \\
\hline
IEEE 57-bus system &   12,13,25 \\
\hline
IEEE 118-bus system  &  17,34,37,42,49,62,72,85,100,118 \\
\hline
IEEE 300-bus system & \makecell[c]{8,20,22,34,38,43,44,48,49,54,64,68,\\
74,77,79,81,89,90,94,99,109,119,132,\\
138,152,175,185,190,197,203,216,\\
221,270,271}  \\
\hline
\end{tabular}
\vspace{-.5em}
\label{tab: ac_pmu_constrained}
\end{table}


}

\section{More Discussion of Related Works}\label{append: related_works_more}
{
\textbf{Different defense techniques against CCPA/FDI: }Following \cite{dibaji2019systems}, we classify defense techniques against CCPAs into the following categories:
\begin{enumerate}
    \item \textbf{Prevention:} Due to the requirements of network information and measurements, prevention methods defend against CCPAs by reducing or postponing the information leakage. Moving target defense (MTD) approach \cite{lakshminarayana2021moving} is a typical technique in this category. Specifically, MTD methods will strategically impose random change to network components (such as line admittance) to mislead the attacker. The CCPAs with falsified network parameters have a higher chance to be detected. Another typical method in this category is dynamic watermarking \cite{porter2020detecting}, which shares a similar spirit of MTD. 
    
    \item \textbf{Detection:} The methods in this category manage to detect the existence of attacks under some assumption on information exposure and attack capability. Traditional BDD is one of the approaches in this category. There are some advanced detection techniques, such as low rank-based detection \cite{liu2014detecting}. Securing measurements or deploying PMUs can also be used for detection. Specifically, an attack that tries to alter the measurements secured by PMUs will be detected by the control center. However, to achieve full detection, full observability by PMUs is required.
    
    \item \textbf{Resilience:} It is critical to keep the system stable when there exist CCPAs that can bypass the detection. In other words, resilience approaches aim at limiting the impact of the attacks. Game-theoretic methods can be regarded as typical ones, such as the budget-constrained formulations in \cite{wu2016efficient, xiang2017game, yao2007trilevel, yuan2016robust, tian2019multilevel}. Our solution lies in this category.
\end{enumerate}
}

\section{Additional Proofs}\label{appendix: additional_proof}
\begin{proof}[Theorem~\ref{theo: NP-hardness}]
We will reduce the dominating set problem to PPOP. Given a graph $\mathcal{N}=(V,E)$, the dominating set problem aims to find a minimum set of vertices $V_1\in V$ such that $\forall u\in V\setminus V_1$, $u$ has at least one neighbor in $V_1$. The dominating set problem is known to be NP-hard. Notice that given the grid $\mathcal{N} = (V,E)$ the parameters for the proposed problem \eqref{eq:PMU_placement}-\eqref{eq:reformulate_attacker} are $\vp_0, \vGamma, \xi_p, \xi_c, \valpha$ and $\{\gamma_e\}_{e\in E}$. We will prove for any given connected grid and the associated dominating set problem, there exists a parameter setting for the proposed problem such that $V_1$ is a minimal dominating set if and only if $V_1$ is an optimal solution to \eqref{eq:PMU_placement}, i.e., $\forall u \in V, x_{N,u} = 1$.

Given any $\vp_0$, suppose that $\vtheta_0$ 
is the associated phase angle without attack, i.e., $\vp_0 = \tilde{\vB}\vtheta_0$, and $\hat{\vtheta}_0$ is the the solution to \eqref{eq:reformulate_OPF}, i.e., $\hat{\vtheta}_0 = \psi_s(\vp_0,\tilde{\vD})$, which gives $\hat{\vf}_0 := \vGamma\tilde{\vD}^T\hat{\vtheta}_0$.

Then, we set $\vp_0 = \bm{0}$, $\xi_p = 0$, $\xi_c = \infty$, $\alpha = \infty$ and $\vGamma$ as identity matrix, which results in $\vtheta_0 = \hat{\vtheta}_0 = \bm{0}$ and $\hat{\vf}_0 = \bm{0}$. In addition, we set $\gamma_e = 0, \forall e\in E$, which transform \eqref{eq: attack_overflow} to
\begin{align}
    |\Gamma_e\vd_e^T\vtheta_3| = 0 \leftrightarrow \pi_e = 0.
\end{align}
Next, we show by contradiction that $|\Gamma_e\vd_e^T\vtheta_3| = 0$ holds for all $e\in E$ only if $\tilde{\vtheta}_2 = \bm{0} = \vtheta_0$. Suppose $\tilde{\vtheta}_2 \ne \bm{0}$, we must have $\tilde{\vB}\tilde{\vtheta}_2 \ne \bm{0}$, which leads to $\bm{0}\ne \tilde{\vtheta}_3 = \psi_s(\tilde{\vB}\tilde{\vtheta}_2,\tilde{\vD})$ and thus $\vLambda_g\tilde{\vB}\tilde{\vtheta}_3 \ne \bm{0}$ due to the constraint \eqref{eq: zero_ref_theta}. The non-zero $\vLambda_g\tilde{\vB}\tilde{\vtheta}_3$ implies that $\exists e\in E$ such that $\Gamma_e\vd_e^T\vtheta_3 \ne 0$. That is to say, the constraint \eqref{eq: PMU_no_overflow} holds only when $\tilde{\vtheta}_2 = \vtheta_0 = \bm{0}$, which indicates{ that the defender has to place PMUs to guarantee that the only feasible solution to \eqref{eq:reformulate_attacker} is $\va_c = \bm{0}$. In another word, $\vbeta$ needs to satisfy $\forall u \in V, x_{N,u} = 1$, which completes the proof.}
\end{proof}

\begin{proof}[Theorem~\ref{thm:Optimality of alternating optimization}]
First, we introduce some definitions:  $\mathcal{B} := \{ \vbeta| \psi_a(\vbeta) = 0\}$ denotes the set of feasible solutions, $\mathcal{B}^c := \{ \vbeta| \psi_a(\vbeta) \ge 1 \}$ the infeasible solutions, $\mathcal{M}(\mathcal{B}^c) := \{ \vbeta| (\vbeta, \vbeta' \in \mathcal{B}^c) \wedge (\vbeta'\ge \vbeta) \rightarrow (\vbeta'=\vbeta) \}$ the maximal infeasible solutions, and $\mathcal{\vP}:= \{ \check{\vbeta} \in [0,1]^{|V|}| \forall \vbeta\in \mathcal{M}(\mathcal{B}^c):\sum_{u:\beta_u=0}\check{\beta}_u \ge 1\}$ the polytope excluding all the maximal infeasible solutions. 

Then, based on the results in \cite{vielma2015mixed}, we have the following characterization:
\begin{lemma}\label{lem: facet}
The following statements hold: (i) $\mathcal{\vP}\cap \{0,1\}^{|V|} = \mathcal{B}$; (ii) $\forall \vbeta'\in \mathcal{M}(\mathcal{B}^c)$, $\sum_{u:\beta_u'=0}\beta_u \ge 1$ defines a facet of $\mathcal{\vP}$. 
\end{lemma}
\begin{proof}
To prove statement (i), we first prove that $\mathcal{B}\subseteq (\mathcal{P}\cap \{0,1\}^{|V|})$ by contradiction. Suppose $\exists \vbeta_1\in \mathcal{B}$ but $\vbeta_1\notin \mathcal{P}$. Then by definition of $\mathcal{P}$, there must exist $\vbeta_1'\in \mathcal{B}^c$ such that $\sum_{u:\beta_{1,u}'=0}\beta_{1,u} = 0$, which implies $\Omega(\vbeta_1) \subseteq \Omega(\vbeta_1')$. By Lemma~\ref{lem: ob1}, we must have $\vbeta_1\in \mathcal{B}^c$, which contradicts with the assumption that $\vbeta_1\in \mathcal{B}$. Thus, $\mathcal{B}\subseteq (\mathcal{P}\cap \{0,1\}^{|V|})$. Then, we prove $(\mathcal{P}\cap \{0,1\}^{|V|}) \subseteq \mathcal{B}$ by contradiction. Suppose there exists $\check{\vbeta} \in (\mathcal{P}\cap \{0,1\}^{|V|})$ but $\check{\vbeta} \notin \mathcal{B}$, which implies that $\check{\vbeta} \in \mathcal{B}^c$. That is to say, $\exists \check{\vbeta}' \ge \check{\vbeta}$ such that $\check{\vbeta}'\in \mathcal{M}(\mathcal{B}^c)$. Then by definition of $\mathcal{P}$, we have $\sum_{u:\check{\beta}'_{u}=0} \check{\beta}_u \ge 1$. However, since $\check{\vbeta}'\ge \check{\vbeta}$, $\forall u:\check{\beta}_u = 0$, we must have $\check{\beta}_u = 0$ and leads to $\sum_{u:\check{\beta}'_{u}=0} \check{\beta}_u = 0$, which introduces contradiction. In summary, $\mathcal{\vP}\cap \{0,1\}^{|V|} = \mathcal{B}$.

We then prove statement (ii) by contradiction, i.e., $\exists \check{\vbeta}' \in \mathcal{M}(\mathcal{B}^c)$ such that when we remove the inequality $\sum_{u:\check{\beta}'_{u} = 0} \beta_u \ge 1$ from $\mathcal{P}$, we still have $\mathcal{P}$. By definition of $\mathcal{M}(\mathcal{B}^c)$, we must have $\check{\vbeta}' \in \mathcal{B}^c$, which implies $\sum_{u:\check{\beta}'_{u} = 0} \check{\beta}_u' = 0$, i.e., $\check{\vbeta}' \notin \mathcal{P}$. That is to say, there exists some inequality to cut $\check{\vbeta}'$ out from $\mathcal{P}$, i.e., $\exists\check{\vbeta}'' \in \mathcal{M}(\mathcal{B}^c)$ and $\check{\vbeta}'' \neq \check{\vbeta}'$ such that $\sum_{u:\check{\beta}''_{u} = 0} \check{\beta}_u' = 0$. Notice that $\forall u: (\check{\beta}_u'' =0) \rightarrow (\check{\beta}_u' =0)$, which implies $\Omega(\check{\vbeta}')\subseteq \Omega(\check{\vbeta}'')$. By definition of $\mathcal{M}(\mathcal{B}^c)$, we must have $\check{\beta}_u'' = \check{\beta}_u'$, which contradicts with $\check{\beta}_u'' \neq \check{\beta}_u'$ and completes the proof.
\end{proof}


We now prove Theorem~\ref{thm:Optimality of alternating optimization} based on Lemma~\ref{lem: facet}. 
%
First notice that each $\hat{\vbeta} \in \mathcal{B}^c$ will be enumerated at most once in Alg.~\ref{alg: alter_opt_v1} due to the ``no-good'' constraints, and hence the algorithm will converge in finite time. 
Then, consider an arbitrary $\hat{\vbeta}'$ obtained through \eqref{eq: opt_max_nogood_but}. The generated ``no-good'' constraint $\sum_{i:\hat{\beta}'_i=0}\beta_i \ge 1$ must be satisfied by all the feasible solutions in $\mathcal{B}$, as any PMU placement violating this constraint must be infeasible according to Lemma~\ref{lem: ob1}. Finally, for any $\vbeta_1,\vbeta_2\in\mathcal{B}$ with $\lp{\vbeta_1}_0 < \lp{\vbeta_2}_0$, $\vbeta_1$ will be found by Alg.~\ref{alg: alter_opt_v1} before $\vbeta_2$, since each guess of PMU placement is obtained by minimizing $\|\vbeta\|_0$ in \eqref{eq:PMU_placement}, which completes the proof.
\end{proof}

\begin{proof}[Theorem~\ref{thm:Optimality of AODC}]
As Alg.~\ref{alg: alter_opt_v1} always returns a feasible solution that defends against all attack pairs, we only need to prove that the solution $\vbeta_1$ returned by AODC requires the minimum number of PMUs. 
We will prove this by contradiction. 
Suppose that there exists $\vbeta_2$ such that $\lp{\vbeta_2}_0 < \lp{\vbeta_1}_0$ and $\psi_a(\vbeta_2) = 0$. Then $\vbeta_2$ must be feasible to 
the instance of \eqref{eq: vcg_formulation} constructed based on the attack pairs $\{(\va_p^{(k)},e^{(k)})\}_{k=1}^K$ and the maximal infeasible solutions $\{\hat{\vbeta}^{'(k)}\}_{k=1}^K$ found by AODC as it defends against all attacks. This contradicts with the fact that $\vbeta_1$ is optimal to \eqref{eq: vcg_formulation}.\looseness=-1
\end{proof}

\begin{proof}[Lemma~\ref{lem: lp_relax_attack_denial}]
We first observe that $\vx_N$ and $\vx_L$ are unique under the constraints \eqref{eq: const_beta_x_N}-\eqref{eq: const_beta_x_L}. Thus, we will use $\vx_N(\vbeta)$ and $\vx_L(\vbeta)$ to denote the values of $\vx_N$ and $\vx_L$ satisfying \eqref{eq: const_beta_x_N}-\eqref{eq: const_beta_x_L} for a given $\vbeta \in \{0,1\}^{|V|}$.

For a given attack pair $(\va_p,e)$, $(\check{\vq}_1, \check{\vq}_2, \check{\vbeta})$ can be feasible to \eqref{eq: attack_denial_lp} in two different cases. The first case is that 
\begin{align}
\sum_{a_{p,e} = 1} x_{L,e}(\ceil{\check{\vbeta}}) \ge 1,
\end{align}
which makes $(\vq_1= \bm{0}, \vq_2= \bm{0}, \ceil{\check{\vbeta}}, \vx_N(\ceil{\check{\vbeta}}), \vx_L(\ceil{\check{\vbeta}}))$ feasible for \eqref{eq: attack_denial_constr} with $w_{a} = 1$.

The second case is that $x_{L,e}(\ceil{\check{\vbeta}}) = 0$ for all $e$ with $a_{p,e} = 1$, in which case we must have $(\check{\vq}_1, \check{\vq}_2, \ceil{\check{\vbeta}}, \vx_N(\ceil{\check{\vbeta}}), \vx_L(\ceil{\check{\vbeta}}))$ feasible to \eqref{eq: attack_denial_constr} with $w_a = 0$. To prove this, we only need to show that 
\begin{align}
\lrbrackets{\vF_3\vx_N(\ceil{\check{\vbeta}})}^T \check{\vq}_2 \le \vF_3 \check{\vq}_2.
\end{align}
According to \eqref{eq: expand_primal_ineq}, $F_{3,i,u}$ is either $0$ or $-M_{\theta}$, which together with the fact that $\vx_{N,u}(\ceil{\check{\vbeta}}) \ge 0$ and $\check{q}_{2,i} \ge 0$ implies that 
\begin{align}
\lrbrackets{\vF_3\vx_N(\ceil{\check{\vbeta}})}^T \check{\vq}_2 &= \sum_{u\in V} \vx_{N,u}(\ceil{\check{\vbeta}}) \left( \sum_{i=1}^{m_2} F_{3,i,u} \check{q}_{2,i} \right) \label{eq: appendix_51}\\
& \le \sum_{u\in V} \bm{1} \left( \sum_{i=1}^{m_2} F_{3,i,u} \check{q}_{2,i} \right) = \vF_3 \check{\vq}_2,
\end{align}
which completes the proof. 
\end{proof}

\begin{proof}[Theorem~\ref{theo: polynomial_time_heuristic}]
Under the assumption of $\xi_p = \mathcal{O}(1)$, the number of possible attack pairs is $|E|\left(\sum_{i=1}^{\xi_p}\binom{|E|}{i}\right) \le \xi_p |E|^{\xi_p+1} = \mathcal{O}\lrbrackets{|E|^{\xi_p+1}}$. Therefore, the time complexity of solving \eqref{eq:reformulate_attacker} for a given $\vbeta$ is polynomial in $|E|$ and $|V|$, since in the worst case \eqref{eq:reformulate_attacker} can be solved by checking the feasibility of \eqref{eq: LP_given_pair} for all the $\mathcal{O}\lrbrackets{|E|^{\xi_p+1}}$ attack pairs.  \looseness=-1

We first characterize the complexity of Alg.~\ref{alg: primitive_find_candidate}. {Since each candidate placement $\Omega_i$ either has one more node or can defend against all attack pairs in $\mathcal{A}$ after one iteration of the while loop, Alg.~\ref{alg: primitive_find_candidate} converges within $|V|$ iterations. Each iteration of Alg.~\ref{alg: primitive_find_candidate} is dominated by solving \eqref{eq: vcg_LP_relax} (Line~\ref{line: alg3_sol_LP}) for at most $K_c$ times. Since the numbers of variables and constraints of \eqref{eq: vcg_LP_relax} are both $\mathcal{O}((|E|+|V|)|\mathcal{A}|)$ and $|\mathcal{A}|=\mathcal{O}\lrbrackets{|E|^{\xi_p+1}}$, the complexity of solving \eqref{eq: vcg_LP_relax} is polynomial \footnote{The exact order depends on the specific algorithm used to solve LP \cite{terlaky2013interior}.} in $|V|$ and $|E|$. In summary, the complexity of Alg.~\ref{alg: primitive_find_candidate} is polynomial in $|V|,\: |E|$, and $K_c$ since it solves a polynomial-sized LP for at most $K_c|V|$ times.} It is worth noting that the effect of $K_A$ and $K_L$ in Alg.~\ref{alg: primitive_find_candidate}'s complexity is dominated by $|V|$ and $|E|$. To see this, we note that $K_L$ only appears in Line~\ref{line: alg3_candidate_phy_attack} of Alg.~\ref{alg: primitive_find_candidate}, in which we must have $K_L \le |E|$. Then, $K_A$ only appears in Line~\ref{line: alg3_up_Q} of Alg.~\ref{alg: primitive_find_candidate}, in which we must have $K_A \le |V|$. Thus, we do not consider the effect of $K_A$ and $K_L$ in Alg.~\ref{alg: primitive_find_candidate}'s computational complexity.

The complexity of Alg.~\ref{alg: heuristic_3_phase} comes from: (\romannumeral1) solving \eqref{eq:reformulate_attacker} $\mathcal{O}(|E|^{\xi_p+1})$ times (Line~\ref{Line: ph1_atttacker_v1} and Line~\ref{Line: attack_generate_phase3}); (\romannumeral2) solving \eqref{eq: vcg_LP_relax} for $|\mathcal{A}_0|=\mathcal{O}(|E|^{\xi_p+1})$ times (Line~\ref{Line: ph1_vcg_lp}), {each of which deals with an LP containing $\mathcal{O}((|E|+|V|)|\mathcal{A}_0|)$ variables and constraints and thus takes polynomial time}; (\romannumeral3) calling Alg.~\ref{alg: primitive_find_candidate} at Line~\ref{Line: init_candidates_lp} for $1$ time and at Line~\ref{Line: uc_A3} for $\mathcal{O}(|E|^{\xi_p+1})$ times, whose complexity is polynomial in $|V|,\: |E|$, and $K_c$. In summary, Alg.~\ref{alg: heuristic_3_phase} is a polynomial-time algorithm in terms of $|V|$, $|E|$, and $K_c$.
\end{proof}

{
\begin{proof}[Theorem~\ref{theorem: ac_thr_set_linear_guarantee}]
According to \cite{yang2018general, coffrin2015distflow} and \eqref{eq: ac_t3_linear_If}, we have
\begin{align}\label{appendix: I_hat_3}
|\hat{I}_{3,f,e}|^2 &= \frac{1}{|Z_e|^2}\big( 2(\tilde{Z}_{R,e}\hat{p}_{3,f,e}+\tilde{Z}_{I,e}\hat{q}_{3,f,e}) + \hat{v}_{k}^2 - \nonumber\\
&(1+2\tilde{Z}_{R,e}\tilde{G}_{c,e} - 2\tilde{Z}_{I,e}\tilde{B}_{c,e})\hat{v}_{i}^2\big) + 2(\tilde{G}_{c,e}\hat{p}_{3,f,e}- \nonumber \\
&\tilde{B}_{c,e}\hat{q}_{3,f,e}) - |\tilde{Y}_{c,e}|^2\hat{v}_{i}^2
\end{align}
for each $e = (i,k)\in E$ with $a_{p,e} = 0$. Based on \eqref{appendix: I_hat_3} and the assumption on $\vepsilon_{\theta} = (\epsilon_{\theta,u})_{u\in V}$, $\vepsilon_{\vv} = (\epsilon_{v,u})_{u\in V}$, $\vepsilon_{\vp} = (\epsilon_{p,e})_{e\in E}$ and $\vepsilon_{\vq} = (\epsilon_{q,e})_{e\in E}$, we can easily derive an upperbound $\epsilon_{I,e}\ge ||\hat{I}_{3,e}| - |I_{3,e}| |, \forall e\in E$. Specifically, we can set
\begin{align}
\epsilon_{I,e}&:=\frac{1}{|Z_e|^2}\big( 2(|\tilde{Z}_{R,e}|\epsilon_{p,e}+|\tilde{Z}_{I,e}|\epsilon_{q,e}) + \epsilon_{v,i}^2 + \nonumber\\
&|1+2\tilde{Z}_{R,e}\tilde{G}_{c,e} + 2\tilde{Z}_{I,e}\tilde{B}_{c,e}|\epsilon_{v,i}^2\big) + 2(|\tilde{G}_{c,e}|\epsilon_{p,e}+ \nonumber \\
&|\tilde{B}_{c,e}|\epsilon_{q,e}) + |\tilde{Y}_{c,e}|^2\epsilon_{v,i}^2.
\end{align}

If there exists an successful attack pair $(\va_p,e)$ that cannot be found by Alg.~\ref{alg: check_AC} for a given PMU placement, we must have one of the following cases:
\begin{enumerate}
    \item There exists $\tilde{\vv}_2, \tilde{\vtheta}_2$ such that $|I_{3,e}| > \gamma_e I_{max,e}$. In the meantime, at least one of \eqref{eq: ac_t3_meet_load_p} and \eqref{eq: ac_t3_meet_load_q} are violated.
    \item Let $|\hat{I}_{3,f,e}^*|$ be the optimal solution of \eqref{eq: ac_formulate_overload}. There exists $\tilde{\vv}_2^{(1)}, \tilde{\vtheta}_2^{(1)}, \tilde{\vv}_3^{(1)}, \tilde{\vtheta}_3^{(1)}$ such that $|I_{3,e}^{(1)}| > \gamma_e I_{max,e}$. Let $|\hat{I}_{3,f,e}^{(1)}|$ be the corresponding approximated solution for $\tilde{\vv}_2^{(1)}, \tilde{\vtheta}_2^{(1)}, \tilde{\vv}_3^{(1)}, \tilde{\vtheta}_3^{(1)}$. Then we must have $\hat{I}_{max,e} \ge |\hat{I}_{3,f,e}^*|\ge |\hat{I}_{3,f,e}^{(1)}|$.
\end{enumerate}
We first show that the case one can be avoided if we properly set $\eta_{3,p,i}$ in \eqref{eq: ac_t3_meet_load_p} and $\eta_{3,q,i}$ in \eqref{eq: ac_t3_meet_load_q}. Specifically, according to \eqref{eq: ac_t3_meet_load_p}, we must have 
\begin{align}
\vD_{i}\hat{\vp}_{3,f} + \hat{v}_{3,i}^2 \sum_{k=1}^{|V|} \tilde{G}_{ik} - p_{0,i} \le \eta_{3,p,i} 
\end{align}
if we set 
\begin{align}
\eta_{3,p,i} \ge (\Delta_{ii}-1)\epsilon_{p,i} + |\sum_{k=1}^{|V|} \tilde{G}_{ik}| \epsilon_{v,i},
\end{align}
where $(\Delta_{ii}-1)$ denotes the number of neighbors of node $i$ as defined in \eqref{eq: const_beta_x_N}. Similarly, we can define $\eta_{3,q,i}$ to avoid the first case. Then, we will show how to set $\hat{I}_{max,e}$ so that the second case will not happen. In case two, we must have 
\begin{align}\label{appendix: derive_I_hat_max}
\hat{I}_{max,e} \ge |\hat{I}_{3,f,e}^*|\ge |\hat{I}_{3,f,e}^{(1)}| \ge |I_{3,e}^{(1)}|  - \epsilon_{I,e} > \gamma_e I_{max,e} -\epsilon_{I,e}
\end{align}
Thus, if we set $\hat{I}_{max,e} \le \gamma_e I_{max,e} -\epsilon_{I,e}$, \eqref{appendix: derive_I_hat_max} cannot hold, which rules out the possibility of case two. In summary, by properly setting $\eta_{3,p,i}, \eta_{3,q,i}$ and set $\hat{I}_{max,e} \le \gamma_e I_{max,e} -\epsilon_{I,e}$, a PMU placement that can pass the test of Alg.~\ref{alg: check_AC} will achieve our defense goal. 
\end{proof}
}

\end{appendices}
\fi

\end{document}